\documentclass[fleqn,12pt,letterpaper]{article}

\usepackage{amsfonts,amssymb,amsmath,amsthm,color,float,graphicx,natbib,rotating}
\usepackage{epstopdf}
\usepackage[margin=2.54cm]{geometry}
\definecolor{darkred}{rgb}{0.5,0.2,0.2}
\usepackage[hyphens]{url}
\usepackage[allcolors=darkred,bookmarks=false,colorlinks=true]{hyperref}
\usepackage{booktabs}

\theoremstyle{plain}
\newtheorem{theorem}{Theorem}[section]
\newtheorem{assumption}{Assumption}[section]

\newtheorem{lemma}{Lemma}[section]

\newtheorem{definition}{Definition}
\theoremstyle{definition}

\newtheorem{remark}{Remark}

\newcommand{\R}{\mathbb{R}}

\numberwithin{equation}{section}
\makeatletter
\newcommand*\rel@kern[1]{\kern#1\dimexpr\macc@kerna}
\newcommand*\widebar[1]{%
  \begingroup
  \def\mathaccent##1##2{%
    \rel@kern{0.8}%
    \overline{\rel@kern{-0.8}\macc@nucleus\rel@kern{0.2}}%
    \rel@kern{-0.2}%
  }%
  \macc@depth\@ne
  \let\math@bgroup\@empty \let\math@egroup\macc@set@skewchar
  \mathsurround\z@ \frozen@everymath{\mathgroup\macc@group\relax}%
  \macc@set@skewchar\relax
  \let\mathaccentV\macc@nested@a
  \macc@nested@a\relax111{#1}%
  \endgroup
}
\makeatother

\begin{document}

\title{To be or not to be:\\ Roughness or long memory in volatility?\footnote{We are grateful to the co-editor Viktor Todorov, an associate editor, and two anonymous referees for their constructive and insightful feedback that were helpful in revising the paper. We thank Morten Ørregaard Nielsen for his inspiration on likelihood-based estimation of fractionally integrated processes. We benefitted from comments by and discussions with Bezirgen Veliyev, Chen Zhang, Jun Yu, Mikko Pakkanen, Peter Hansen, Shuping Shi, Tim Bollerslev, Uwe Hassler, plus seminar and conference participants at University of Macau, University of North Carolina at Chapel Hill, Aarhus Workshop in Econometrics (AWE) VI, the ReadingMetrics study group, the 2022 Aarhus University--Singapore Management University joint online workshop on Volatility, the 2022 Vienna--Copenhagen (VieCo) conference on Financial Econometrics in Copenhagen, Denmark, and the 2022 annual SoFiE conference in Cambridge, UK. Kim Christensen received funding from the Independent Research Fund Denmark (DFF 1028--00030B) to support his work. Mikkel Bennedsen is grateful for research support from the Aarhus Center for Econometrics (ACE) funded by the Danish National Research Foundation (DNRF186). Send correspondence to: \href{mailto:kim@econ.au.dk}{\nolinkurl{kim@econ.au.dk}}.}}

\author{Mikkel Bennedsen\thanks{Aarhus University, Department of Economics and Business Economics, Denmark.} $^,$\thanks{Aarhus Center for Econometrics (ACE), Aarhus University, Denmark.} $^,$\thanks{Center for Research in Energy: Economics and Markets (CoRE), Aarhus University, Denmark.}
\and Kim Christensen\footnotemark[2] $^,$\footnotemark[4] $^,$\thanks{Research fellow at the Danish Finance Institute (DFI).}
\and Peter K. Christensen\footnotemark[2] $^,$\footnotemark[4]}

\date{January, 2026}

\maketitle

\vspace*{-1.20cm}

\begin{abstract}
We develop a framework for composite likelihood estimation of parametric continuous-time stationary Gaussian processes. We derive the asymptotic theory of the associated maximum composite likelihood estimator. We implement our approach on a pair of models that have been proposed to describe the random log-spot variance of financial asset returns. A simulation study shows that it delivers good performance in these settings and improves upon a method-of-moments estimation. In an empirical investigation, we inspect the dynamic of an intraday measure of the spot log-realized variance computed with high-frequency data from the cryptocurrency market. The evidence supports a mechanism, where the short- and long-term correlation structure of stochastic volatility are decoupled in order to capture its properties at different time scales. This is further backed by an analysis of the associated spot log-trading volume.

\bigskip \noindent \textbf{JEL Classification}: C10; C58; C80.

\medskip \noindent \textbf{Keywords}: Composite likelihood; stationary Gaussian processes; long memory; realized variance; roughness; stochastic volatility.

\end{abstract}

\vfill

\thispagestyle{empty}

\pagebreak

\section{Introduction} \setcounter{page}{1}

The search for the best description of the dynamic of time-varying volatility continues to permeate financial economics. On the one hand, it has long been recognized that realized variance (RV) is highly persistent and exhibits an autocorrelation function (ACF) that perishes so slowly that it can be approximated by a stochastic process featuring long memory \citep[see, e.g.,][]{comte-renault:98a, andersen-bollerslev-diebold-labys:03a,  corsi:09a}. On the other hand, in recent years it has been suggested that the sample path of volatility is more vibrant than those implied by a standard Brownian motion (sBm), and that it may instead be generated by a fractional Brownian motion (fBm) with a Hurst exponent less than a half; the index of the sBm. \citet{gatheral-jaisson-rosenbaum:18a} observe that the shape of the implied volatility surface for short-term at-the-money options is consistent with stochastic volatility being governed by such a fBm. They reinforce this finding with variogram regressions for daily log-RV series from a number of leading stock indexes, which further support their claim that ``volatility is rough.'' Since then, a sequence of follow-up papers have studied various statistical estimation procedures for gauging the roughness of volatility, both under the physical and risk-neutral probability measure, largely confirming this view \citep[e.g.,][]{fukasawa-takabatake-westphal:22a, bolko-christensen-pakkanen-veliyev:23a, chong-todorov:23a, wang-xiao-yu:23a, shi-yu-zhang:24a}.\footnote{A more comprehensive overview of the extant literature is available at the Rough Volatility Network's website: \url{https://sites.google.com/site/roughvol/home/rough-volatility-literature}.}

The workhorse in the recent continuous-time literature on rough volatility is the fractional Ornstein-Uhlenbeck (fOU) process, where the sBm in the Gaussian Ornstein-Uhlenbeck process is replaced with an fBm as impetus. This model is parametric, so in principle we can do maximum likelihood estimation (MLE) of it. This idea is explored in a recent prominent contribution of \citet{wang-xiao-yu-zhang:25a}. In practice, MLE of the fOU process can be challenging, however. This is because the model is generally non-Markovian, which makes it difficult to estimate, because the process can depend on its entire history.\footnote{\citet{damian-frey:24a} study a particle filter technique to estimate roughness in spot volatility. It exploits the notion that the fBm has an infinite-dimensional Markovian representation approximable by a finite sum of Ornstein-Uhlenbeck processes driven by a single sBm. However, their approach centers around the nonstationary fBm, not the stationary fOU, and furthermore it cannot capture long memory.} Thus, Gaussian MLE can be prohibitive with a large sample of such a discretely observed log-spot variance processes---except in special cases---because the calculation of the determinant and inversion of the covariance matrix are computationally intensive, when the number of observations is moderate-to-large, e.g. in excess of 1,000.

To circumvent this issue, in this paper we propose a composite likelihood estimator for parametric stationary Gaussian processes, which nests the fOU as a special case. The composite likelihood approach---introduced in \citet{lindsay:88a}---reduces the complexity of MLE by only including lower-dimensional sub-models in the criterion function. This effectively allows it to operate on arbitrarily large samples with minimal extra computational cost. We derive the asymptotic theory for the maximum composite likelihood estimator (MCLE), which turns out to be somewhat intractable, because it features a discontinuity both in the rate of convergence and the shape of the limiting distribution, as a function of the memory of the process \citep[pursuant to][]{hosking:96a}. Inference can also depend on whether the location parameter is estimated or not. However, while MCLE loses some of the desirable statistical properties of MLE, in practice composite likelihood can be preferred when the full likelihood is impractical or difficult to evaluate. Moreover, in finite samples the loss of efficiency may not be too severe, since the composite likelihood function can be smoother than the full likelihood surface and, hence, much more convenient to navigate and optimize \citep[e.g.,][]{varin-reid-firth:11a}.

The theoretical properties of the composite likelihood estimator have been examined in various settings, see, e.g., \citet{cox-reid:04a}, \citet{davis-yau:11a}, and \citet{varin-reid-firth:11a}. It has also been applied to many fields, including finance \citep{bennedsen-lunde-shephard-veraart:23a, pakel-shephard-sheppard-engle:20a}. The most closely related paper to ours within the composite likelihood literature is \citet{davis-yau:11a}. They study the theoretical properties of a pairwise MCLE in a standard time series setting with discrete-time linear processes---not necessarily Gaussian---where the dependence structure is determined by the decay of the coefficients in the linear filter. By contrast, our processes evolve in continuous-time. Moreover, we develop our theory in the general setting of $q$-wise composite likelihood. Even in the pairwise setting, however, our framework extends \citet{davis-yau:11a} by allowing for the presence of a slowly varying function at infinity in the autocorrelation structure of the process. This is relevant outside the fOU model, such as the Brownian semistationary process with a power kernel \citep[e.g.,][]{barndorff-nielsen-schmiegel:09a}.

In the volatility literature, the state variable in the fOU represents the point-in-time log-variance of an asset log-price, which is latent. Previous work has employed the daily log-RV as an observable surrogate, but this is an estimator of the daily log-integrated variance, i.e. an infinite mixture of log-normal random variables, which is itself not log-normal. This loses the analytic tractability of Gaussian MLE and makes the calculation of the exact transition density a nontrivial exercise. Even worse, if it is not properly accounted for, it runs the risk of mixing theoretical properties of a spot process with empirical properties of an integrated process. Moreover, since roughness is a sample path property that concerns the behavior of a stochastic process at very short time scales, this discussion suggests that estimation of the fOU process perhaps ought to be based on more localized measure of log-variance. In this paper, we recover the log-spot variance at the intraday horizon. As a result, the sample size in our empirical application is so large that MCLE is possibly the only feasible likelihood-type procedure available.

A crucial weakness of the fOU process as a model for the random log-volatility of financial asset returns is that it controls both the short- and long-run persistence in a single parameter, namely the Hurst exponent. However, roughness is a sample path property of a stochastic process (leading to a rapid decline in the ACF at short time scales), whereas long memory is a property of the distribution function (leading to a slow decline in the ACF at long time scales). As these effects are intertwined in the fOU process, it is only capable of featuring \textit{either} roughness \textit{or} long memory. This shortcoming was, in fact, already highlighted by \citet{mandelbrot:82a} in the context of the fBm. Set against this backdrop, we explore a stationary Gaussian process from a so-called Cauchy class \citep{gneiting-schlather:04a}. As for the fOU process, the latter has two parameters to fit the dynamic dependence, but it reserves a separate parameter to describe the short- and long-term decay. Hence, it allows to decouple the behavior of the process at different time scales and is therefore able to account for \textit{both} roughness \textit{and} long memory.

In a simulation study, we examine the small sample properties of the MCLE for the fOU process and the Cauchy class, which we benchmark against a method-of-moments estimator (MME), e.g. \citet{bolko-christensen-pakkanen-veliyev:23a} and \citet{wang-xiao-yu:23a}.\footnote{\citet{fukasawa-takabatake-westphal:22a} and \citet{shi-yu-zhang:24a} develop Whittle-type approximate likelihood estimation of the fOU model.} The advantages of the MCLE over the MME are at least twofold. First, the MME for the Hurst parameter often relies on in-fill asymptotic theory to derive consistency, while our composite likelihood theory is derived within a long-span setting. The latter is appropriate, when the process is sampled discretely on a fixed equidistant partition over an expanding horizon. Second, MCLE extracts the entire parameter vector in a single step rather than via a two-stage approach, as done in the MME. The latter appears to lead to systematic biases in the parameter estimates. As we show, this means that MCLE is often more accurate than the MME, or at the very least heading in the right direction (toward the true parameter value). However, MCLE is not uniformly superior to MME, because it depends on the exact implementation.

In our empirical application, we inspect a vast high-frequency dataset from the cryptocurrency market. The coins we study are extremely liquid, so we can sample them at a much higher frequency than what is possible for more traditional assets that are traded less often. We compute an observable measure of the latent spot variance based on an intraday RV estimator. We implement the MCLE procedure on the log-series and estimate both the fOU process and the Cauchy class. The results are striking in that the fOU process, in agreement with recent work, suggests that the log-spot variance is rough. Conversely, the Cauchy class strongly points toward a dynamic, where both roughness and long memory are required to describe the time-varying volatility. This confirms the findings from equity high-frequency data in related work of \citet*{bennedsen-lunde-pakkanen:22a}. The main shortcoming with this application is that the RV measure is a noisy proxy of the spot variance and the inherent sampling error can distort the inferred dynamic of the variance process. Motivated by the acclaimed volume-volatility relationship, we also conduct an investigation of the trading volume that is directly observed. In the biggest estimation for log-trading volume, the MCLE is confronted with a sample of about 13 million observations, but it converges to an optimum within a few minutes. In sum, our results for log-trading volume largely confirm those from the log-RV. We conclude that \textit{both} roughness \textit{and} long memory are needed to describe stochastic volatility.

The rest of the paper is structured as follows. In Section \ref{section:mcle}, we introduce the class of stationary Gaussian processes and composite likelihood estimation. We also derive the asymptotic behavior of the MCLE. In Section \ref{section:example}, we introduce the fOU process and the Cauchy class that are the concrete parametric models we investigate more in-depth. Section \ref{section:simulation} presents a simulation study of our estimator within this framework and documents its efficacy in small samples, also compared to a MME. Section \ref{section:empirical} contains empirical work, where we implement the technique on a high-frequency time series of log-spot variance estimates and log-trading volume from a number of cryptocurrency spot exchange rates. In Section \ref{section:conclusion}, we conclude and point to directions for further work. An appendix contains mathematical derivations, presents a method-of-moments estimator, and includes supplemental simulation analysis and empirical results.

\section{Theoretical framework} \label{section:mcle}

In this section, we introduce the class of stationary Gaussian processes. We also present the main idea behind composite likelihood estimation. At last, we derive an asymptotic theory for parametric estimation of the former based on the latter.

\subsection{Stationary Gaussian processes}

In this paper, we consider processes of the form
\begin{equation} \label{equation:Y}
Y_{t} = \mu + \nu X_{t}, \quad t \in \R,
\end{equation}
where $\mu \in \R$, $\nu>0$, and $X = (X_{t})_{t \in \mathbb{R}}$ is a stationary Gaussian process with mean zero and unit variance, i.e. the marginal distribution of $X_{t}$ is standard normal. We denote the autocovariance function (ACF) of $Y$ at lag $h$ as $\gamma_{h} = \text{cov}(Y_{t},Y_{t+h}) =  \nu^{2} E(X_{t}X_{t+h}) = \nu^{2} \rho_{h}$, where $\rho_{h}$ is the ACF of $X$.\footnote{Throughout the paper, we loosely employ ACF to represent both the autocovariance function and autocorrelation function. These are identical for $X$ and related through a scaling by $\nu^{2}$ for $Y$.}

\begin{assumption} The ACF of $X$ has the property \label{assumption:correlation}
\begin{equation} \label{equation:maruyama}
\lim_{h \rightarrow \infty} \rho_{h} = 0.
\end{equation}
\end{assumption}
A stationary Gaussian process is ergodic if and only if Assumption \ref{assumption:correlation} is fulfilled \citep{maruyama:49a}. Since we need a law of large numbers to hold for the log-composite likelihood function, the above can therefore be viewed as a minimal regularity condition.

\subsection{Composite likelihood}

We suppose $\Delta > 0$ is a fixed time gap and let $y = (y_{1},y_{2}, \dots, y_{n})^{ \top}$ be a discrete realization of $n$ equidistant observations of the random vector $Y_{n}^{ \Delta} = (Y_{ \Delta}, Y_{2 \Delta}, \dots, Y_{n \Delta})^{ \top}$, where $Y$ is the stationary continuous-time process in \eqref{equation:Y}. We assume $\theta \in \Theta$ is a finite-dimensional parameter vector that determines the distribution of $Y$, where $\Theta \subset \mathbb{R}^{p}$ is a compact set. The data-generating value of $\theta$ is denoted $\theta_{0}$. Moreover, we initially enforce that $\mu = 0$ and this is known to the econometrician. In Remark \ref{remark:mu}, we elaborate further on the effect of estimating the mean.

As the model is now parametric, we can estimate $\theta$ by maximum likelihood:
\begin{equation*}
\hat{ \theta}_{ \text{MLE}} \equiv \underset{ \theta \in \Theta}{\arg \max} \ l( \theta;y),
\end{equation*}
where $l( \theta; y)$ is the full Gaussian log-likelihood function of the sample $y$ given $\theta$, i.e.
\begin{equation*}
l( \theta;y) \propto - \log \vert \Sigma_{n} ( \theta) \vert - y^{ \top} \Sigma_{n}^{-1}( \theta) y,
\end{equation*}
and $\Sigma_{n}( \theta) = E_{ \theta} \big[(Y_{n}^{ \Delta} - \mu)(Y_{n}^{ \Delta} - \mu)^{ \top} \big]$ is the $n \times n$ covariance matrix of $Y_{n}^{ \Delta}$.

The maximum likelihood estimator, $\hat{ \theta}_{ \text{MLE}}$, is consistent, asymptotically efficient, and follows a limiting normal distribution under standard regularity conditions \citep[e.g.,][]{newey-mcfadden:94a}.

However, even in the Gaussian setting, moderate values of $n$ can render numerical optimization of the log-likelihood function infeasible, because the complexity of computing the determinant and inverting the covariance matrix grows very rapidly. For example, deploying an algorithm such as the Cholesky factorization or LU decomposition results in a computational budget of $O(n^{3})$. Less naive approaches exploit the fact that the covariance matrix of a stationary process observed on an equidistant partition has a Toeplitz structure \citep[e.g.,][]{levinson:46a, durbin:60a}, yielding a faster calculation speed of $O(n^{2})$. An extra refinement has been achieved with so-called ``superfast'' algorithms that operate on Toeplitz matrices using the fast Fourier transform \citep[see][]{brent-gustavson-yun:80a}. The latter exhibit near-linear growth $O(n\log^2(n))$ but suffer from numerical instability in practice \citep[see][]{stewart:03a}. Moreover, the scaling of these algorithms are still inferior to the linear rate $O(n)$. This poses serious issues in a high-frequency setting, where $n$ is often prohibitively large.

Set against this backdrop, we exploit the composite likelihood framework of \citet{lindsay:88a}, building on the earlier concept of pseudo-likelihood from \citet{besag:74a} in the spatial setting, see also the survey by \citet{varin-reid-firth:11a}. To describe the idea, assume for the moment that $y$ is a sample of length $n$ of a continuous random variable $Y$ defined on a probability space $( \Omega, \mathcal{F}, \mathbb{P})$.\footnote{In our setting, the observations are a discrete realization of a continuous-time process.} In general, the composite likelihood estimator maximizes a weighted product of likelihoods of marginal or conditional events. We let $f$ denote the density of $Y$ and suppose that $( \mathcal{A}_{1}, \dots, \mathcal{A}_{M})$ is a collection of events, $\mathcal{A}_{m} \in \mathcal{F}$, with likelihood $L_{m}( \theta;y) \propto f(y \in \mathcal{A}_{m}; \theta)$. The composite likelihood is then defined as
\begin{equation} \label{equation:cl}
CL( \theta;y) = \prod_{m=1}^{M} L_{m}( \theta;y)^{w_{m}},
\end{equation}
where $w_{1}, \dots, w_{M}$ are nonnegative weights with $\sum_{m} w_{m} = 1$. In the remainder of the paper, we set $w_{k} = M^{-1}$ and omit it from \eqref{equation:cl}.\footnote{\citet{lindsay-yi-sun:11a} analyze a more formal approach for designing an efficient weighting scheme---called the Best Weighted Estimating Function (BWEF)---in order to maximize efficiency. It bears resemblance to choosing an optimal weight matrix in GMM estimation. This is a formidable numerical challenge in composite likelihood estimation, as it normally requires inversion of a large-dimensional matrix. To the extend that unequal weights can improve inference, our results can be viewed as conservative.}

The information consists of any conditional or marginal events. Hence, many variants of composite likelihood exist. Full likelihood is the special case $CL( \theta; y) = L( \theta; y) = f( y; \theta)$. The independence likelihood is $CL( \theta; y) = \prod_{i=1}^{n} f(y_{i}; \theta)$. It permits inference on marginal parameters only and is the full likelihood under actual independence. However, it is necessary to add events formed from blocks of observations to estimate parameters that control the dependence structure, such as the pairwise likelihood of \citet{cox-reid:04a}, $CL( \theta; y) = \prod_{i=1}^{n} \prod_{j=i+1}^{n} f(y_{i},y_{j}; \theta)$, or the restricted version named consecutive pairwise likelihood by \citet{davis-yau:11a}, $CL( \theta; y) = \prod_{i=1}^{n-K} \prod_{j=i}^{i+K}f(y_{i},y_{j}; \theta)$, which includes neighboring pairs up to some order $K \in \mathbb{N}$ with $K < n$. Another candidate is the conditional likelihood $CL( \theta; y) = \prod_{i=1}^{n} f(y_{i} \mid y \setminus y_{i}; \theta)$.

The maximum composite likelihood estimator (MCLE) is the argmax of $CL( \theta; y)$---or the natural logarithm of it:
\begin{equation*}
\hat{ \theta}_{ \mathrm{MCLE}} \equiv \underset{\theta \in \Theta}{ \arg \max} \ cl( \theta;y),
\end{equation*}
where
\begin{equation*}
cl( \theta;y) = \log CL( \theta;y) = \sum_{m=1}^{M} \log L_{m}( \theta;y).
\end{equation*}
The composite log-likelihood function is, in general, not proportional to the full likelihood, so the model is misspecified \citep[see, e.g., ][]{white:82a}. However, the data-generating probability measure is closely related to the assumed parametric distribution, since the marginal (or conditional) densities included in \eqref{equation:cl} are extracted from the true model and correctly specified. Hence, the score of $cl( \theta;y)$ satisfies the first Bartlett identity and forms a collection of unbiased estimating equations. As we show below, it follows by a law of large numbers that $\hat{ \theta}_{ \mathrm{MCLE}}$ is consistent for the true parameter value, $\theta_{0}$, and, if the memory in the process is not too strong, it is also asymptotically normal.\footnote{If the parameterization of the ACF is wrong, $\hat{ \theta}_{ \mathrm{MCLE}}$ generally does not converge to $\theta_{0}$, but rather a pseudo-parameter, $\theta^{*}$, that minimizes a ``composite'' Kullback-Leibler divergence between the true and assumed model, i.e. $\theta^{*} = \underset{\theta \in \Theta}{ \arg \max} \ \mathbb{E} \left[ cl( \theta; Y) \right]$, where the expectation is with respect to the true probability measure.}

To showcase the promise of MCLE compared to ``brute force'' MLE, we simulate an fOU process (defined in Section \ref{section:fOU}) with a true parameter vector given by Panel B in Table \ref{table:sim-ou-q=3-N=12-mu=0}. The Hurst exponent (the parameter of main interest) is $H = 0.1$, so we are in a rough setting. We implement CL as explained in Section \ref{section:simulation}. In Panel A of Figure \ref{figure:comparison}, we plot the relative runtime of CL versus ML as a function of $n = N \times T$ (measured against the left $y$-axis, while the actual runtime can be read of from the right $y$-axis). Here, $T$ is the number of days in the sample, whereas $N$ is the number of intraday observations. We set $N = 12$ and take $T = 10, 20, \dots, 250$. As the sample size grows, CL is completed in a small fraction of the ML runtime. The maximal length of the data vector is $n = 3{,}000$ here, which is rather small, but around this point MLE  becomes unresponsive (at least on our desktop) or runs out of memory. In Panel B of the figure, we plot the relative root mean squared error (RMSE) of CL versus ML, formatted as before, in terms of estimating $H$. The chasm between the RMSE of MCLE and MLE is relatively constant, because both converge at rate $n^{-1/2}$ in this setting. Hence, here the main drawback of MCLE is the loss of (asymptotic) efficiency.

\begin{figure}[ht!]
\begin{center}
\caption{Comparison of MLE and MCLE.}
\label{figure:comparison}
\begin{tabular}{cc}
\small{Panel A: Runtime.} & \small{Panel B: RMSE.} \\
\includegraphics[height=8.00cm,width=0.48\textwidth]{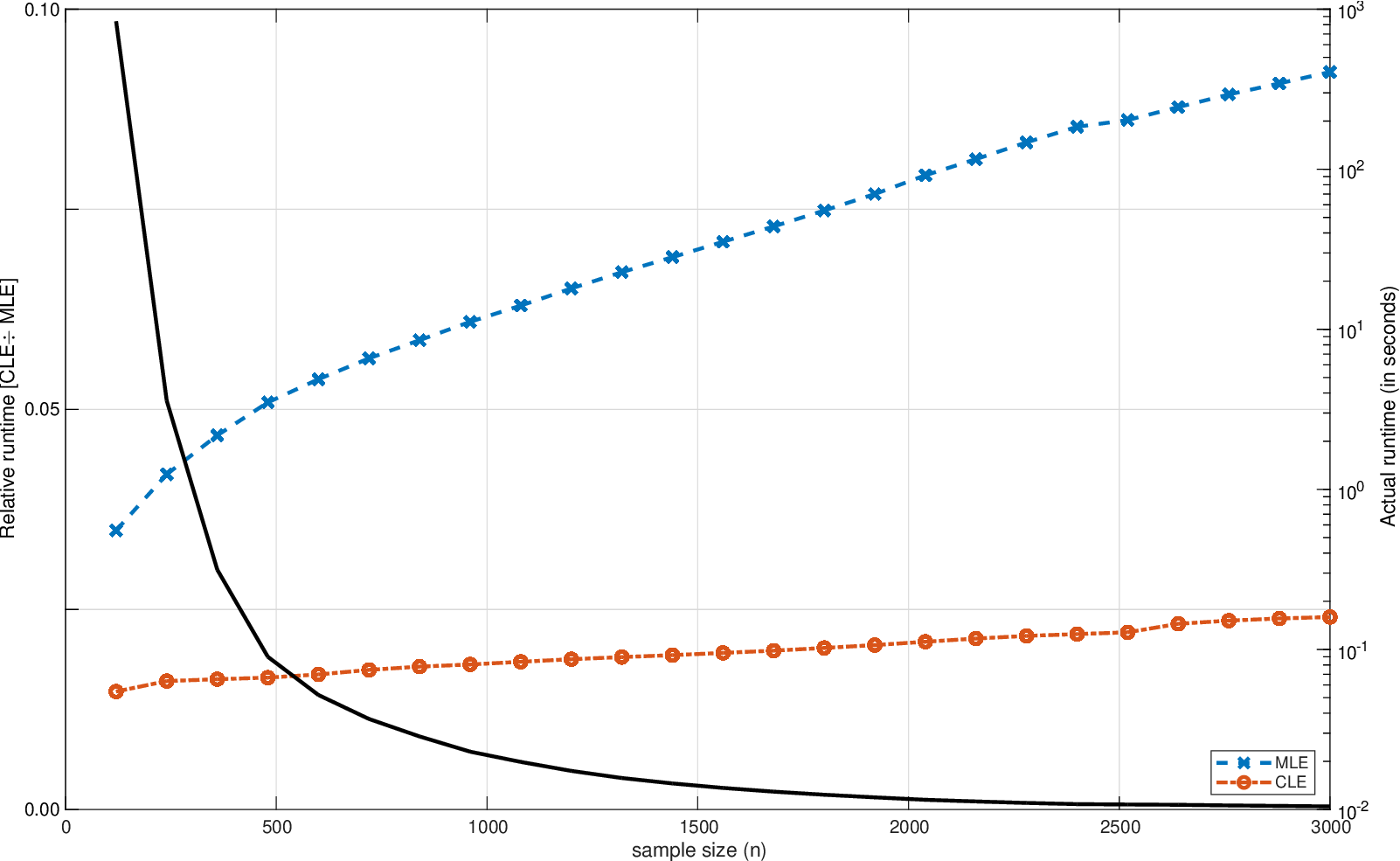} &
\includegraphics[height=8.00cm,width=0.48\textwidth]{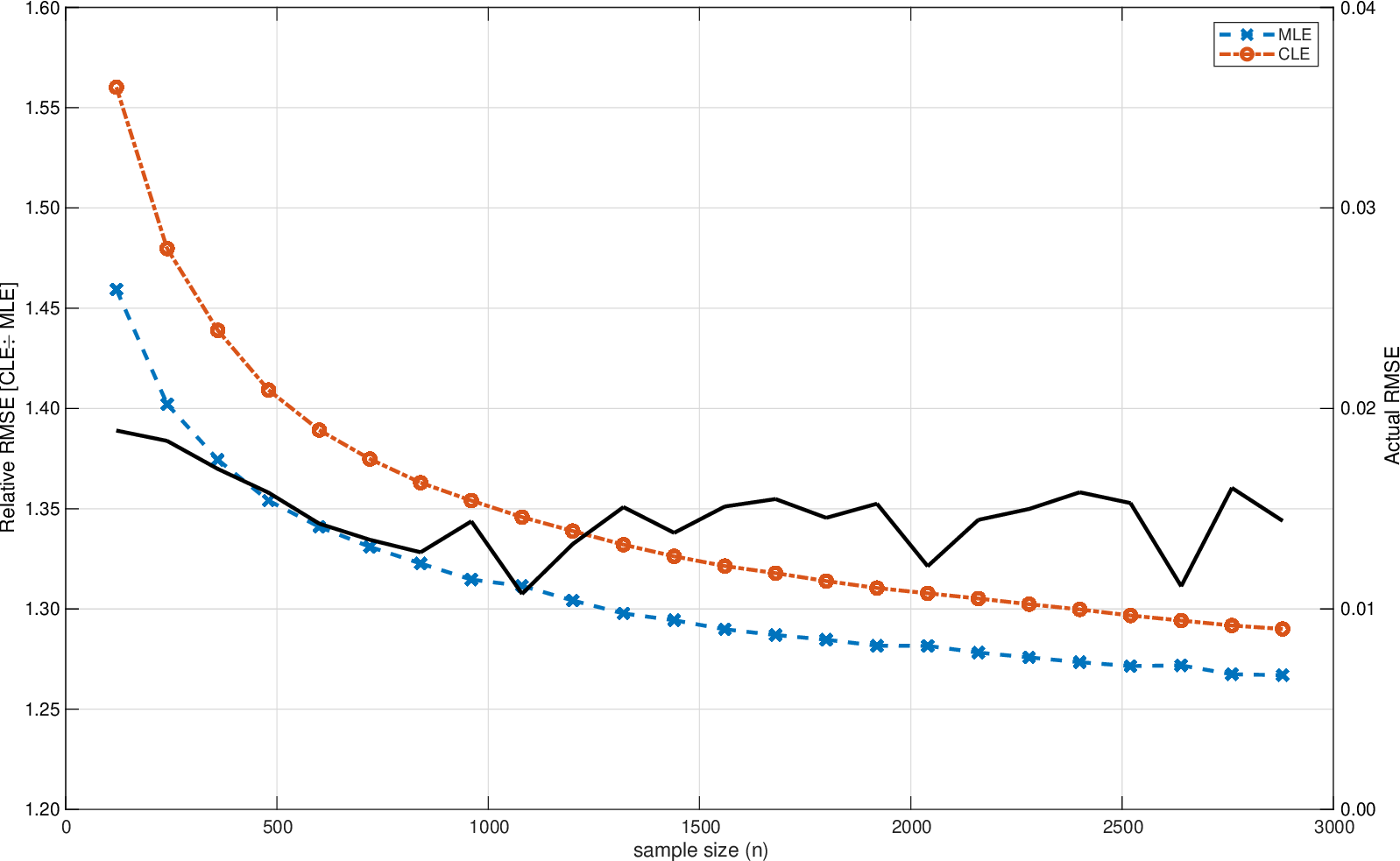} \\
\end{tabular}
\begin{scriptsize}
\parbox{\textwidth}{\emph{Note.} We simulate an fOU process, defined in Section \ref{section:fOU}. The three-dimensional parameter vector $\theta = ( \kappa, \nu, H)^{ \top}$ follows Panel B in Table \ref{table:sim-ou-q=3-N=12-mu=0} (the mean is not estimated). The sample size is $n = T \times N$, which means $N$ intraday observations over a period of $T = 10, 20, \dots, 250$ days. As in our simulation and empirical implementation, we fix $N = 12$. We estimate $\theta$ with ML and MCL. In Panel A, we show the (actual and relative) runtime, while Panel B shows the associated (actual and relative) root mean squared error (RMSE) of the Hurst parameter, $H$. The true value of $H$ is 0.1.}
\end{scriptsize}
\end{center}
\end{figure}

In the next subsection, we establish an asymptotic theory for a $q$-wise composite likelihood estimator, where the composite likelihood function is formed as an equal weighted product of marginal events defined via the selection of tuples of up to length $q$. It nests both the independence likelihood and the pairwise likelihood as special cases. We take $K \in \mathbb{N}$ fixed and let $Q$ be a collection of $K$ $q^{j}$-tuples of natural numbers, where $q^{j} \in \mathbb{N}$, i.e. tuples of the form $k^{j} = (0,k_{2}^{j}, \dots, k_{q^{j}}^{j}) \in \{0 \} \times \mathbb{N}^{q^{j}-1}$ for $j = 1, \dots, K$, indexing which observations to add. Here, we use the convention that $\mathbb{N}^{0} = \emptyset$. To ease notation, we set $q = \max_{j = 1, \dots, K} q^{j}$. Moreover, we suppose without loss of generality that the indices are increasing and that $\max_{j = 1, \dots, K}  k_{q^{j}}^{j} < n$. For $k^{j} \in Q$, we let $f_{k^{j}}(y_{i}^{k^{j}}; \theta) = f_{k^{j}}(y_{i}, y_{i+k_{2}^{j}}, \dots, y_{i+k_{q^{j}}^{j}}; \theta)$ denote the density of $Y_{i}^{k^{j}} = (Y_{i\Delta},Y_{(i+k_{2}^{j})\Delta}, \dots,Y_{(i+k_{q^{j}}^{j})\Delta})$, which by stationary is independent of $i$.\footnote{With this notation, maximum likelihood corresponds to $q = n$ and $K = 1$.}

The $q$-wise log-composite likelihood can be written
\begin{equation} \label{equation:qcl}
cl( \theta;y) = \sum_{k^{j} \in Q} \sum_{i=1}^{n-k^{j}_{q^{j}}} \log f_{k^{j}}(y_{i}^{k^{j}}; \theta).
\end{equation}

\subsection{Asymptotic theory}

To prove our asymptotic theory, we need an identification assumption to ensure that the log-composite likelihood is uniquely maximized at $\theta_{0}$:

\begin{assumption} \label{assumption:identification}
We assume that for any $\theta \in \Theta$ with $\theta \neq \theta_{0} : \sum_{k^{j} \in Q} f_{k^{j}}(y^{k^{j}}, \theta) \neq \sum_{k^{j} \in Q} f_{k^{j}} (y^{k^{j}}, \theta_{0})$ for a set of vectors $y \in \mathbb{R}^{q}$ of non-zero measure.
\end{assumption}
Assumption \ref{assumption:identification} is a standard identification condition formulated in terms of a density. Together with the moment condition $\mathbb{E}(| \log f(y; \theta)|) < \infty$, it is sufficient to ensure that the log-composite likelihood has a unique maximum at $\theta_{0}$ by the information inequality, see \citet[][Lemma 2.2]{newey-mcfadden:94a}. The existence of the moment holds trivially for the Gaussian distribution. Moreover, in our setting the identification condition can be verified if $K$ and $q$ are sufficiently large and diverse relative to the dimension of $\theta$, $p$. This requires that all the parameters of the model enter into the densities of the chosen events in $Q$, and they should do so in a way that is linearly independent.

In the Gaussian setting, this amounts to establish that the selected autocovariances identify $\theta_{0}$. This is the content of the next lemma.

\begin{lemma} \label{lemma:identification}
Let $\gamma_{h}( \theta)$ denote the parameterized ACF of $Y$. Suppose that
\begin{equation*}
\gamma_{k_{i}^{j}}( \theta) = \gamma_{k_{i}^{j}}( \theta_{0}) \text{ for } i = 1, \dots, q^{j}, \quad j = 1, \dots, K \Longrightarrow \theta = \theta_{0}.
\end{equation*}
Then, $\mathbb{E}[cl( \theta; Y)]$ has a unique maximum at $\theta = \theta_{0}$.
\end{lemma}
The main message of this lemma---as explained further in Appendix \ref{appendix:proof}---is that to ensure identification we need to select more lags, either through one long tuple or multiple shorter tuples, than the maximum number of times $\gamma_{h}( \theta)$ can cross for distinct $\theta$. Since we examine processes with non-oscillatory ACFs, the condition is typically fulfilled even with a relatively primitive implementation.

Our first result is a law of large numbers.

\begin{theorem} \label{theorem:lln}
Assume that the true model is the stationary Gaussian process defined in \eqref{equation:Y} and that Assumptions \ref{assumption:correlation} -- \ref{assumption:identification} hold. Then, under standard regularity conditions, as $n \rightarrow \infty$,
\begin{equation*}
\hat{ \theta}_{ \mathrm{MCLE}} \overset{ \mathbb{P}}{ \longrightarrow} \theta_{0}.
\end{equation*}
\end{theorem}

\noindent The proofs of our results are presented in Appendix \ref{appendix:proof}.

Next, we derive the asymptotic distribution of our estimator. The theorem consists of several parts, because both the rate of convergence of $\hat{ \theta}_{ \mathrm{MCLE}}$ and the shape of the limiting distribution of the estimation error depends on the persistence of the process.

\begin{theorem} \label{theorem:clt}
Suppose the conditions from Theorem \ref{theorem:lln} hold and that $\theta_{0}$ is an interior point, i.e. $\theta_{0} \in \mathrm{int}( \Theta)$. Let $s_{n}( \theta) = \frac{ \partial}{ \partial \theta} cl( \theta; y)$ be the score and $L_{ \infty}$ a slowly varying function at infinity, i.e. $\lim_{x \rightarrow \infty} L_{ \infty}(tx) \div L_{ \infty}(x) = 1$, for all $t > 0$. Then, as $n \rightarrow \infty$, it holds that\footnote{The notation $f(h) \sim g(h)$ means asymptotic equivalence, i.e. $f(h)/g(h) \rightarrow 1$ as $h \rightarrow \infty$.}
\begin{enumerate}
\item If $\int_{0}^{ \infty} | \gamma_{h}| \mathrm{d}h < \infty$ or if $\gamma_{h} \sim h^{-\beta}L_{ \infty}(h)$ as $h \rightarrow \infty$ for $\beta \in(1/2,1]$, then
\begin{equation*}
\sqrt{n} \big( \hat{ \theta}_{ \mathrm{MCLE}} - \theta_{0} \big) \overset{d}{ \longrightarrow} N \left(0,G( \theta_{0})^{-1} \right),
\end{equation*}
with
\begin{equation*}
G( \theta_{0})^{-1} = H( \theta_{0})^{-1} V( \theta_{0}) H( \theta_{0})^{-1},
\end{equation*}
\begin{equation*}
H( \theta) = - \sum_{k \in Q} \mathbb{E} \left( \frac{ \partial^{2}}{ \partial \theta \partial \theta^{ \top}} \log \left( f_{k}(Y_{k}; \theta) \right)   \right),
\end{equation*}
and
\begin{align*}
V( \theta) &= \lim_{n \rightarrow \infty} \frac{1}{n} \mathbb{E} \left(s_{n}( \theta) s_{n}( \theta)^{ \top} \right) \\
&= \frac{1}{4} \sum_{k^{1}, k^{2} \in Q} \sum_{l=- \infty}^{ \infty} \sum_{j_{1}, j_{2} \in k_{1}} \sum_{j_{3}, j_{4} \in k_{2}} \left[ \frac{ \partial}{ \partial \theta_{r}} \big( \Sigma_{k^{1}}^{-1}( \theta) \big)_{j_{1},j_{2}} \right]_{r=1}^{p}\\&\times \left( \left[ \frac{ \partial}{ \partial \theta_{r}} \big( \Sigma_{k^{2}}^{-1}( \theta) \big)_{j_{3},j_{4}} \right]_{r=1}^{p} \right)^{ \top} \\
& \times \left( \gamma_{(l+j_{1}-j_{3}) \Delta} \gamma_{(l+j_{2}-j_{4}) \Delta} + \gamma_{(l+j_{1}-j_{4}) \Delta} \gamma_{(l+j_{2}-j_{3}) \Delta} \right).
\end{align*}

\item If $\gamma_{h} \sim h^{-1/2}L_{ \infty}(h)$ as $h \rightarrow \infty$, then
\begin{equation*}
\frac{\sqrt{n}}{ \sqrt{L_{ \gamma}(n)}} \big( \hat{ \theta}_{ \mathrm{MCLE}} - \theta_{0} \big) \overset{d}{ \longrightarrow}H( \theta_{0})^{-1} \Psi(\theta_{0}) N(0,1),
\end{equation*}
with $H( \theta_{0})$ defined as above, and where $\Psi$ is a $p \times 1$ vector with elements
\begin{equation}
\Psi_{r}(\theta) =  \sum_{k^{j} \in Q} \sum_{j_{1}, j_{2} \in k^{j}} \frac{ \partial}{ \partial \theta_{r}} \big( \Sigma_{k}^{-1}( \theta) \big)_{j_{1},j_{2}},
\end{equation}
for $r = 1, \dots, p$, and $L_{ \gamma}$ is a slowly varying function, related to $L_\infty$, as defined in Appendix \ref{appendix:proof}.

\item If $\gamma_{h} \sim h^{- \beta}L_{ \infty}(h)$ as $h \rightarrow \infty$ for $\beta \in(0,1/2)$, then
\begin{equation}
\frac{n^{ \beta}}{ \sqrt{L_{2}(n)}} ( \hat{ \theta}_{ \mathrm{MCLE}} - \theta_{0}) \overset{d}{ \longrightarrow} H( \theta_{0})^{-1} \Psi(\theta_{0}) Z_{2,H}(1),
\end{equation}
where $L_{2}(n) = 4C_{2}L_{ \infty}^{2}(n)$, $C_{2} = [(1-2 \beta)(2- \beta)]^{-1}$, $\Psi$ is as above, and $Z_{2,H}$ is a Rosenblatt process with parameter $H=1- \beta/2$. 
\end{enumerate}
\end{theorem}

The first part of Theorem \ref{theorem:clt} covers the short memory setting with an integrable ACF, but it also permits long memory for $\beta \in (1/2,1]$ with a square-integrable ACF. Here, we get a standard central limit theorem with a standard $n^{-1/2}$ rate of convergence. However, because the composite likelihood misspecifies the true likelihood, the second Bartlett identity (or information matrix equality) is violated, so the Fisher information is replaced with the \citet{godambe:60a} information, which has the sandwich form. The ratio of these matrices determines the loss of efficiency of $\hat{ \theta}_{ \text{MCLE}}$ relative to $\hat{ \theta}_{ \text{MLE}}$ in this setting, since the latter achieves the Cramer-Rao lower bound asymptotically.

The second part is the borderline case with $\beta = 1/2$, where the polynomial decay in the long memory setting starts to impair the properties of $\hat{ \theta}_{ \mathrm{MCLE}}$. First, the ACF is not square-integrable, so the variability matrix has to be scaled by an additional term, $L_{ \gamma}(n)$, whereby the rate of convergence is reduced. An inspection of the expression for $L_{ \gamma}(n)$ in Appendix \ref{appendix:proof} shows that, even if $\lim_{h \rightarrow \infty} L_{ \infty}(h)$ exists, the convergence rate of $\hat{ \theta}_{ \text{MCLE}}$ deteriorates to $\log(n)n^{-1/2}$, as consistent with \citet{davis-yau:11a}. Second, although the limit remains Gaussian, it is important to observe that the entire parameter vector is loading on a single standard normal variate, so the asymptotic distribution is degenerate. This follows from Theorem 4 in \citet{hosking:96a}, because the difference between the sample ACF at any distinct lags converge to zero in probability for $\beta \leq 1/2$.

In the last part, $\beta \in (0,1/2)$, the ACF subsides exceedingly slow. Here, we get a non-central limit theorem with a non-standard convergence rate, which is nevertheless a classical result in the realm of long memory processes. The asymptotic distribution remains degenerate, but it is of Rosenblatt form.

The bottom line is that for the MCLE, the asymptotic theory presents a discontinuity in both the rate of convergence and the form of the limiting distribution, such that inference is complicated.  This is the price we pay for not doing MLE. On the other hand, as shown in Figure \ref{figure:comparison}, estimation can be done in a split second, and it may be the only feasible likelihood-based procedure when looking at very large sample sizes, as demonstrated in one of the applications in our empirical investigation.


Our analysis extends \citet{davis-yau:11a} from pairwise likelihood to the $q$-wise setting. However, even with $q = 2$ we allow for the presence of a slowly varying function in the ACF, which is important for some Gaussian processes, such as the Brownian semistationary process with a power law kernel \citep[e.g.][]{bennedsen-lunde-pakkanen:22a}.

\begin{remark} \label{remark:mu}
In the above, $\mu = 0$ is known to hold. It is of course straightforward to transform a problem with a known nonzero mean into the previous setting by appropriate centering. What if the mean is unknown and, hence, estimated?

The composite likelihood estimator is derived from the score of the log-composite likelihood function, and since we work with Gaussian data, the stochastic terms in the score relate to the autocovariance function of the process. Hence, the main difference between these settings corresponds to the limit theory for autocovariances calculated with a known mean or estimated mean. This was analyzed in the Gaussian setting by \citet{hosking:96a}, who shows that if the memory is not too strong (i.e. short memory or long memory with $\beta \geq 1/2$), there is no impact. That is, estimating the mean does not alter neither the rate of convergence nor the limiting distribution.

On the other hand, with pervasive long memory (i.e. $\beta < 1/2$), estimation of the mean makes a difference relative to our case 3. While the convergence rate is unchanged, the limiting distribution is no longer a Rosenblatt distribution, but it has a more complicated expression.\footnote{To get a better understanding of the difference, one can compare the cumulants of the limiting distribution in Theorem 4.1 of \citet{hosking:96a} with those of the Rosenblatt distribution, which are given by $\kappa_{1} = 0$, $\kappa_{2} = 1$, and $\kappa_{k} = 2^{k-1}(k-1)! [\sigma( \beta)]^{k} c_{k}$, for $k \geq 2$, where $\sigma( \beta) = 2^{-1/2}[(1-2 \beta)(1- \beta)]^{1/2}$ and $c_{k} = \int_{0}^{1} \int_{0}^{1} \cdots  \int_{0}^{1} |x_{1}-x_{2}|^{- \beta} |x_{2}-x_{3}|^{- \beta} \cdots |x_{k-1}-x_{k}|^{- \beta} |x_{k}-x_{1}|^{- \beta} \mathrm{d}x_{1} \mathrm{d}x_{2} \cdots \mathrm{d}x_{k}$.}

Mean estimation also instills a finite sample downward bias in the sample ACF, which impairs the estimation in practice. We shed light on this in the simulation study.
\end{remark}

\begin{remark}
The assumption of Gaussianity for the background driving stationary process can be weakened by assuming that $X_{t}$ is strictly stationary and ergodic with a finite second moment. We can then interpret \eqref{equation:qcl} as a Gaussian quasi-composite likelihood function. In that case, $\hat{ \theta}_{ \text{MCLE}}$ is consistent for the subset of the parameter vector that determines the mean and covariance function of $X_{t}$, since the Gaussian log-likelihood is uniquely maximized at the true value of these \citep[][Lemma A.2]{bollerslev-wooldridge:92a}. Furthermore, the rates of convergence in Theorem \ref{theorem:clt} remain unchanged. We can also establish asymptotic normality in the short-memory case by imposing that $X_{t}$ is strong mixing with an appropriate rate of decay on the mixing coefficients together with a moment condition \citep[see, e.g.,][]{keenan:08a}. In the long-memory setting, however, deriving an explicit expression for the limiting distribution appears to be a more onerous problem.
\end{remark}

\section{A view toward stochastic volatility} \label{section:example}

In this section, we narrow down the estimation of continuous-time stationary Gaussian processes to a particular pair of parametric models. The first is the fractional Ornstein-Uhlenbeck process (fOU), which has attracted a lot of attention in the recent literature on roughness in the stochastic volatility of financial asset returns \citep[e.g.,][]{gatheral-jaisson-rosenbaum:18a, fukasawa-takabatake-westphal:22a, bolko-christensen-pakkanen-veliyev:23a, wang-xiao-yu:23a, shi-yu-zhang:24a, wang-xiao-yu-zhang:25a}. The second is from a so-called Cauchy class, which was proposed in \cite{gneiting-schlather:04a}. In financial econometrics, it has---to the best of our knowledge---only been studied as a volatility model by \citet{bennedsen-lunde-pakkanen:22a}.\footnote{The Cauchy class has seen widespread application in other fields, such as modeling infectious disease spread, network traffic, or forecasting wind speed.}

We begin with a notion of roughness and memory.

\begin{definition} \label{definition:roughness} A stationary stochastic process has roughness index $\alpha$ if its ACF has the following behavior around the origin:
\begin{equation} \label{equation:roughnesss}
(1- \rho_{h}) \sim L(h) \vert h \vert^{2 \alpha+1}, \quad \vert h \vert \rightarrow 0,
\end{equation}
for a function, $L(h)$, that is slowly varying at zero, and $\alpha \in (-1/2,1/2)$.
\end{definition}
The concept of roughness is closely related to the behavior of the (second-order) variogram of a zero-mean real-valued stochastic process with stationary increments that was studied in \citet{gatheral-jaisson-rosenbaum:18a}. For example, (the increments of) a Brownian motion has $\alpha = 0$. Moreover, because it is a continuous Gaussian process, there is a one-to-one correspondence between the roughness index and the H\"{o}lder continuity of its sample paths. Thus, we refer to a process as being \textit{rough} if it has a negative roughness index, $\alpha < 0$, as its sample paths are then more irregular---i.e. less H\"{o}lder continuous---than those of a Brownian motion. Conversely, we call a process \textit{smooth} if it has a positive roughness index, $\alpha > 0$.

\begin{definition} \label{definition:memory}
A stationary stochastic process has \textit{long memory} of degree $\beta$ if its ACF decays at a polynomial rate as the lag length increases:
\begin{equation} \label{equation:memory}
\rho_{h} \sim L_{ \infty}(h) \vert h \vert^{- \beta}, \quad \vert h \vert \rightarrow \infty,
\end{equation}
for some function, $L_{ \infty}(h)$, that is slowly varying at infinity, and $\beta \in (0,1]$.
\end{definition}
The requirement on $\beta$ implies that the ACF is not integrable, which is another common way to define long memory in a time series. Here, we stick to Definition \ref{definition:memory}, which suffices for the purposes of this paper. If the ACF is integrable, we say the process has \textit{short memory}, even if it is does not vanish at an exponential rate.

\subsection{The fOU process} \label{section:fOU}

The first class of parametric stationary Gaussian processes that we entertain in this paper is the fOU process:
\begin{equation} \label{equation:fOU}
\mathrm{d}Y_{t} = - \kappa (Y_{t} - \mu) \mathrm{d}t + \tilde{ \nu} \mathrm{d}B_{t}^{H},
\end{equation}
where $B^{H}$ is a fractional Brownian motion (fBm).\footnote{A fBm is a centered Gaussian process, which is uniquely characterized by its covariance function $\mathbb{E} \big[B_{t}^{H}B_{s}^{H} \big] = \frac{1}{2} \big(t^{2H}+s^{2H}- \vert t-s \vert^{2H} \big)$. The parameter $H \in (0,1)$ is the Hurst exponent. The fBm has stationary increments called fractional Gaussian noise (fGn).}

The unique stationary solution to this stochastic differential equation was derived in \citet{cheridito-kawaguchi-maejima:03a}:\footnote{The fBm is not a semimartingale, except for $H = 1/2$ where it collapses to a sBm. Therefore, the solution in \eqref{equation:solution} cannot be defined as a standard It\^{o} integral, but it should be interpreted as a pathwise Riemann-Stieltjes integral. Although this is more restrictive, it is meaningful here, since the fBm is continuous and the exponential function is of bounded variation.}
\begin{equation} \label{equation:solution}
Y_{t} = \mu + \nu X_{t},
\end{equation}
where $X_{t} = b \int_{- \infty}^{t} e^{- \kappa(t-s)} \mathrm{d} B_{s}^{H}$, $\nu = \tilde{ \nu} / b$, $b = \sqrt{ \dfrac{ \kappa^{2H}}{ H \Gamma(2H)}}$, and $\Gamma$ is the Gamma function. The additional scaling in front of the stochastic integral in \eqref{equation:solution} relative to \eqref{equation:fOU} is a reparameterization of the model to ensure that $\nu$ is the standard deviation of $Y_{t}$. It follows from $\text{var}\left( \int_{- \infty}^{t} e^{- \kappa(t-s)} \mathrm{d} B_{s}^{H} \right) = \kappa^{-2H} H \Gamma(2H)$.

The ACF of the fOU is available from \citet{garnier-solna:18a}:
\begin{align} \label{equation:acf}
\begin{split}
\rho_{h} = \frac{1}{2 \kappa^{2H}} \left( \frac{1}{2} \int_{-\infty}^{ \infty} e^{- \vert y \vert} \vert \kappa h+y \vert^{2H} \mathrm{d}y - \vert \kappa h \vert^{2H} \right).
\end{split}
\end{align}
The ACF of the fBm decays hyperbolically with $\rho_{h} = O(h^{2(H-1)})$ for $H \neq 1/2$, whereas for $H = 1/2$ it is geometrically bounded. It follows from the Kolmogorov-Chentsov theorem that fBm has a modification, where the sample paths are H\"{o}lder continuous with exponent $\gamma \in(0,H)$, so $H$ determines the regularity of the process. In particular, the Hurst parameter is linked to the roughness index in Definition \ref{definition:roughness} through the relation $\alpha = H - 1/2$. Hence, for $H < 1/2$ the fBm is rough, whereas for $H > 1/2$ it is smooth. Furthermore, it is readily seen that its connection with Definition \ref{definition:memory} is $\beta = 2(1-H)$. This implies that the fBm has long memory for $H > 1/2$, whereas it possesses short memory for $H < 1/2$. As explained in \citet{cheridito-kawaguchi-maejima:03a}, the ACF of the fOU process inherits the decay rate of the background driving fBm, so these properties are passed on one-to-one.

The bottom line is that the Hurst exponent controls both the roughness and memory properties of the fOU process. This renders it capable of exhibiting \textit{either} roughness \textit{or} long memory, but not both. This is also rather intuitive, since an examination of the ACF of the time-changed process $B_{ct}^{H}$, for some $c > 0$, reveals that the fBm is self-similar.

\subsection{The Cauchy class} \label{section:cc}

There is, to the best of our knowledge, no known representation of these processes in the form of a stochastic differential equation describing their dynamic. Hence, the Cauchy class is solely defined via the correlation structure:
\begin{equation} \label{equation:cauchy-acf}
\rho_{h} = (1 + |h|^{2 \alpha+1})^{- \beta / (2 \alpha+1)},
\end{equation}
for $\alpha \in (-1/2, 1/2)$ and $\beta > 0$. It follows that $\rho_{h}$ behaves as \eqref{equation:roughnesss} for $h \rightarrow 0$ and as \eqref{equation:memory} for $h \rightarrow \infty$.

The Cauchy class can also exhibit roughness and long memory. However, whereas these properties are forged together by the Hurst exponent for the fOU process, here the features are decoupled and controlled by separate parameters, $\alpha$ and $\beta$, such that it can exhibit \textit{both} roughness ($\alpha < 0$) \textit{and} long memory ($\beta < 1$).

Motivated by its ability to decouple the short- and long-term persistence, this process was studied as an empirical model for log-spot variance in \citet{bennedsen-lunde-pakkanen:22a}. They fitted it with method-of-moments to high-frequency data from the E-mini S\&P 500 futures contract and indeed found evidence of both roughness and long memory.

\section{Monte Carlo analysis} \label{section:simulation}

\subsection{Simulation design}

To inspect the small sample properties of our CLE procedure, we simulate the stationary Gaussian processes from Section \ref{section:example}, i.e. the fOU process and Cauchy class. The parameter vector is four-dimensional, i.e. $\theta_{ \text{OU}} = ( \mu, \kappa, \nu, \alpha)^{ \top}$ and $\theta_{ \text{CC}} = ( \mu, \beta, \nu, \alpha)^{ \top}$. To shed light on the importance of Remark \ref{remark:mu}, we fix $\mu = 0$ and examine the impact of knowing this a priori, so $\mu$ is not estimated, and the setup where $\mu$ is inferred alongside the covariance-related parameters.

The fOU model was studied in \citet{bolko-christensen-pakkanen-veliyev:23a} and \citet{wang-xiao-yu:23a}, among others, to describe the stochastic log-variance of financial asset returns, $Y_{t} = \log \sigma_{t}^{2}$.\footnote{\citet{wang-xiao-yu-zhang:25a} prove the asymptotic theory for the time-domain MLE of the fOU process. In contrast to MCLE, the limiting distribution of the covariance-related part of $\hat{ \theta}_{ \text{MLE}}$ is normal with an $n^{-1/2}$ rate of convergence for every $H \in (0,1)$ whether or not the mean is estimated. The mean estimator itself has a slower convergence rate (continuous in $H$) for $H > 1/2$, but it stays Gaussian.} We follow this analogue here and adopt the Monte Carlo design of the former article by examining five different settings for $\theta_{ \text{OU}}$  (listed in Table \ref{table:sim-ou-q=3-N=12-mu=0}) ranging from very rough and short memory ($\alpha = -0.45$ or $H = 0.05$) to smooth and long memory ($\alpha = 0.20$ or $H = 0.70$).

We include an equal number of distinct parameter values for $\theta_{ \text{CC}}$ (listed in Table \ref{table:sim-cc-q=3-N=12-mu=0}), keeping the roughness index fixed at the values of the Hurst exponent from the fOU process, i.e. $\alpha = -0.45, -0.40, \dots, 0.20$, while moving the long-run persistence in the opposite direction, i.e. $\beta = 0.25, 0.50, \dots, 1.25$. This design ensures that the latter features multiple settings with both roughness and long memory.

We simulate over the interval $[0,T]$, where $T$ can be interpreted as the number of days, setting $T = 1095, 1825, 2555$. This amounts to three, five, and seven years worth of data. The largest value roughly matches the sample size in our empirical application to realized variance. We assume the process is observed discretely on an equispaced partition of $N$ points per day, so the $n = N \times T$ data vector is $(Y_{i \Delta})_{i = 1, \dots, N \times T}$, where $\Delta = 1/N$. We follow \citet{christensen-thyrsgaard-veliyev:19a} and base the analysis on $N = 12$ observations per unit interval, yielding a new observation of the log-spot variance every second hour.

We simulate the Cauchy class via circulant embedding, which is an exact discretization that runs in $O(n \log n)$ time. However, since the ACF of the fOU needs to be evaluated with the help of numerical integration we opt for a slightly different approach for that model. In particular, at a generic time $t$ the solution of the fOU process in \eqref{equation:solution} implies that:
\begin{equation*}
Y_{t} = \mu + (Y_{t- \Delta} - \mu) e^{- \kappa \Delta} + \nu \sqrt{ \dfrac{ \kappa^{2H}}{H \Gamma(2H)}} \int_{t- \Delta}^{t} e^{- \kappa (t-s)} \text{d}B_{s}^{H}.
\end{equation*}
We approximate the stochastic integral as $\int_{t- \Delta}^{t} e^{- \kappa (t-s)} \text{d}B_{s}^{H} \simeq e^{- \kappa \Delta / 2} (B_{t}^{H} - B_{t- \Delta}^{H})$ and simulate increments of the fBm with circulant embedding.

To implement MCLE, we need to decide the terms to include in \eqref{equation:qcl} through the set $Q$. We experiment here with triwise CLE (i.e. $q^{j} = q = 3$) with $K = 5$, where the elements of $k^{j}$ are separated by a fixed amount, i.e. $k^{j} = (0, \ell, 2 \ell)$ for $\ell = 1, 6, 12, 24, 60$. This is a sufficient number of lags to guarantee the identification of our models.\footnote{A routine, but tedious, derivation shows that the Cauchy class is identified with at least four different lags in the criterion function. The proof of this result is available at request. Identification of the fOU process is more difficult to establish, but it is examined in some length in Appendix B of \citet{bolko-christensen-pakkanen-veliyev:23a}.} This setup means that $cl( \theta; y)$ is a summation of densities of the form $f(y_{i}, y_{i+ \ell}, y_{i+2 \ell}; \theta)$. Apart from lag one with a two-hour gap, the other choices of $\ell$ correspond do a half-, whole-, two-, and five-day horizon.\footnote{In a previous version of this article, we relied on a more ``casual'' implementation of CLE that incorporated various lags of the ACF, including ones at the very long end of the curve, based on ``intuition.'' However, upon a closer examination we noted that the performance of CLE deteriorates with such a configuration. We therefore switched to the current setup to boost the performance of CLE and better illuminate its promise.} The reported results are rather robust to the exact configuration of the CL procedure.\footnote{Appendix \ref{appendix:simulation} reports results for the pairwise setting with $q = 2$. The conclusion from this analysis is that triwise CLE systematically leads to lower RMSE in the rough setting, but the pairwise cousin can also be a superior alternative. The change of going from $q = 2$ to $q = 3$, or vice versa, is often small, however. We leave the pursuit of, arguably model-dependent, optimal choices of $q$ and $K$ to a future endeavor.}

We benchmark MCLE against a method-of-moments estimator (MME). We follow the implementation of \citet{wang-xiao-yu:23a} for the fOU process and \citet{bennedsen-lunde-pakkanen:22a} for the Cauchy class.\footnote{We do not compare our MCLE approach to the GMM estimator of \citet{bolko-christensen-pakkanen-veliyev:23a}, since the latter relies on integrated variance, whereas we study spot variance. Neither do we benchmark it against the time-domain MLE of \citet{wang-xiao-yu-zhang:25a} or the frequency-domain approximate likelihood of \citet{fukasawa-takabatake-westphal:22a} and \citet{shi-yu-zhang:24a} for the fOU, since our procedure cannot be expected to be superior to these in terms of efficiency. As we argue, MCLE trades (a slight) efficiency loss against (a much) faster runtime. The MME of \citet{wang-xiao-yu:23a}, which is also extremely fast and has an asymptotic theory that closely resembles that of MCLE---with a discontinuity in the rate of convergence and the limiting distribution---represents a more level playing field and natural competitor in our view.} The details are reviewed in Appendix \ref{appendix:mm}. It is worth pointing out that the first-step MME estimator of $\alpha$ (or $H$) relies on change-of-frequency approach based on second-differenced data. Hence, it does not depend on whether the mean is known or estimated. We should also emphasize that we start the MCLE routine with the MME as initial value. The number of Monte Carlo replications is 10,000.

\subsection{Simulation outcome}

The Monte Carlo evidence for the fOU process and Cauchy class are summarized separately in a series of tables below. In each one, we report the outcome of the MCLE next to that of the MME. For every setting, we show the bias of the parameter estimator (i.e., sample average less true value) and its standard error across replica (in parentheses). To compare MCLE with MME, we add the relative root mean squared error, i.e., $\text{RMSE}( \hat{ \theta}_{ \text{MCLE}}) / \text{RMSE}( \hat{ \theta}_{ \text{MME}})$. A number below (above) one indicates that MCLE is more (less) efficient than MME.\footnote{The relative RMSE should be interpreted with caution. The actual RMSEs of MCLE and MME are often small, so it does not take much to move the ratio by a lot in either direction.}

\subsubsection{The fOU process}

Looking at the results in Tables \ref{table:sim-ou-q=3-N=12-mu=0} -- \ref{table:sim-ou-q=3-N=12-mu=1} for the fOU process first, we observe that both MCLE and MME perform well. The parameters are estimated with an immaterial bias and the standard errors are small, even with a limited sample. The sole exception is the mean-reversion parameter, $\kappa$, which is sometimes estimated with a slight upward bias, but it vanishes as the sample grows and is negligible for $T = 2{,}555$. Whether or not the mean is estimated, the conclusion from Panels A -- D is that MCLE is far superior to MME in terms of estimating the covariance-related parameters $(\kappa, \nu, H)$. This includes the relevant scenario with a Hurst exponent in the roughness territory. However, we also observe that its relative efficiency decreases as we move toward long memory. Indeed, in Panel E with $H = 0.7$ the MCLE is slightly inferior to MME for $(\kappa, H)$ that determine the correlation structure of $Y$. This observation is consistent with the evidence for the Cauchy class. In contrast to that model, however, there is only a minuscule change in performance on $(\kappa, \nu, H)$ between the known and estimated $\mu$. In the latter setting, the estimation of $\mu$ itself is a dead draw.

\begin{sidewaystable}[ht!]
\setlength{ \tabcolsep}{0.15cm}
\begin{center}
\caption{Parameter estimation of the fOU process [$q = 3$ and $\mu$ known].}
\label{table:sim-ou-q=3-N=12-mu=0}
\begin{footnotesize}
\begin{tabular}{lrrrrrrrrrrrrrr}
\hline \hline
Parameter & Value & \multicolumn{3}{c}{$T = 1{,}095$} && \multicolumn{3}{c}{$T = 1{,}825$} && \multicolumn{3}{c}{$T = 2{,}555$} \\
\cline{3-5} \cline{7-9} \cline{11-13}
&& \multicolumn{1}{c}{MCLE} & \multicolumn{1}{c}{MME} & \multicolumn{1}{c}{RMSE$_{r}$} && \multicolumn{1}{c}{MCLE} & \multicolumn{1}{c}{MME} & \multicolumn{1}{c}{RMSE$_{r}$} && \multicolumn{1}{c}{MCLE} & \multicolumn{1}{c}{MME} & \multicolumn{1}{c}{RMSE$_{r}$} \\
\hline 
Panel A: \\
$\mu$ & 0.000 & & & & & & & & & & & \\
$\kappa$ & 0.005 & -0.000 (0.005) & 0.006 (0.020) & 0.241 & & -0.000 (0.004) & 0.003 (0.013) & 0.270 & & -0.000 (0.003) & 0.003 (0.010) & 0.291 \\
$\nu$ & 1.250 & 0.000 (0.009) & 0.000 (0.041) & 0.214 & & 0.000 (0.007) & -0.000 (0.032) & 0.210 & & -0.000 (0.006) & 0.000 (0.027) & 0.208 \\
$\alpha$ & -0.450 & 0.000 (0.003) & -0.000 (0.013) & 0.206 & & 0.000 (0.002) & -0.000 (0.010) & 0.203 & & 0.000 (0.002) & -0.000 (0.009) & 0.200 \\
Panel B: \\
$\mu$ & 0.000 & & & & & & & & & & & \\
$\kappa$ & 0.010 & 0.001 (0.006) & 0.004 (0.014) & 0.438 & & 0.001 (0.005) & 0.002 (0.010) & 0.459 & & 0.001 (0.004) & 0.002 (0.008) & 0.476 \\
$\nu$ & 0.750 & 0.000 (0.006) & 0.000 (0.024) & 0.260 & & 0.000 (0.005) & 0.000 (0.019) & 0.256 & & -0.000 (0.004) & 0.000 (0.016) & 0.254 \\
$\alpha$ & -0.400 & 0.000 (0.004) & -0.000 (0.013) & 0.331 & & 0.000 (0.003) & -0.000 (0.010) & 0.326 & & 0.000 (0.003) & -0.000 (0.009) & 0.321 \\
Panel C: \\
$\mu$ & 0.000 & & & & & & & & & & & \\
$\kappa$ & 0.015 & 0.002 (0.006) & 0.002 (0.008) & 0.806 & & 0.001 (0.005) & 0.001 (0.006) & 0.804 & & 0.001 (0.004) & 0.001 (0.005) & 0.809 \\
$\nu$ & 0.500 & 0.000 (0.007) & 0.000 (0.016) & 0.440 & & 0.000 (0.006) & -0.000 (0.013) & 0.434 & & 0.000 (0.005) & 0.000 (0.011) & 0.428 \\
$\alpha$ & -0.200 & 0.000 (0.008) & -0.000 (0.012) & 0.647 & & 0.000 (0.006) & -0.000 (0.010) & 0.640 & & 0.000 (0.005) & -0.000 (0.008) & 0.630 \\
Panel D: \\
$\mu$ & 0.000 & & & & & & & & & & & \\
$\kappa$ & 0.035 & 0.002 (0.009) & 0.002 (0.010) & 0.928 & & 0.001 (0.007) & 0.001 (0.008) & 0.917 & & 0.001 (0.006) & 0.001 (0.007) & 0.918 \\
$\nu$ & 0.300 & 0.000 (0.006) & 0.000 (0.010) & 0.609 & & 0.000 (0.005) & -0.000 (0.008) & 0.603 & & 0.000 (0.004) & -0.000 (0.007) & 0.594 \\
$\alpha$ & 0.000 & 0.000 (0.010) & -0.000 (0.012) & 0.888 & & 0.000 (0.008) & -0.000 (0.009) & 0.882 & & 0.000 (0.007) & -0.000 (0.008) & 0.869 \\
Panel E: \\
$\mu$ & 0.000 & & & & & & & & & & & \\
$\kappa$ & 0.070 & 0.004 (0.018) & 0.003 (0.016) & 1.141 & & 0.003 (0.014) & 0.002 (0.012) & 1.132 & & 0.002 (0.012) & 0.001 (0.010) & 1.134 \\
$\nu$ & 0.200 & 0.001 (0.007) & 0.000 (0.008) & 0.927 & & 0.000 (0.006) & -0.000 (0.006) & 0.921 & & 0.000 (0.005) & -0.000 (0.005) & 0.913 \\
$\alpha$ & 0.200 & 0.001 (0.013) & -0.000 (0.011) & 1.229 & & 0.000 (0.010) & -0.000 (0.008) & 1.230 & & 0.000 (0.009) & -0.000 (0.007) & 1.224 \\
\hline \hline
\end{tabular}
\end{footnotesize}
\smallskip
\begin{scriptsize}
\parbox{0.98\textwidth}{\emph{Note.} 
We simulate the process in the caption of the table 10,000 times on the interval $[0,T]$, where $T$ is interpreted as the number of days. 
There are $N = 12$ observations per unit interval corresponding to twelve every day.
The true value of the parameter vector appears to the left in Panel A -- E. 
We estimate $\theta$ with the maximum composite likelihod estimation (MCLE) procedure developed in the main text, and benchmark it against a method-of-moments estimator (MME).
The table shows the Monte Carlo average of each parameter estimate across simulations (standard deviation in parenthesis).
The last column reports the RMSE ratio defined as $\text{RMSE}_{r} = \text{RMSE}( \hat{ \theta}_{ \text{MCLE}}) / \text{RMSE}( \hat{ \theta}_{ \text{MME}})$. 
}
\end{scriptsize}
\end{center}
\end{sidewaystable}

\begin{sidewaystable}[ht!]
\setlength{ \tabcolsep}{0.15cm}
\begin{center}
\caption{Parameter estimation of the fOU process [$q = 3$ and $\mu$ estimated].}
\label{table:sim-ou-q=3-N=12-mu=1}
\begin{footnotesize}
\begin{tabular}{lrrrrrrrrrrrrrr}
\hline \hline
Parameter & Value & \multicolumn{3}{c}{$T = 1{,}095$} && \multicolumn{3}{c}{$T = 1{,}825$} && \multicolumn{3}{c}{$T = 2{,}555$} \\
\cline{3-5} \cline{7-9} \cline{11-13}
&& \multicolumn{1}{c}{MCLE} & \multicolumn{1}{c}{MME} & \multicolumn{1}{c}{RMSE$_{r}$} && \multicolumn{1}{c}{MCLE} & \multicolumn{1}{c}{MME} & \multicolumn{1}{c}{RMSE$_{r}$} && \multicolumn{1}{c}{MCLE} & \multicolumn{1}{c}{MME} & \multicolumn{1}{c}{RMSE$_{r}$} \\
\hline 
Panel A: \\
$\mu$ & 0.000 & -0.001 (0.258) & -0.001 (0.258) & 0.998 & & 0.000 (0.164) & 0.000 (0.165) & 0.997 & & 0.000 (0.120) & 0.000 (0.121) & 0.997 \\
$\kappa$ & 0.005 & 0.001 (0.006) & 0.008 (0.023) & 0.240 & & 0.000 (0.004) & 0.004 (0.014) & 0.268 & & 0.000 (0.003) & 0.003 (0.011) & 0.290 \\
$\nu$ & 1.250 & 0.000 (0.009) & 0.000 (0.041) & 0.214 & & 0.000 (0.007) & -0.000 (0.032) & 0.210 & & -0.000 (0.006) & 0.000 (0.027) & 0.208 \\
$\alpha$ & -0.450 & 0.000 (0.003) & -0.000 (0.013) & 0.207 & & 0.000 (0.002) & -0.000 (0.010) & 0.204 & & 0.000 (0.002) & -0.000 (0.009) & 0.200 \\
Panel B: \\
$\mu$ & 0.000 & -0.000 (0.094) & -0.000 (0.094) & 0.999 & & 0.000 (0.062) & 0.000 (0.062) & 0.999 & & -0.000 (0.047) & -0.000 (0.047) & 0.999 \\
$\kappa$ & 0.010 & 0.002 (0.007) & 0.005 (0.015) & 0.443 & & 0.001 (0.005) & 0.003 (0.010) & 0.462 & & 0.001 (0.004) & 0.002 (0.008) & 0.479 \\
$\nu$ & 0.750 & 0.000 (0.006) & 0.000 (0.024) & 0.260 & & 0.000 (0.005) & 0.000 (0.019) & 0.256 & & -0.000 (0.004) & 0.000 (0.016) & 0.254 \\
$\alpha$ & -0.400 & 0.000 (0.004) & -0.000 (0.013) & 0.331 & & 0.000 (0.003) & -0.000 (0.010) & 0.327 & & 0.000 (0.003) & -0.000 (0.009) & 0.321 \\
Panel C: \\
$\mu$ & 0.000 & -0.000 (0.097) & -0.000 (0.097) & 1.002 & & 0.000 (0.070) & 0.000 (0.070) & 1.001 & & -0.000 (0.056) & -0.000 (0.056) & 1.001 \\
$\kappa$ & 0.015 & 0.003 (0.007) & 0.003 (0.009) & 0.814 & & 0.002 (0.005) & 0.002 (0.006) & 0.811 & & 0.001 (0.004) & 0.001 (0.005) & 0.815 \\
$\nu$ & 0.500 & 0.000 (0.007) & 0.000 (0.016) & 0.441 & & 0.000 (0.006) & -0.000 (0.013) & 0.435 & & 0.000 (0.005) & 0.000 (0.011) & 0.429 \\
$\alpha$ & -0.200 & 0.000 (0.008) & -0.000 (0.012) & 0.647 & & 0.000 (0.006) & -0.000 (0.010) & 0.641 & & 0.000 (0.005) & -0.000 (0.008) & 0.631 \\
Panel D: \\
$\mu$ & 0.000 & -0.000 (0.068) & -0.000 (0.068) & 1.001 & & 0.000 (0.053) & 0.000 (0.053) & 1.001 & & -0.000 (0.045) & -0.000 (0.045) & 1.001 \\
$\kappa$ & 0.035 & 0.004 (0.010) & 0.004 (0.011) & 0.938 & & 0.003 (0.007) & 0.002 (0.008) & 0.926 & & 0.002 (0.006) & 0.002 (0.007) & 0.927 \\
$\nu$ & 0.300 & 0.001 (0.006) & 0.000 (0.010) & 0.613 & & 0.000 (0.005) & -0.000 (0.008) & 0.605 & & 0.000 (0.004) & -0.000 (0.007) & 0.596 \\
$\alpha$ & 0.000 & 0.001 (0.010) & -0.000 (0.012) & 0.890 & & 0.001 (0.008) & -0.000 (0.009) & 0.883 & & 0.000 (0.007) & -0.000 (0.008) & 0.871 \\
Panel E: \\
$\mu$ & 0.000 & 0.000 (0.069) & 0.000 (0.069) & 1.001 & & 0.000 (0.059) & 0.000 (0.059) & 1.001 & & 0.000 (0.054) & 0.000 (0.054) & 1.000 \\
$\kappa$ & 0.070 & 0.012 (0.018) & 0.009 (0.016) & 1.157 & & 0.008 (0.013) & 0.006 (0.012) & 1.148 & & 0.007 (0.011) & 0.005 (0.010) & 1.151 \\
$\nu$ & 0.200 & 0.002 (0.007) & 0.000 (0.008) & 0.951 & & 0.002 (0.006) & -0.000 (0.006) & 0.938 & & 0.001 (0.005) & -0.000 (0.005) & 0.926 \\
$\alpha$ & 0.200 & 0.004 (0.013) & -0.000 (0.011) & 1.219 & & 0.003 (0.010) & -0.000 (0.008) & 1.222 & & 0.002 (0.009) & -0.000 (0.007) & 1.216 \\
\hline \hline
\end{tabular}
\end{footnotesize}
\smallskip
\begin{scriptsize}
\parbox{0.98\textwidth}{\emph{Note.} 
We simulate the process in the caption of the table 10,000 times on the interval $[0,T]$, where $T$ is interpreted as the number of days. 
There are $N = 12$ observations per unit interval corresponding to twelve every day.
The true value of the parameter vector appears to the left in Panel A -- E. 
We estimate $\theta$ with the maximum composite likelihod estimation (MCLE) procedure developed in the main text, and benchmark it against a method-of-moments estimator (MME).
The table shows the Monte Carlo average of each parameter estimate across simulations (standard deviation in parenthesis).
The last column reports the RMSE ratio defined as $\text{RMSE}_{r} = \text{RMSE}( \hat{ \theta}_{ \text{MCLE}}) / \text{RMSE}( \hat{ \theta}_{ \text{MME}})$. 
}
\end{scriptsize}
\end{center}
\end{sidewaystable}

\subsubsection{The Cauchy class}

\begin{sidewaystable}[ht!]
\setlength{ \tabcolsep}{0.15cm}
\begin{center}
\caption{Parameter estimation of the Cauchy class [$q = 3$ and $\mu$ known].}
\label{table:sim-cc-q=3-N=12-mu=0}
\begin{footnotesize}
\begin{tabular}{lrrrrrrrrrrrrrr}
\hline \hline
Parameter & Value & \multicolumn{3}{c}{$T = 1{,}095$} && \multicolumn{3}{c}{$T = 1{,}825$} && \multicolumn{3}{c}{$T = 2{,}555$} \\
\cline{3-5} \cline{7-9} \cline{11-13}
&& \multicolumn{1}{c}{MCLE} & \multicolumn{1}{c}{MME} & \multicolumn{1}{c}{RMSE$_{r}$} && \multicolumn{1}{c}{MCLE} & \multicolumn{1}{c}{MME} & \multicolumn{1}{c}{RMSE$_{r}$} && \multicolumn{1}{c}{MCLE} & \multicolumn{1}{c}{MME} & \multicolumn{1}{c}{RMSE$_{r}$} \\
\hline 
Panel A: \\
$\mu$ & 0.000 & & & & & & & & & & & \\
$\beta$ & 0.250 & 0.073 (0.141) & -0.077 (0.099) & 1.330 & & 0.061 (0.123) & -0.096 (0.082) & 1.226 & & 0.054 (0.114) & -0.105 (0.072) & 1.164 \\
$\nu$ & 1.250 & -0.001 (0.072) & -0.001 (0.072) & 1.000 & & -0.001 (0.066) & -0.001 (0.066) & 1.000 & & -0.001 (0.062) & -0.001 (0.062) & 1.000 \\
$\alpha$ & -0.450 & 0.006 (0.014) & -0.031 (0.010) & 0.622 & & 0.005 (0.012) & -0.032 (0.008) & 0.544 & & 0.005 (0.012) & -0.032 (0.007) & 0.503 \\
Panel B: \\
$\mu$ & 0.000 & & & & & & & & & & & \\
$\beta$ & 0.500 & 0.034 (0.131) & -0.147 (0.121) & 0.835 & & 0.024 (0.109) & -0.160 (0.102) & 0.727 & & 0.019 (0.098) & -0.168 (0.091) & 0.659 \\
$\nu$ & 0.750 & -0.000 (0.021) & -0.000 (0.021) & 1.001 & & -0.000 (0.018) & -0.000 (0.018) & 1.001 & & -0.000 (0.016) & -0.000 (0.016) & 1.000 \\
$\alpha$ & -0.400 & 0.003 (0.013) & -0.059 (0.013) & 0.303 & & 0.002 (0.011) & -0.059 (0.010) & 0.257 & & 0.002 (0.010) & -0.059 (0.009) & 0.230 \\
Panel C: \\
$\mu$ & 0.000 & & & & & & & & & & & \\
$\beta$ & 0.750 & 0.014 (0.105) & 0.023 (0.138) & 0.759 & & 0.009 (0.083) & 0.005 (0.123) & 0.675 & & 0.006 (0.071) & -0.005 (0.111) & 0.638 \\
$\nu$ & 0.500 & -0.000 (0.016) & -0.000 (0.016) & 1.004 & & -0.000 (0.013) & -0.000 (0.013) & 1.003 & & -0.000 (0.011) & -0.000 (0.011) & 1.003 \\
$\alpha$ & -0.200 & 0.001 (0.012) & -0.060 (0.013) & 0.265 & & 0.001 (0.009) & -0.060 (0.010) & 0.212 & & 0.000 (0.008) & -0.060 (0.008) & 0.182 \\
Panel D: \\
$\mu$ & 0.000 & & & & & & & & & & & \\
$\beta$ & 1.000 & 0.011 (0.104) & 0.025 (0.158) & 0.656 & & 0.006 (0.081) & 0.011 (0.132) & 0.614 & & 0.004 (0.069) & 0.005 (0.116) & 0.594 \\
$\nu$ & 0.300 & -0.000 (0.009) & -0.000 (0.009) & 1.004 & & -0.000 (0.007) & -0.000 (0.007) & 1.004 & & -0.000 (0.006) & -0.000 (0.006) & 1.004 \\
$\alpha$ & 0.000 & 0.000 (0.010) & -0.019 (0.012) & 0.585 & & 0.000 (0.008) & -0.018 (0.009) & 0.506 & & 0.000 (0.007) & -0.018 (0.008) & 0.451 \\
Panel E: \\
$\mu$ & 0.000 & & & & & & & & & & & \\
$\beta$ & 1.250 & 0.009 (0.109) & 0.038 (0.184) & 0.590 & & 0.005 (0.085) & 0.029 (0.146) & 0.577 & & 0.003 (0.072) & 0.025 (0.125) & 0.572 \\
$\nu$ & 0.200 & -0.000 (0.006) & -0.000 (0.006) & 1.003 & & -0.000 (0.005) & -0.000 (0.005) & 1.003 & & -0.000 (0.004) & -0.000 (0.004) & 1.003 \\
$\alpha$ & 0.200 & 0.000 (0.009) & 0.041 (0.011) & 0.279 & & -0.000 (0.007) & 0.041 (0.008) & 0.222 & & -0.000 (0.006) & 0.041 (0.007) & 0.190 \\
\hline \hline
\end{tabular}
\end{footnotesize}
\smallskip
\begin{scriptsize}
\parbox{0.98\textwidth}{\emph{Note.} 
We simulate the process in the caption of the table 10,000 times on the interval $[0,T]$, where $T$ is interpreted as the number of days. 
There are $N = 12$ observations per unit interval corresponding to twelve every day.
The true value of the parameter vector appears to the left in Panel A -- E. 
We estimate $\theta$ with the maximum composite likelihod estimation (MCLE) procedure developed in the main text, and benchmark it against a method-of-moments estimator (MME).
The table shows the Monte Carlo average of each parameter estimate across simulations (standard deviation in parenthesis).
The last column reports the RMSE ratio defined as $\text{RMSE}_{r} = \text{RMSE}( \hat{ \theta}_{ \text{MCLE}}) / \text{RMSE}( \hat{ \theta}_{ \text{MME}})$. 
}
\end{scriptsize}
\end{center}
\end{sidewaystable}

\begin{sidewaystable}[ht!]
\setlength{ \tabcolsep}{0.15cm}
\begin{center}
\caption{Parameter estimation of the Cauchy class [$q = 3$ and $\mu$ estimated].}
\label{table:sim-cc-q=3-N=12-mu=1}
\begin{footnotesize}
\begin{tabular}{lrrrrrrrrrrrrrr}
\hline \hline
Parameter & Value & \multicolumn{3}{c}{$T = 1{,}095$} && \multicolumn{3}{c}{$T = 1{,}825$} && \multicolumn{3}{c}{$T = 2{,}555$} \\
\cline{3-5} \cline{7-9} \cline{11-13}
&& \multicolumn{1}{c}{MCLE} & \multicolumn{1}{c}{MME} & \multicolumn{1}{c}{RMSE$_{r}$} && \multicolumn{1}{c}{MCLE} & \multicolumn{1}{c}{MME} & \multicolumn{1}{c}{RMSE$_{r}$} && \multicolumn{1}{c}{MCLE} & \multicolumn{1}{c}{MME} & \multicolumn{1}{c}{RMSE$_{r}$} \\
\hline 
Panel A: \\
$\mu$ & 0.000 & -0.000 (0.360) & -0.000 (0.360) & 1.000 & & -0.000 (0.345) & -0.000 (0.345) & 1.000 & & 0.000 (0.336) & 0.000 (0.336) & 1.000 \\
$\beta$ & 0.250 & 0.229 (0.107) & 0.038 (0.109) & 1.725 & & 0.196 (0.087) & -0.004 (0.093) & 1.753 & & 0.177 (0.077) & -0.030 (0.081) & 1.743 \\
$\nu$ & 1.250 & -0.053 (0.016) & -0.052 (0.016) & 1.000 & & -0.048 (0.015) & -0.048 (0.015) & 1.000 & & -0.046 (0.014) & -0.046 (0.014) & 1.000 \\
$\alpha$ & -0.450 & 0.021 (0.010) & -0.031 (0.010) & 0.754 & & 0.018 (0.008) & -0.032 (0.008) & 0.642 & & 0.016 (0.007) & -0.032 (0.007) & 0.579 \\
Panel B: \\
$\mu$ & 0.000 & 0.001 (0.139) & 0.001 (0.139) & 1.001 & & 0.001 (0.127) & 0.001 (0.127) & 1.000 & & 0.001 (0.119) & 0.001 (0.119) & 1.000 \\
$\beta$ & 0.500 & 0.131 (0.107) & -0.034 (0.107) & 1.287 & & 0.103 (0.087) & -0.067 (0.090) & 1.112 & & 0.088 (0.076) & -0.089 (0.081) & 0.961 \\
$\nu$ & 0.750 & -0.013 (0.011) & -0.013 (0.011) & 1.001 & & -0.011 (0.010) & -0.011 (0.010) & 1.001 & & -0.009 (0.009) & -0.009 (0.009) & 1.000 \\
$\alpha$ & -0.400 & 0.012 (0.011) & -0.059 (0.013) & 0.303 & & 0.009 (0.009) & -0.059 (0.010) & 0.250 & & 0.008 (0.008) & -0.059 (0.009) & 0.220 \\
Panel C: \\
$\mu$ & 0.000 & 0.001 (0.077) & 0.001 (0.077) & 1.001 & & 0.000 (0.065) & 0.000 (0.065) & 1.001 & & 0.000 (0.058) & 0.000 (0.058) & 1.001 \\
$\beta$ & 0.750 & 0.057 (0.100) & 0.140 (0.085) & 0.827 & & 0.039 (0.079) & 0.096 (0.085) & 0.768 & & 0.030 (0.068) & 0.071 (0.083) & 0.730 \\
$\nu$ & 0.500 & -0.006 (0.014) & -0.006 (0.014) & 1.004 & & -0.004 (0.011) & -0.004 (0.011) & 1.003 & & -0.003 (0.010) & -0.003 (0.010) & 1.003 \\
$\alpha$ & -0.200 & 0.005 (0.011) & -0.060 (0.013) & 0.261 & & 0.003 (0.009) & -0.060 (0.010) & 0.208 & & 0.003 (0.007) & -0.060 (0.008) & 0.179 \\
Panel D: \\
$\mu$ & 0.000 & 0.000 (0.032) & 0.000 (0.032) & 1.001 & & 0.000 (0.025) & 0.000 (0.025) & 1.001 & & 0.000 (0.022) & 0.000 (0.022) & 1.001 \\
$\beta$ & 1.000 & 0.034 (0.103) & 0.116 (0.119) & 0.732 & & 0.021 (0.080) & 0.076 (0.106) & 0.686 & & 0.015 (0.068) & 0.056 (0.097) & 0.660 \\
$\nu$ & 0.300 & -0.002 (0.009) & -0.002 (0.009) & 1.004 & & -0.001 (0.007) & -0.001 (0.007) & 1.003 & & -0.001 (0.006) & -0.001 (0.006) & 1.004 \\
$\alpha$ & 0.000 & 0.002 (0.010) & -0.018 (0.012) & 0.581 & & 0.001 (0.008) & -0.018 (0.009) & 0.501 & & 0.001 (0.007) & -0.018 (0.008) & 0.447 \\
Panel E: \\
$\mu$ & 0.000 & 0.000 (0.016) & 0.000 (0.016) & 1.001 & & 0.000 (0.013) & 0.000 (0.013) & 1.001 & & 0.000 (0.011) & 0.000 (0.011) & 1.001 \\
$\beta$ & 1.250 & 0.024 (0.110) & 0.108 (0.157) & 0.637 & & 0.014 (0.085) & 0.075 (0.131) & 0.610 & & 0.010 (0.072) & 0.059 (0.115) & 0.595 \\
$\nu$ & 0.200 & -0.001 (0.006) & -0.001 (0.006) & 1.003 & & -0.001 (0.005) & -0.000 (0.005) & 1.003 & & -0.000 (0.004) & -0.000 (0.004) & 1.003 \\
$\alpha$ & 0.200 & 0.001 (0.009) & 0.041 (0.011) & 0.278 & & 0.001 (0.007) & 0.041 (0.008) & 0.222 & & 0.000 (0.006) & 0.041 (0.007) & 0.190 \\
\hline \hline
\end{tabular}
\end{footnotesize}
\smallskip
\begin{scriptsize}
\parbox{0.98\textwidth}{\emph{Note.} 
We simulate the process in the caption of the table 10,000 times on the interval $[0,T]$, where $T$ is interpreted as the number of days. 
There are $N = 12$ observations per unit interval corresponding to twelve every day.
The true value of the parameter vector appears to the left in Panel A -- E. 
We estimate $\theta$ with the maximum composite likelihod estimation (MCLE) procedure developed in the main text, and benchmark it against a method-of-moments estimator (MME).
The table shows the Monte Carlo average of each parameter estimate across simulations (standard deviation in parenthesis).
The last column reports the RMSE ratio defined as $\text{RMSE}_{r} = \text{RMSE}( \hat{ \theta}_{ \text{MCLE}}) / \text{RMSE}( \hat{ \theta}_{ \text{MME}})$. 
}
\end{scriptsize}
\end{center}
\end{sidewaystable}

Turning next to the Cauchy class, the conclusion is two-fold. First, we see from Table \ref{table:sim-cc-q=3-N=12-mu=0}, where the mean is known, that the distribution of the covariance-related parameter estimators is again centered around the true value with minimal sampling variation for MCLE. The main difference is that we observe a slight upward bias in $\hat{ \beta}_{ \text{MCLE}}$, although it is well within a Monte Carlo standard error. By contrast, $\hat{ \beta}_{ \text{MME}}$ is regularly off target with a negative bias in Panel A reversing into a positive one in Panel E. It even moves in the wrong direction (toward a too low value) as $T$ increases for $\alpha = -0.45$ or $\alpha = -0.40$. The origin of this problem appears to be the first-stage MME of $\alpha$. As it turns out, in our numerical experiments $\hat{ \alpha}_{ \text{MME}}$ is biased for a fixed $\Delta$, and this causes the estimated persistence to deviate from the true one. In turn, this distorts the second-stage MME of $\beta$ with the latter acting as a suspension that pulls the model-implied ACF in the opposite direction. This explains the positive correlation between $\hat{ \beta}_{ \text{MME}}$ and $\hat{ \alpha}_{ \text{MME}}$. Since $\hat{ \alpha}_{ \text{MME}}$ is consistent in the infill asymptotic limit $\Delta \rightarrow 0$, this effect ought to disappear with more frequent sampling of the log-variance process, but that is of course not possible to enforce here, where $\Delta$ is fixed. This highlights a potential drawback of the two-step MME procedure.\footnote{To investigate this issue further, in Appendix \ref{appendix:simulation} we also report on a Monte Carlo study, where the volatility process is sampled once per day, i.e. $N = 1$. The broad conclusion is that the bias problems in $\alpha$ with MME are intensified, indicating that it may be more imprecise with less frequently sampled data. This is especially troublesome, since the MME is often applied to daily RV \citep[e.g.][]{gatheral-jaisson-rosenbaum:18a,bennedsen-lunde-pakkanen:22a,wang-xiao-yu:23a}.}

Second, as shown in Table \ref{table:sim-cc-q=3-N=12-mu=1}, where the mean is estimated, $\hat{ \beta}_{ \text{MCLE}}$ exhibits a more pronounced upward bias than before, in particular for the extreme persistence configurations in Panels A -- B, although it still dissipates as the sample size is increased. Measured by relative RMSE, the bias in $\hat{ \beta}_{ \text{MCLE}}$ is hurting a lot. Here, $\hat{ \beta}_{ \text{MME}}$ only has a slight bias, which is, however, again heading the wrong way as we collect more information. In Panel A, this is, in fact, initially helping to reduce the RMSE of $\hat{ \beta}_{ \text{MME}}$, because the bias crosses zero around $T = 1{,}825$. As consistent with the above, the sign of the bias in $\hat{ \beta}_{ \text{MME}}$ systematically is positive in Panels C -- E and is, in fact, more pronounced than before. Here, $\hat{ \beta}_{ \text{MCLE}}$ is evidently a better estimator both in terms of bias and variance.

The MCLE of the mean---a feasible Generalized Least Squares (GLS)-type weighted average---is unbiased \textit{across} replica, but it exhibits a great deal of sampling variation and is typically way off target in each \textit{individual} simulation, as evident from the magnitude of the standard error.\footnote{The stationary Gaussian process in \eqref{equation:Y} is basically a linear regression on a constant, for which we impose a parametric assumption on the variance-covariance matrix of the disturbance term. As the mean is unrelated to the other parameters in the models we inspect, it follows from \citet{magnus:78a} that the MLE of $\mu$ is the feasible GLS statistic calculated by plug-in of the MLEs of the parameters indexing $\Sigma$, i.e. $\hat{ \mu}_{ \text{MLE}} = ( \iota^{ \top} \hat{ \Sigma}^{-1} \iota)^{-1} \iota^{ \top} \hat{ \Sigma}^{-1} Y_{n}^{ \Delta}$, where $\iota = (1, \dots, 1)^{ \top}$. An inspection of the first-order condition $\partial cl( \theta;y) / \partial \mu = 0$ reveals that, since each summand in the composite log-likelihood function is Gaussian, $\hat{ \mu}_{ \text{MCLE}}$ also has a  GLS-type structure, but it entails a more complicated weighting scheme. The derivation of the exact expression is available at request. \label{footnote:mean}} In essence, the mean estimator is very noisy, especially for persistent processes. As explained in Remark \ref{remark:mu}, although this does not affect the convergence rate, \citet{hosking:96a} shows that in the long memory setting it induces a negative $O(n^{- \beta})$ bias in the ACF. The intuition is that when a process is highly persistent, even if it is far away from and never actually crosses the mean, it still tends to oscillate around the average. This is interpreted by the estimation procedure as shorter memory. The lower the $\beta$, the stronger is this effect.\footnote{To what extent does the increased roughness in the Cauchy class of processes explain the deterioration in $\hat{ \beta}_{ \text{MCLE}}$ in the presence of long memory? To shed light on this, we performed a robustness check, i.e. an extra simulation with the most long memory, $\beta = 0.25$ from Panel A, combined with the least roughness, $\alpha = 0.20$ from Panel E, together with an intermediate standard deviation, $\nu = 0.50$ from Panel C. In these experiments, the average $\hat{ \beta}_{ \text{MCLE}}$ value across different $T$ was around 0.37, so the bias was roughly cut in half relative to what we observe in Panel A of Table \ref{table:sim-cc-q=3-N=12-mu=1}, but it remained elevated.} We should note that while this effect is in principle also present for the fOU process, it is not visible. This is probably because the autoregressive parameter is increased in parallel with the Hurst index, forcing the process back toward its mean despite the increased dependence. In the Cauchy class, we do not exert direct control over the degree of mean reversion.

In the close, we should point out that if one is concerned about the slow convergence rate of the mean estimator under long memory it is possible to transform the problem by first-differencing the data. While this drops information about the mean, it should reduce the bias in the estimation of the autocovariance-related parameters at the expense of an increased variance. However, as the simulation results suggest that the MCLE performs adequately even with the original data in levels, also when the mean is estimated, we do not pursue this option here.

Overall, the results suggest that our MCLE framework works as intended, and it is at least on par with or even outperforms the MME approach.

\section{Empirical application} \label{section:empirical}

We implement our composite likelihood procedure on high-frequency data from the cryptocurrency market.\footnote{These markets are ideal to investigate the (pathwise) properties of a stochastic process, since trading is more or less never interrupted. In contrast, other asset classes feature periodic market closure. This can distort estimation of the persistence, since intermittent data are missing, potentially leading to incorrect matching of observations for the calculation of the sample ACF.} We look at the five largest free-floating coins in terms of market value at the end of 2024 (see, e.g., \url{coingecko.com}): Bitcoin (BTC), Ethereum (ETH), Solana (SOL), Ripple (XRP), and Binance (BNB). We examine the evolution of their spot exchange rate against Tether (USDT). The latter is a so-called ``stablecoin'', whose value is pegged at parity against the US dollar, i.e. 1USDT = 1USD.\footnote{According to Tether Limited, the token is backed one-to-one against fiat currency or cash equivalents (e.g., short-dated US Treasury bonds), although the alleged reserves were never audited. This is the cause of much controversy and has lead to occasional sell-offs in USDT, where its price temporarily dropped far below parity against the USD. By and large, however, Tether maintains a fairly constant price at the official peg with limited volatility.} We downloaded millisecond time-stamped tick-by-tick data free of charge from the Binance archive.\footnote{\url{https://data.binance.vision/}.} Following \citet{hansen-kim-kimbrough:22a}, we restrict attention to the period after January 1, 2019, where trading volume on the Binance platform was adequate. We collect data until March 27, 2025, so the sample consists of $T = 2{,}278$ days. The exception is SOLUDST that is only available from August 11, 2020 and therefore has a slightly shorter span of $T = 1{,}690$ days.

\begin{sidewaystable}[ht!]
\setlength{ \tabcolsep}{0.10cm}
\begin{center}
\caption{Descriptive statistics of cryptocurrency high-frequency data.}
\label{table:binance-descriptive}
\begin{small}
\begin{tabular}{lrrrrrrrrrrrrrrrrrrrr}
\hline \hline
&&&&&& \multicolumn{7}{c}{Realized variance} && \multicolumn{7}{c}{Trading volume} \\
\cline{7-13} \cline{15-21}
Ticker & $N$ & $\Delta_{N}^{ \text{ms}}$ & $\widebar{RV}$ & \multicolumn{1}{c}{$\widebar{TV}$}
&& \multicolumn{3}{c}{fOU} && \multicolumn{3}{c}{Cauchy}
&& \multicolumn{3}{c}{fOU} && \multicolumn{3}{c}{Cauchy} \\
\cline{7-9} \cline{11-13} \cline{15-17} \cline{19-21}
& & & &
&& $100 \times \kappa$ & $\nu$ & $\alpha$
&& $\beta$ & $\nu$ & $\alpha$
&& $100 \times \kappa$ & $\nu$ & $\alpha$
&& $\beta$ & $\nu$ & $\alpha$ \\
\hline
BTC & 2,044,826 & 42.3 & 52.0 & 174.2 && $\underset{(0.050)}{0.652}$ & $\underset{(0.008)}{1.182}$ & $\underset{(0.001)}{-0.375}$ && $\underset{(0.009)}{0.236}$ & $\underset{(0.017)}{1.181}$ & $\underset{(0.002)}{-0.293}$  && $\underset{(0.000)}{0.010}$ & $\underset{(0.006)}{0.972}$ & $\underset{(0.001)}{-0.415}$ && $\underset{(0.003)}{0.103}$ & $\underset{(0.009)}{0.972}$ & $\underset{(0.001)}{-0.364}$ \\
ETH & 977,713 & 88.4 & 65.2 & 80.6 && $\underset{(0.060)}{0.449}$ & $\underset{(0.008)}{1.143}$ & $\underset{(0.002)}{-0.377}$ && $\underset{(0.008)}{0.206}$ & $\underset{(0.016)}{1.142}$ & $\underset{(0.001)}{-0.300}$  && $\underset{(0.000)}{0.003}$ & $\underset{(0.009)}{1.119}$ & $\underset{(0.001)}{-0.417}$ && $\underset{(0.002)}{0.083}$ & $\underset{(0.013)}{1.119}$ & $\underset{(0.001)}{-0.371}$ \\
SOL & 693,803 & 124.5 & 110.7 & 29.8 && $\underset{(0.060)}{1.079}$ & $\underset{(0.006)}{1.016}$ & $\underset{(0.001)}{-0.345}$ && $\underset{(0.008)}{0.241}$ & $\underset{(0.013)}{1.015}$ & $\underset{(0.001)}{-0.256}$  && $\underset{(0.001)}{0.020}$ & $\underset{(0.015)}{1.353}$ & $\underset{(0.003)}{-0.386}$ && $\underset{(0.003)}{0.086}$ & $\underset{(0.023)}{1.353}$ & $\underset{(0.001)}{-0.329}$ \\
XRP & 466,774 & 185.1 & 82.9 & 24.1 && $\underset{(0.050)}{1.167}$ & $\underset{(0.007)}{1.154}$ & $\underset{(0.001)}{-0.355}$ && $\underset{(0.010)}{0.262}$ & $\underset{(0.016)}{1.153}$ & $\underset{(0.002)}{-0.266}$  && $\underset{(0.004)}{0.060}$ & $\underset{(0.008)}{1.153}$ & $\underset{(0.002)}{-0.385}$ && $\underset{(0.004)}{0.116}$ & $\underset{(0.016)}{1.152}$ & $\underset{(0.001)}{-0.323}$ \\
BNB & 438,524 & 197.0 & 69.8 & 22.8 && $\underset{(0.060)}{0.630}$ & $\underset{(0.007)}{1.093}$ & $\underset{(0.001)}{-0.362}$ && $\underset{(0.007)}{0.214}$ & $\underset{(0.014)}{1.092}$ & $\underset{(0.001)}{-0.279}$  && $\underset{(0.000)}{0.007}$ & $\underset{(0.009)}{1.149}$ & $\underset{(0.001)}{-0.409}$ && $\underset{(0.003)}{0.088}$ & $\underset{(0.014)}{1.148}$ & $\underset{(0.001)}{-0.358}$ \\
\hline \hline
\end{tabular}
\end{small}
\smallskip
\begin{scriptsize}
\parbox{\textwidth}{\emph{Note}. In the left-hand side of this table, we show descriptive statistics of the cryptocurrency high-frequency data. ``Ticker'' is the short name of the exchange rate. $N$ is the average daily number of transactions, while $\Delta_{N}^{ \text{ms}}$ in the associated intertrade duration (in milliseconds). $\widebar{RV}$ is the average two-hour RV (converted to an annualized standard deviation), while $\widebar{TV}$ is the average two-hour volume (measured in million USDT). In the right-hand side of the table, we report parameter estimates of the fOU process and the Cauchy class (standard errors are placed underneath in parentheses) for each two-hour log-RV and log-trading volume series.}
\end{scriptsize}
\end{center}
\end{sidewaystable}

In the left-hand side of Table \ref{table:binance-descriptive}, we provide a summary of the amount of transaction data available. The selected pairs are vastly liquid. For example, BTC averages 2,044,826 transactions per day (median of 1,257,976), equivalent to an intertrade duration of 42.3 milliseconds (median of 68.7). Meanwhile, the least liquid pair BNB still has more than five trades per second.\footnote{As a comparison, from January 1, 2019 to October 14, 2021 the average daily number of transactions in the front contract of the E-mini S\&P 500 futures (ES) was 318,157 (median of 267,100 and maximum of 1,604,309). The dollar volume in the ES was about 2.5 times larger than in the BTC.}

We follow \citet{wang-xiao-yu:23a, shi-yu-zhang:24a} and apply our composite likelihood framework to both the volatility and trading volume of these cryptocurrencies. While the dynamic of volatility is of primary interest, in practice the spot variance is latent, and it needs to be inferred from the discretely observed high-frequency data. This may distort the subsequent analysis, as further discussed below. By contrast, the trading volume is immediately accessible, and it is known to be strongly correlated with volatility. The evidence from the latter further buttresses our findings for the estimated spot variance.

\clearpage

\subsection{Analysis of spot variance}

In this subsection, we follow the simulation section and let $Y$ be the log-variance of the log-spot exchange rate process, meaning that $Y_{t} = \log \sigma_{t}^{2}$. In practice, this variable is unobserved, and to recover it, we employ a standard RV calculated over short non-overlapping time intervals, see, e.g., \citet{andersen-bollerslev:98a} and \citet{barndorff-nielsen-shephard:02a}.

Suppose there are $m$ equidistant increments of a log-price in the interval $[t, t + \Delta]$. Then, the RV over this interval is defined as
\begin{equation}
RV_{t,t+ \Delta}^{m} = \sum_{j=1}^{m} ( \log P_{t + \frac{j}{m} \Delta} - \log P_{t + \frac{j-1}{m} \Delta})^{2},
\end{equation}
for $t = \Delta, 2 \Delta, \dots, n \Delta$.

In theory, RV is able to estimate the latent spot, or point-in-time, variance. This follows from a double-asymptotic framework in which an increasing number of log-price increments are observed over a shrinking time interval, i.e. $\Delta \rightarrow 0$ and $m \rightarrow \infty$ such that $\Delta m \rightarrow \infty$. In this case, under weak regularity conditions, $RV_{t,t+ \Delta}^{m} \overset{ \mathbb{P}}{ \longrightarrow} \sigma_{t}^{2}$ \citep[see][]{jacod-protter:12a}. By the continuous mapping theorem, $\log RV_{t,t+ \Delta}^{m} \overset{ \mathbb{P}}{ \longrightarrow} \log \sigma_{t}^{2} = Y_{t}$. In fact, RV is the MLE in a parametric model with no drift and constant volatility, so by the invariance principle, the logarithm of RV is the MLE of the log-spot variance. In this sense, RV is the best possible proxy available.

In practice, however, RV contains at least two important sources of error. First, $\Delta$ is normally fixed and RV is instead consistent for the integrated variance (IV), $RV_{t,t+ \Delta}^{m} \overset{ \mathbb{P}}{ \longrightarrow} IV_{t,t+ \Delta} = \int_{t}^{t + \Delta} \sigma_{s}^{2} \mathrm{d}s$ as $m \rightarrow \infty$. The integration introduces a smoothing effect, which tends to bias estimators of the Hurst parameter upward. Second, with real data it is ill-advised to sample too often, because of microstructure noise \citep{hansen-lunde:06b}, which puts an upper bound on $m$. While RV converges in probability to IV in the infill asymptotic limit, $\Delta$ fixed and $m \rightarrow \infty$, with $m$ finite it contains an estimation error. Thus, the blockwise RV fluctuates around the true value of IV, which tends to bias estimators of the Hurst parameter downward.\footnote{This has been dubbed elusive roughness in \citet{fukasawa-takabatake-westphal:22a}.} Although these effects may partially cancel out, they can severely impair the estimation.

Even though the bid-ask spread in the included coins is routinely less than a basis point, the presence of other market frictions, such as price discreteness and bid-ask bounce, can evidently instill a nontrivial bias in RV if we insist on sampling the spot exchange rate process at the maximum available tick-by-tick level in order to mitigate the above concerns. To circumvent this problem and strike a balance, we construct a 15-second equidistant log-price series, corresponding to 5,760 daily log-price increments. This choice was guided by a volatility signature plot, which is a diagnostic tool designed to gauge the impact of microstructure noise on RV  \citep[see, e.g.,][]{hansen-lunde:06b}. It computes the sample average daily RV as a function of the sampling frequency:
\begin{equation}
\overline{RV}(m) = T^{-1} \sum_{t=1}^{T} RV_{t-1,t}^{m},
\end{equation}
where $RV_{t-1,t}^{m}$ is the RV over the whole day $t$ sampled with a time gap $1/m$.

In general, the behavior of a volatility signature plot depends on the properties of the noise. If it is absent (or negligible), $\overline{RV}(m)$ is an estimator of the average IV, so the graph of the function should be a flat line up to sampling variation. A noise that is independent of the efficient price causes an upward bias in the estimate and, if the noise is itself i.i.d., $\overline{RV}(m) \overset{ \mathbb{P}}{ \longrightarrow} \infty$ as $m \rightarrow \infty$.

\begin{figure}[ht!]
\begin{center}
\caption{Microstructure noise and intraday variation in volatility.}
\label{figure:noise}
\begin{tabular}{cc}
\small{Panel A: Volatility signature.} & \small{Panel B: Diurnal variation.} \\
\includegraphics[height=8.00cm,width=0.48\textwidth]{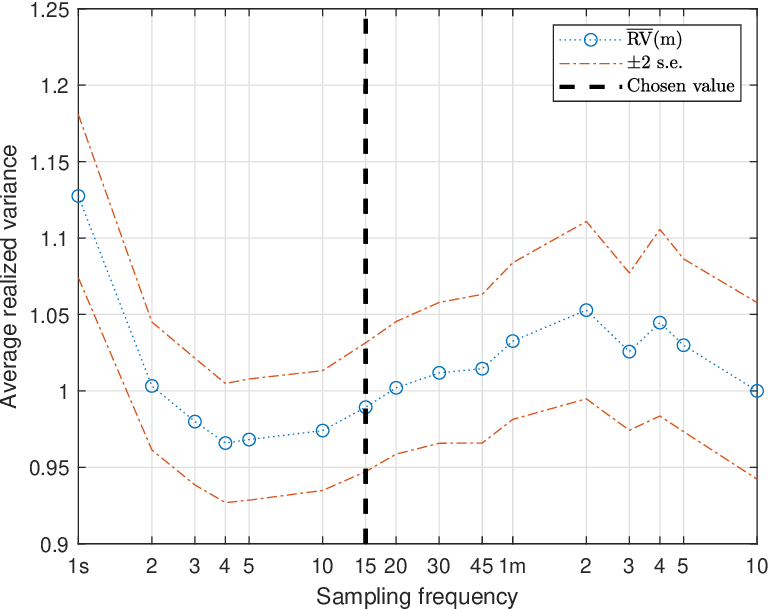} &
\includegraphics[height=8.00cm,width=0.48\textwidth]{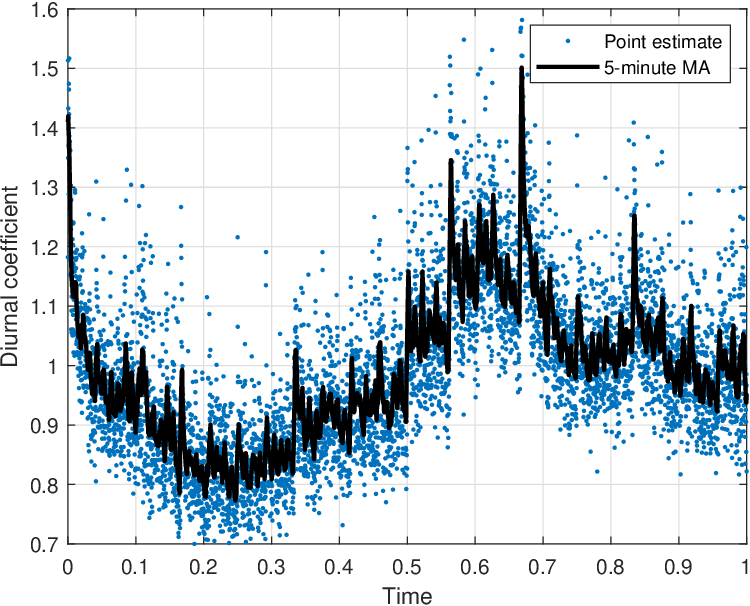} \\
\end{tabular}
\begin{scriptsize}
\parbox{\textwidth}{\emph{Note.} In Panel A, we present a volatility signature plot of BTC, i.e. the sample average of daily RV as a function of the sampling frequency ranging from 10-minute (slowest) to one-second (fastest). Each observation is scaled by the 10-minute value. We show a two standard error band around the point estimate. In Panel B, we plot a nonparametric estimator of the intraday periodicity in the associated 15-second squared log-return process. The line is a 5-minute moving average. The sample runs from January 1, 2019 to March 27, 2025. Time 0 is midnight in Coordinated Universal Time (UTC).}
\end{scriptsize}
\end{center}
\end{figure}

In Panel A of Figure \ref{figure:noise}, we present the volatility signature plot of BTC, which is the coin used throughout as an illustration. The charts for the other cryptocurrencies are contained in Appendix \ref{appendix:noise}. We vary $m$ from 144 to 86,400, which correspond to 10-minute sampling as the slowest to 1-second sampling as the fastest. Each $\overline{RV}(m)$ has been scaled by the 10-minute estimate such that, as $m$ increases, the evolution of the graph can be interpreted as the relative bias induced by the accumulation of microstructure noise. To get a crude impression of the sampling variation, we surround each point estimate by a two standard error confidence band. Overall, there is little evidence to suggest that microstructure noise wreaks havoc on the RV of BTC with 15-second---or even much faster---sampling of the price process. This impression is corroborated by the remaining figures in the appendix. The exception is SOL, where we do notice an uptick in $\overline{RV}(m)$ at the 15-second horizon, and the estimate is slightly outside the 10-minute confidence interval, so we stop there.\footnote{As an extra robustness check, we computed the Hausman test for the presence of microstructure noise in high-frequency data developed by \citet{ait-sahalia-xiu:19a}. A caveat with their theoretical framework is that the volatility process is assumed to be an It\^{o} semimartingale, so this piece of analysis should be taken with a grain of salt. Nonetheless, when evaluating their stochastic volatility- and price jump-robust test statistic at the 1\% significance level, the null hypothesis of no microstructure noise is rejected only 2.19\% of the days in the BTC sample at the 15-second frequency. This is rather close to the nominal level, especially since their test statistic was found to be slightly oversized in finite samples.}

To remove traces of noise even at the 15-second horizon, we follow \citet{jacod-li-mykland-podolskij-vetter:09a} by pre-averaging transaction prices that are closest to each time point on the sampling grid.

In the Panel B of Figure \ref{figure:noise}, we plot an estimate of the intraday periodicity in the second moment of the BTC log-return series.\footnote{In general, the intraday volatility of the various exchange rates exhibit near-identical behavior, indicative of a strong commonality.} We employ a nonparametric estimator that averages the time-of-the-day squared 15-second log-returns (a proxy for the average point-in-time intraday variance) and normalizes the sum of these numbers to one to form a seasonality estimate. As readily seen, there is a discernible nonstationary component in the intraday evolution of the volatility of BTC, which features a complex behavior. This clashes with our assumed stationarity in \eqref{equation:Y}, which can distort the estimation of the stationary Gaussian processes. To avoid this, we pre-filter the 15-second log-return series with the estimated diurnal coefficient to ensure that the rescaled time series is closer to stationary. In addition to this, we observe a marked day-of-the-week effect in the variance process, which we too eliminate.

We next construct a non-overlapping two-hour RV for the corrected 15-second log-returns.\footnote{To ensure that price jumps do not interfere with our estimation results, we also zero out too large absolute log-returns with the truncation approach of \citet{mancini:09a}.} Although the block size is a bit arbitrary, it is an attempt to balance the inherent errors in the RV measure, as discussed above. The concrete choice follows extant literature on spot volatility estimation \citep*[see, e.g.,][]{christensen-thyrsgaard-veliyev:19a}. The typical value of RV across currencies at this sampling frequency is shown in Table \ref{table:binance-descriptive}. The estimates imply that the diffusive variance is rather high on average and exceeds the overall level of the total quadratic return variation often seen in even very risky individual stocks.

\begin{figure}[ht!]
\begin{center}
\caption{Distribution and persistence of log-RV.}
\label{figure:distribution}
\begin{tabular}{cc}
\small{Panel A: Kernel density estimate.} & \small{Panel B: Sample ACF.} \\
\includegraphics[height=8.00cm,width=0.48\textwidth]{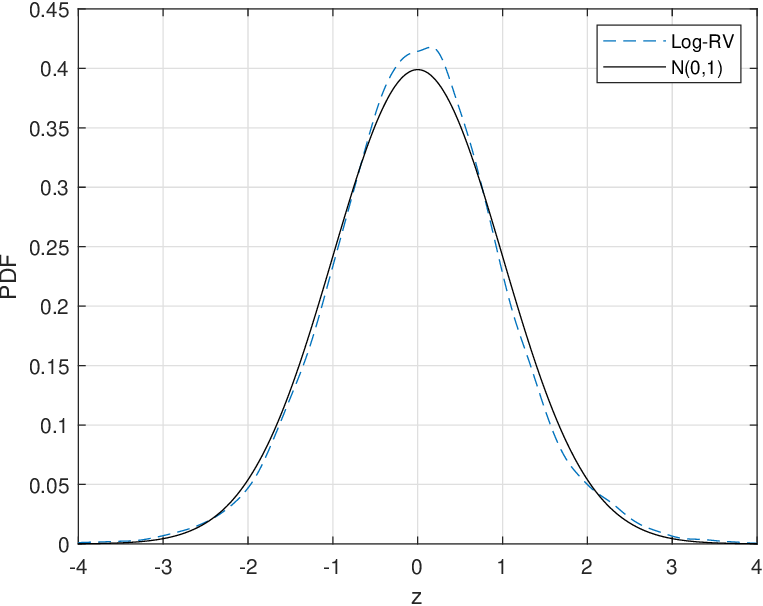} &
\includegraphics[height=8.00cm,width=0.48\textwidth]{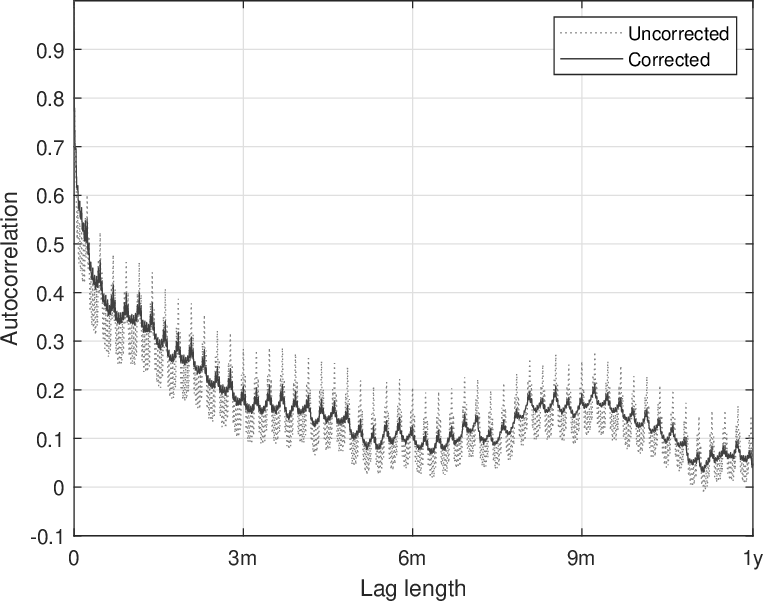} \\
\end{tabular}
\begin{scriptsize}
\parbox{\textwidth}{\emph{Note.} In Panel A, we present a kernel density estimate of the two-hour log-RV. A standard normal distribution is superimposed as a benchmark. In Panel B, we plot the sample ACF of this measure. ``Corrected'' is the version that is purged from intraday periodicity and a day-of-the-week effect in volatility. As a comparison, we also show the ACF of the raw ``Uncorrected'' RV.}
\end{scriptsize}
\end{center}
\end{figure}

How close to Gaussian is the two-hour log-RV? To get a crude impression, we center and scale the time series by its sample average and sample standard deviation. In Panel A of Figure \ref{figure:distribution}, we plot the histogram of the standardized statistic. The distribution of spot log-RV is very close to the bell-shaped curve of the Gaussian density function, which is consistent with previous literature \citep[e.g.][]{andersen-bollerslev-diebold-ebens:01a, andersen-bollerslev-diebold-labys:03a}.

We next examine the correlation structure of the process. In Panel B of Figure \ref{figure:distribution}, we present a correlogram of the spot log-RV. The lag length goes up to a full year. Here, ``Uncorrected'' represents the raw estimator, while ``Corrected'' is based on the version that accounts for intraday periodicity and the day-of-the-week effect. Both point toward a strong persistence in the extracted log-spot variance process, which---as noted above---tends to understate the population autocorrelation in the presence of long memory. This effect is arguably excerbated here, because we are working with a noisy proxy of the volatility \citep[e.g.,][]{hansen-lunde:14a}. It is evidently hard to say what type of model is ``best'' based on the graph alone, but the ability of being able to disentangle the short- and long-end of the memory lane does appear useful in order to fit the initial steep decline of the curve combined with the slow leveling off.

The two stationary Gaussian processes from Section \ref{section:example} are estimated following the design from Section \ref{section:simulation}. We input the corrected log-RV into the numerical optimization procedure. To get a gauge at the magnitude of the standard errors inherent in the parameter estimates, we use the limiting covariance expressions from Theorem \ref{theorem:clt}. Here, we plug our estimates into the limiting covariance expression, where we numerically compute the derivatives of the autocovariance terms. The results are presented in the right-hand side of Table \ref{table:binance-descriptive}. We only include the parameters relating to the ACF, since the mean is always estimated in line with the sample average of log-RV (with a weighted average of log-RV, as explained in Footnote \ref{footnote:mean}).

We observe that the roughness index of spot log-RV is located around $-0.35$ for the fOU process, which we recall is related to the Hurst exponent via the relation $\alpha = H - 1/2$. Hence, as consistent with the recent literature on rough volatility, we uncover that there is overwhelming evidence in support of the variance process being more erratic than a sBm at short time scales \citep[see, e.g.][]{bolko-christensen-pakkanen-veliyev:23a, fukasawa-takabatake-westphal:22a, gatheral-jaisson-rosenbaum:18a, wang-xiao-yu:23a}. The $\kappa$ parameter is very close to zero, so the process exhibits near-unit root behavior.

Turning next to the Cauchy class, the results are striking. On the one hand, $\alpha$ is now estimated around $-0.27$. While the values are higher than for the fOU process, they are still located far in roughness space. On the other hand, $\beta$ is estimated far in the long memory range around 0.20 -- 0.25. Noting that the link between $\beta$ and the fractional differencing parameter $d$ in a discrete-time ARFIMA($p$,$d$,$q$) process is $d = (1 - \beta)/2$ \citep[see][equation 13.2.1]{brockwell-davis:91a}, this corresponds to a $d$ close to 0.40 on average. This agrees with early evidence from the literature on modeling and forecasting log-RV as a fractionally integrated process, see, e.g., \citet{andersen-bollerslev-diebold-ebens:01a, andersen-bollerslev-diebold-labys:03a}.

\begin{figure}[ht!]
\begin{center}
\caption{Heatmap of composite likelihood function.}
\label{figure:heatmap}
\begin{tabular}{cc}
\small{Panel A: fOU.} & \small{Panel B: Cauchy.} \\
\includegraphics[height=8.00cm,width=0.48\textwidth]{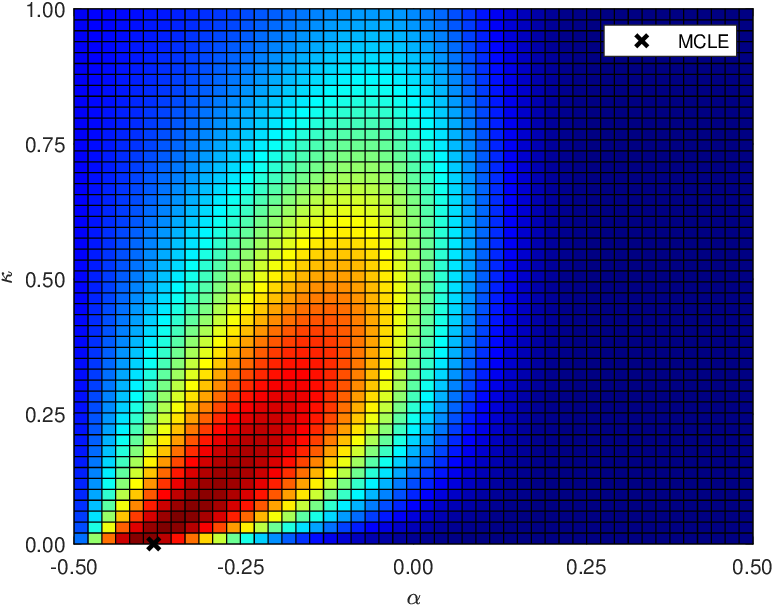} &
\includegraphics[height=8.00cm,width=0.48\textwidth]{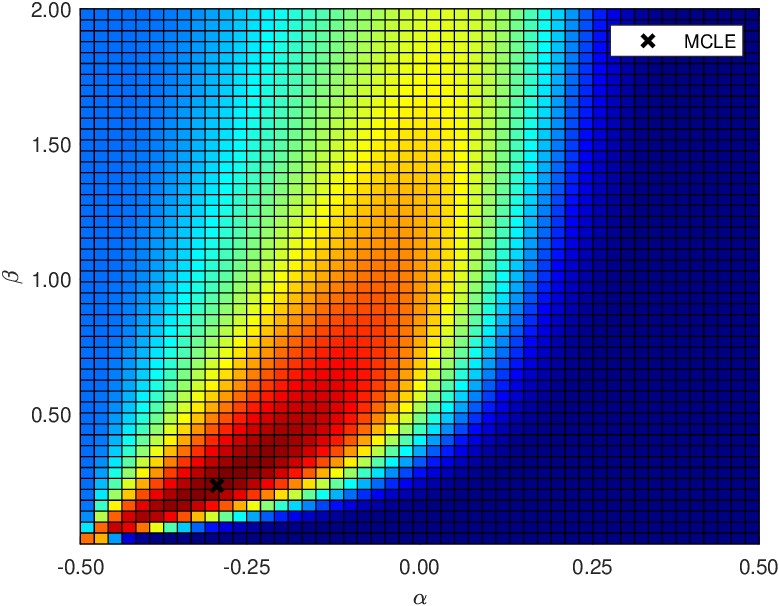} \\
\end{tabular}
\begin{scriptsize}
\parbox{\textwidth}{\emph{Note.} We create a heatmap of the composite likelihoood function for the fOU process (in Panel A) and Cauchy class (in Panel B). A cooler/warmer temperature (blue/red) indicates that the likelihood is lower/higher. We fix the mean and standard deviation ($\mu$ and $\nu$) at the sample average and sample standard deviation of the log-RV and vary the remaining free parameters over a broad range of the parameter space. The MCLE is reported with a black cross.}
\end{scriptsize}
\end{center}
\end{figure}

In Figure \ref{figure:heatmap}, we construct a heatmap of the composite likelihood function for the log-RV of the BTC spot exchange rate. Panel A is for the fOU process, while Panel B is for the Cauchy class. To reduce the plot to a two-dimensional plane, we fix $\mu$ and $\nu$ at the sample average and sample standard deviation of the log-RV, while varying $\kappa$ and $\beta$ over a broad range of the parameter space and $\alpha$ through the permissible area. The color code indicates the height of the likelihood function with higher (lower) contour values being indicated as a warmer (cooler) temperature.\footnote{We should emphasize that this piece of analysis is not readily available for the MME approach.}

As consistent with the parameter estimates of BTC from Table \ref{table:binance-descriptive}, Panel A of Figure \ref{figure:heatmap} indicates that the composite likelihood of the fOU process prefers a rough dynamic with a Hurst exponent less than a half. The optimal value of the mean reversion hovers close to the boundary of the parameter space at zero, where the fOU process becomes a scaled version of the non-stationary fBm. \citet{shi-yu:23a} look at the discrete-time ARFIMA($1$,$d$,$0$) model: $(1- \phi L) (Y_{t}- \mu) = \sigma_{ \epsilon} (1-L)^{-d} \epsilon_{t}$, where $L$ is the lag operator and $\epsilon_{t} \sim D(0,1)$ is a white noise. They note that $\phi = 0$ and $d = 0.5$ is observationally equivalent to $\phi = 1$ and $d = -0.5$, leading to identification failure, see also \citet{liu-shi-yu:20a} and \cite{li-phillips-shi-yu:25a}. This can be interpreted as either fast mean reversion and long memory or no mean reversion and anti-persistence. This agrees with a first impression of the likelihood function: the fOU model is prepared to trade roughness for additional mean reversion with only a minimal drop in likelihood and the surface appears to exhibit a plateau over an extended region in that direction. This is a symptom that the model is scrambling to fit both ends of the ACF.

To get a gauge at the exact curvature, in Figure \ref{figure:profile-likelihood} we include the profile composite likelihood function of the fOU process in Panel A, $CL_{ \alpha}( \kappa) = \max_{ \alpha} CL( \alpha, \kappa)$ and the Cauchy class in Panel B, $CL_{ \alpha}( \beta) = \max_{ \alpha} CL( \alpha, \beta)$. We again set $\mu$ and $\nu$ equal to the sample average and sample standard deviation of the log-RV and suppress their influence on the composite likelihood. We then plot this as a function of $\kappa$ or $\beta$ after maximizing with respect to $\alpha$, where the optimal value of the latter is shown on the right-hand $y$-axis. To make the graphs comparable, we normalize the profile likelihood to reach a maximum of one. Overall, the evolution of the fOU profile likelihood in Panel A shows that it is, in fact, monotonically increasing toward the boundary of the parameter space, whereas the Cauchy class has a maximum in the interior. It is interesting to observe that as $\kappa$ increases for the fOU, the optimal value of $\alpha$ does not venture into the long memory region but levels out around zero, as implied by a sBm.

\begin{figure}[ht!]
\begin{center}
\caption{Profile composite likelihood function.}
\label{figure:profile-likelihood}
\begin{tabular}{cc}
\small{Panel A: fOU.} & \small{Panel B: Cauchy.} \\
\includegraphics[height=8.00cm,width=0.48\textwidth]{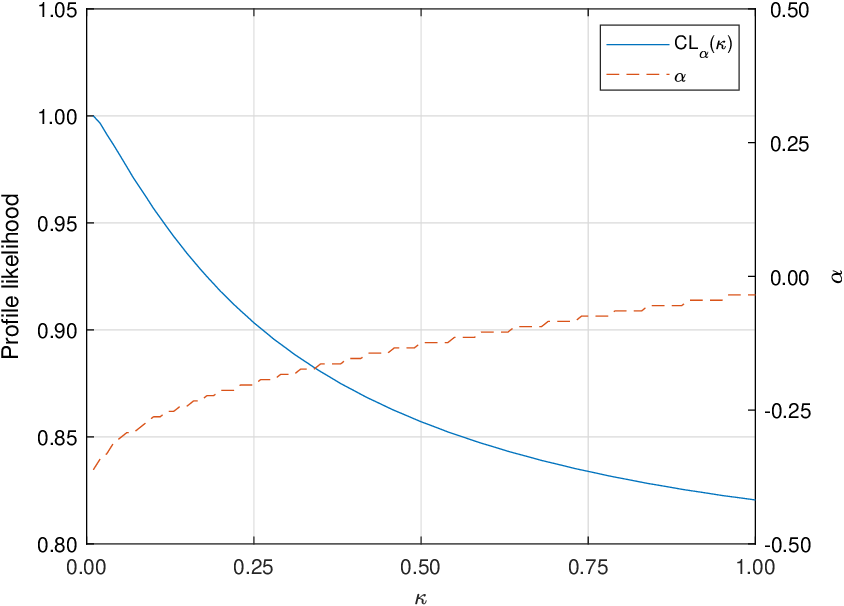} &
\includegraphics[height=8.00cm,width=0.48\textwidth]{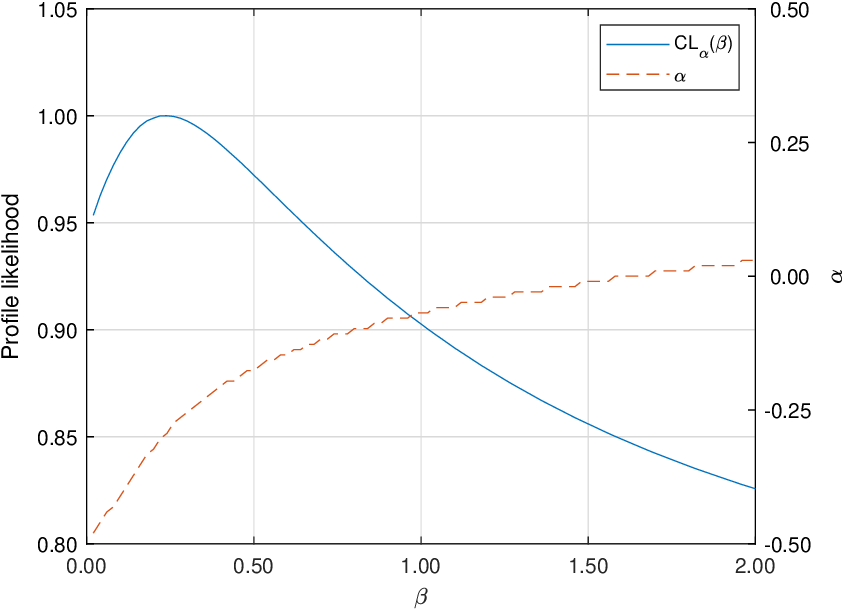} \\
\end{tabular}
\begin{scriptsize}
\parbox{\textwidth}{\emph{Note.} We show the profile composite likelihood function of the fOU process in Panel A, $CL_{ \alpha}( \kappa) = \max_{ \alpha} CL( \alpha, \kappa)$ and the Cauchy class in Panel B, $CL_{ \alpha}( \beta) = \max_{ \alpha} CL( \alpha, \beta)$. We fix $\mu$ and $\nu$ at the sample average and sample standard deviation of the log-RV. We then plot the composite likelihood as a function of $\kappa$ or $\beta$ after maximizing with respect to $\alpha$, where the optimal value of the latter is shown on the right-hand $y$-axis.}
\end{scriptsize}
\end{center}
\end{figure}

At the end of the day, the fOU process has to make a decision: \textit{either} roughness \textit{or} long memory. It selects roughness to capture the initial decline in the ACF and to fit the subsequent slow decay it forces the mean reversion parameter near zero to approximate long-range dependence. However, the Cauchy class does not suffer from this shortcoming and is able to separate these effects to demonstrate that both are required to describe the data. The $( \alpha, \beta)$ estimate is within the interior of the parameter space and appears to be identified with a high degree of accuracy, gauging from the heatmap, profile likelihood, and the magnitude of the standard errors.

In sum, our empirical results from the Cauchy class of stationary Gaussian processes translate to a non-integrable ACF with a sharp initial decline, i.e. roughness, followed by a slow decay toward zero, i.e. long memory. Our results thus provide compelling evidence that spot variance exhibits \textit{both} roughness \textit{and} long memory.

\subsection{Analysis of trading volume}

There is abundant empirical evidence that the trading volume of an asset is closely connected to its return volatility \citep[e.g.,][]{epps-epps:76a,hasbrouck:91a}. At a theoretical foundation, this observation is embedded in the so-called mixture of distribution hypothesis \citep[e.g.,][]{clark:73a,tauchen-pitts:83a}, where trading volume and volatility are subordinated by a ``news'' process. A positive correlation between them is also a common prediction in market microstructure theory \citep[e.g.,][]{glosten-milgrom:85a,easley-ohara:92a}. In contrast to spot volatility, though, trading volume is observed.

In this section, we apply our MCLE procedure to the trading volume (measured in USD) of the five cryptocurrencies with the aim of establishing whether roughness and long memory is also present in this variable. We follow the implementation from above and first remove intraday periodicity in the 15-second spot trading volume, which is nearly identical to that observed in the 15-second squared log-return process, adding further support to their strong relationship. We cumulate trading volume over each non-overlapping two-hour block matching the RV estimation and log-transform the end product. However, we observe a slight uptick in the log-trading volume over time, which we remove by fitting a regression with a deterministic linear trend to the series.

The right-hand side of Table \ref{table:binance-descriptive} reports the estimated parameter vector of the fOU process and Cauchy class based on the de-trended and de-seasonalized log-trading volume. The results are more or less in tandem with those obtained for log-RV. In particular, there is strong evidence of roughness in log-trading volume with the roughness index being estimated around $-0.35$ to $-0.40$ for both models. Moreover, the mean reversion parameter in the fOU model is close to the zero lower bound, again indicative of near-unit root behavior. Furthermore, the memory parameter of the Cauchy class is assessed to be small, pointing to long memory, and is, in fact, even lower than for the log-RV, but it remains squarely within the interior of the parameter space.

The analysis of two-hour log-RV sought to balance the ``smoothing'' and ``roughing'' effect in the volatility proxy, while dodging microstructure noise with 15-second sampling of the log-price process. We do not encounter this trade-off with the trading volume, so we can sample it at a faster pace, which can help to gauge whether our conclusion about roughness and long memory is an artifact of temporal aggregation. To shed light on this, we sample trading volume starting at the two-hour horizon and moving down to the 15-second window, and we repeat the MCLE procedure at every step. In the last instance, the optimizer is confronted with $5{,}760 \times 2{,}278 = 13{,}121{,}280$ observations, but it nonetheless manages to converge in a few minutes.

In Figure \ref{figure:temporal-aggregation}, we present the estimated roughness index and memory parameter ($\kappa$ for the fOU process and $\beta$ for the Cauchy class) as a function of the sampling frequency. Several noteworthy observations emerge. First, the parameter estimates of the five included cryptocurrencies exhibit near-identical behavior, consistent with previous ``universality'' evidence from other asset classes \citep[e.g.,][]{bennedsen-lunde-pakkanen:22a, rosenbaum-zhang:24a}. Second, the estimated roughness index is slightly decreasing, as we sample more often, and it ends at around $\widehat \alpha \approx -0.47$ for the fOU and $\widehat \alpha \approx -0.40$ for the Cauchy class. This indicates an extreme degree of roughness over short time intervals, echoing the empirical evidence from volatility in \cite{bolko-christensen-pakkanen-veliyev:23a}. Third, the estimated $\kappa$ and $\beta$ are rather robust to the choice of time aggregation.

In summary, our parameter estimates based on the log-trading volume corroborate the previous application to the log-RV measure.

\begin{figure}[ht!]
\begin{center}
\caption{Parameter estimates for log-trading volume.}
\label{figure:temporal-aggregation}
\begin{tabular}{cc}
\small{Panel A: fOU.} & \small{Panel B: fOU.} \\
\includegraphics[height=8.00cm,width=0.48\textwidth]{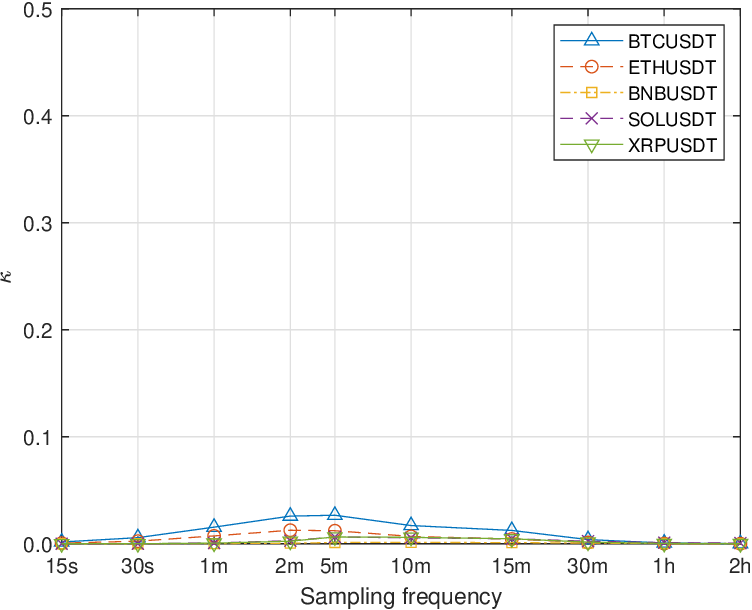} &
\includegraphics[height=8.00cm,width=0.48\textwidth]{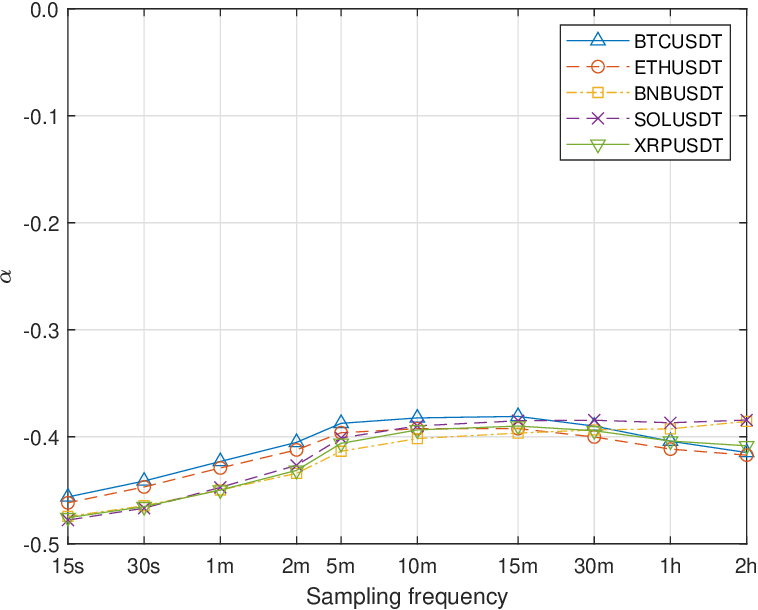} \\
\small{Panel C: Cauchy.} & \small{Panel D: Cauchy.} \\
\includegraphics[height=8.00cm,width=0.48\textwidth]{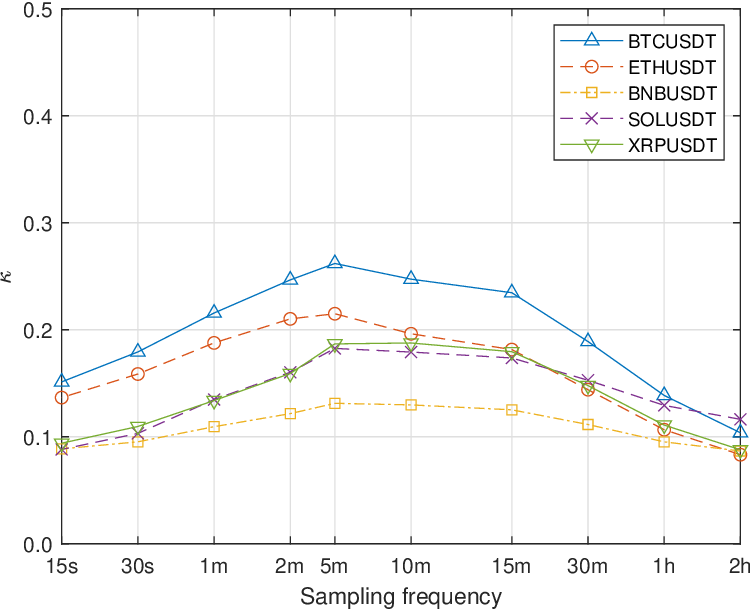} &
\includegraphics[height=8.00cm,width=0.48\textwidth]{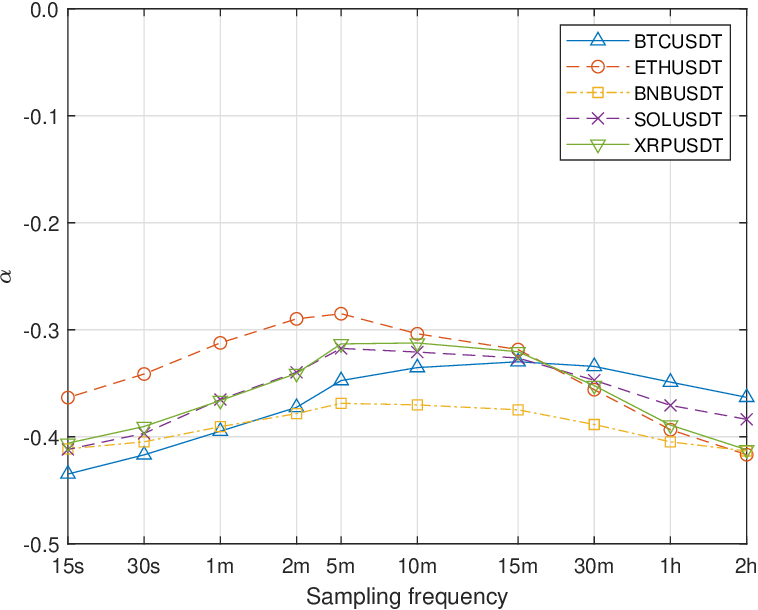} \\
\end{tabular}
\begin{scriptsize}
\parbox{\textwidth}{\emph{Note.} We show the parameter estimates of $\kappa$ (Panel A) and $\alpha$ (Panel B) for the fOU process and $\beta$ (Panel C) and $\alpha$ (Panel D) for the Cauchy class based on log-trading volume. They are plotted as a function of the sampling frequency going from a two-hour down to a 15-second horizon. Note that the $x$-axis is on a logarithmic scale.}
\end{scriptsize}
\end{center}
\end{figure}

\section{Conclusion} \label{section:conclusion}

In this paper, we develop a framework for doing composite likelihood estimation of a parametric class of stationary Gaussian processes. We derive the consistency and asymptotic distribution of the maximum composite likelihood estimator (MCLE) under various degrees of persistence.

As an adaptation, we back out the parameters of a fractional Ornstein-Uhlenbeck (fOU) process. In comparison to the method-of-moments estimator (MME) from \citet{wang-xiao-yu:23a}, the composite likelihood estimator has the advantage that the full parameter vector is estimated in a single loop rather than via a two-step approach. Moreover, in a simulation study we demonstrate superior finite sample properties of the MCLE over the MME. This is possibly because the MME requires in-fill asymptotic theory, which may not kick in fast enough with realistic choices of sampling frequency.

In our empirical implementation of the fOU process, we discover that the log-variance of the large-cap segment of the cryptocurrency spot exchange rate market exhibits roughness. On the one hand, this agrees with \citet{gatheral-jaisson-rosenbaum:18a} and the universality of the volatility formation process, see \citet{rosenbaum-zhang:24a}. On the other hand, it stands in contrast to the long memory version of the fOU process, proposed as a model for stochastic volatility in \citet{comte-renault:98a}, and the abundant literature on modeling and forecasting RV as a long-range dependent variable, e.g. \citet{andersen-bollerslev-diebold-labys:03a} and \citet{corsi:09a}. This ``volatility puzzle'' is scrutinized further in \citet{shi-yu:23a}, who note that the likelihood function of a discrete-time version of the fOU is bimodal, leading to near-observational equivalence between these two opposing regimes. \citet{li-phillips-shi-yu:25a} explore the issue in terms of weak identification and the empirical evidence of their identification robust inference is rooting for long memory.

In our opinion, the source of the dispute is that the fOU specification is not flexible enough. It has to control both the short- and long-term persistence in a single parameter---the Hurst exponent. As it cannot accommodate both, it therefore needs to make a stance: \textit{either} roughness \textit{or} long memory.

Set against this backdrop, a logical step is to divert more attention toward models that decouple the short- and long-run persistence. There are numerous ways one can do that. A standard approach in the SV literature is to superposition independent driving factors \citep[see, e.g.,][]{alizadeh-brandt-diebold:02a, barndorff-nielsen-shephard:02a}. However, such multi-factor models usually exhibit exponentially decaying ACF and, hence, at best approximate roughness and long memory. An intriguing avenue for further research is to superposition two independent fOU processes and check whether the empirical estimates of the Hurst indexes are located on either side of the memory spectrum. This idea is being pursued in \citet{christensen-espinosa-pakkanen-veliyev:25a}, where the GMM approach of \cite{bolko-christensen-pakkanen-veliyev:23a} is extended to a two-factor setting.

\citet{gneiting-schlather:04a} proposed a so-called Cauchy class of Gaussian processes. As the fOU, it has two parameters to fit the ACF, but---opposite the fOU---it reserves one to fit the short-run persistence and another one to fit the long-run persistence. \citet{bennedsen-lunde-pakkanen:22a} implement this model with MME on a vast cross-sectional set of empirical high-frequency equity data and find that roughness and long memory are embedded in the log-variance process. In the present paper, we also adopt this strategy to explore the issue based on transaction price data from the cryptocurrency market. Interestingly, our empirical MCLE of the Cauchy class also point toward \textit{both} roughness \textit{and} long memory being present in the log-spot variance. Thus, our results can help to reconcile the conflicting evidence in the previous literature.

At the moment, our procedure suffers from the fact that spot volatility is latent. The sampling variation inherent in estimators, such as RV, impairs the MCLE and causes a potentially severe degree of spurious roughness. \citet{bolko-christensen-pakkanen-veliyev:23a} account for this in their GMM estimation of the fOU model. In future research, we should attempt to do this as well. In lieu of this, we here include an application to the log-trading volume, which confirms our findings for the log-variance.

\pagebreak

\appendix

\section{Mathematical proofs} \label{appendix:proof}

\begin{proof}[Proof of Lemma \ref{lemma:identification}]

In the proof, we look at a single tuple, $k$, of length $q$ only. The general case follows by extension of this result to the sum defining the log-composite likelihood. As explained in the proof of Theorem \ref{theorem:lln}, a law of large number implies that, uniformly in $\theta$,
\begin{equation*}
\frac{1}{n} cl( \theta;Y) \overset{ \mathbb{P}}{ \longrightarrow} \mathbb{E}_{ \theta_{0}} \left( \log f_{k}(Y^{k}; \theta) \right),
\end{equation*}
as $n \rightarrow \infty$, where $Y^{k}$ denotes the selected $k$-tuple of the random vector $Y$.

Now, recall the information inequality:
\begin{equation*}
\mathbb{E}_{ \theta_{0}} \left[ \log \frac{f_{k}(Y^{k}; \theta)}{f_{k}(Y^{k}; \theta_{0})} \right] \leq \log \mathbb{E}_{ \theta_{0}} \left[ \frac{f_{k}(Y^{k}; \theta)}{f_{k}(Y^{k}; \theta_{0})} \right] = 0,
\end{equation*}
by Jensen's inequality. This implies that
\begin{equation*}
\mathbb{E}_{ \theta_{0}} \left[ \log f_{k}(Y^{k}; \theta) \right] \leq \mathbb{E}_{ \theta_{0}} \left[ \log f_{k}(Y^{k}; \theta_{0}) \right].
\end{equation*}
As the natural logarithm is strictly concave, the above inequality is also strict if $\mathbb{P} \big( \log f_{k}( Y^{k}; \theta) = \log f_{k}( Y^{k} ; \theta_{0}) \big) \neq 1$. Recall now that
\begin{equation*}
\log f_k(y;\theta) = C - \frac{1}{2}\log \vert \Sigma(\theta) \vert - \frac{1}{2}y^{\top} \Sigma^{-1}(\theta) y.
\end{equation*}
This expression depends on $\theta$ only through $\gamma_{h}( \theta)$, and viewed as a function of $y$ it is a second-order polynomial with an intercept. As our assumption implies that for $\theta \neq \theta_{0}$ there exists at least one $h$ such that $\gamma_{h}( \theta_{0}) \neq \gamma_{h}( \theta)$, it follows that $\Sigma( \theta) \neq \Sigma( \theta_{0})$, and, therefore, $\Sigma( \theta)^{-1} \neq \Sigma( \theta_{0})^{-1}$ by uniqueness of the inverse. The proof is then completed by noting that a quadratic function can agree for at most finitely many (in fact, two) values of $y$.
\end{proof}

\begin{proof}[Proof of Theorem \ref{theorem:lln}]

We observe that a stationary Gaussian process for which $\lim_{h \rightarrow \infty} \rho_{h} = 0$ is strongly mixing and, hence, ergodic \citep[see, e.g.,][Theorem 9(ii)]{maruyama:49a}. Combined with the stationarity assumption, the ergodic theorem implies that the sample average log-composite likelihood converges in probability to its population counterpart:
\begin{equation*}
C_{n}( \theta) \equiv \frac{1}{n} cl( \theta; Y) = \frac{1}{n} \sum_{k_{j} \in Q} \sum_{i=1}^{n-k_{q^{j}}^{j}} \log f_{k^{j}} \big(Y_{i}^{k^{j}}; \theta \big) \overset{ \mathbb{P}}{ \longrightarrow} \mathbb{E} \Bigg( \sum_{k^{j} \in Q} \log f_{k^{j}} \big(Y_{1}^{k^{j}}; \theta \big) \Bigg) \equiv C( \theta),
\end{equation*}
as $n \rightarrow \infty$.

If we let $\theta_{0}$ denote the true data-generating value of $\theta$, the information inequality implies that $C( \theta)$ is uniquely maximized at $\theta_{0}$ \citep[e.g.,][Lemma 2.2]{newey-mcfadden:94a}. This requires $\theta_{0}$ is identified from the density function, which holds by Assumption \ref{assumption:identification}. Theorem 4.1 and 4.3 in \citet{wooldridge:94a} then establishes the uniform weak law of large numbers for the sequence of $q$-wise likelihood functions, from which it follows that $\hat{\theta}_{ \text{MCLE}} \overset{ \mathbb{P}}{ \longrightarrow} \theta_{0}$.
\end{proof}

\begin{proof}[Proof of Theorem \ref{theorem:clt}]

We recall that
\begin{equation*}
cl( \theta;y) = \sum_{k^{j} \in Q} \sum_{i=1}^{n-k_{q^{j}}^{j}} \log f_{k^{j}}(y_{i}^{k^{j}}; \theta),
\end{equation*}
where
\begin{equation*}
f_{k^{j}}(y_{i}^{k^{j}}; \theta) = \frac{1}{ \sqrt{(2 \pi)^{q^{j}} \det \big( \Sigma_{k^{j}}( \theta) \big)}} \exp \left(- \frac{1}{2}(y_{i}^{k^{j}})^{ \top} \Sigma_{k^{j}}^{-1}( \theta)y_{i}^{k^{j}} \right),
\end{equation*}
and $\Sigma_{k^{j}} = \text{cov}(Y_{i}^{k^{j}})$, which is independent of $i$ by stationarity. Here, we again emphasize that the dependence on $\theta$ is solely through $\Sigma( \theta)$.

The idea is now to compute the score of $cl( \theta;y)$ and expand it around $\theta_{0}$. We next derive an expression for the variance of the score to determine the rate of convergence. The limiting distribution is then found with the help of \citet{breuer-major:83a}, \citet{arcones:94a}, \citet{hosking:96a}, and \citet{beran-feng-ghosh-kulik:13a} for the various settings covered in the theorem.

To calculate the score, we write
\begin{equation*}
cl( \theta;y) = \sum_{k^{j} \in Q} \sum_{i=1}^{n-k_{q^{j}}^{j}} \left(C - \frac{1}{2} \log \det \big( \Sigma_{k^{j}}( \theta) \big) - \frac{1}{2}(y_{i}^{k^{j}})^{ \top} \Sigma_{k^{j}}^{-1}( \theta)y_{i}^{k^{j}} \right),
\end{equation*}
where $C$ is a constant.

Now, define
\begin{align} \label{equation:score}
\begin{split}
s_{n}( \theta) &= \frac{ \partial}{ \partial \theta}cl( \theta;y) \\
&= \left[ \sum_{k^{j} \in Q} \sum_{i=1}^{n-k_{q^{j}}^{j}} \left(- \frac{1}{2} \frac{ \partial}{ \partial \theta_{r}} \log \det \big( \Sigma_{k^{j}}( \theta) \big) - \frac{1}{2}(y_{i}^{k^{j}})^{ \top} \frac{ \partial}{ \partial \theta_{r}} \Sigma_{k^{j}}^{-1}( \theta)y_{i}^{k^{j}} \right) \right]_{r=1}^{p}.
\end{split}
\end{align}
By definition $s_{n}( \hat{ \theta}_{ \text{MCLE}}) = 0$, so from the mean value theorem there exists an interior point, $\bar{ \theta}$, on the line segment in $\mathbb{R}^{p}$ connecting $\theta_{0}$ and $\hat{ \theta}_{ \text{MCLE}}$, such that
\begin{equation*}
0 = s_{n}( \hat{ \theta}_{ \text{MCLE}}) = s_{n}( \theta_{0}) + \frac{ \partial}{ \partial \theta^{ \top}}s_{n}( \bar{ \theta})( \hat{ \theta}_{ \text{MCLE}} - \theta_{0}).
\end{equation*}
After rearranging this expression and multiplying by $\sqrt{n}$, we get:
\begin{equation*}
\sqrt{n}( \hat{ \theta}_{ \text{MCLE}} - \theta_{0}) = - \left( \frac{1}{n} \frac{ \partial}{ \partial \theta^{ \top}}s_{n}( \bar{ \theta}) \right)^{-1} \frac{1}{ \sqrt{n}} s_{n}( \theta_{0}).
\end{equation*}
Since $\hat{ \theta}_{ \text{MCLE}}$ is consistent by Theorem \ref{theorem:lln}, we know from the squeeze theorem for convergence in probability that $\bar{ \theta} \overset{ \mathbb{P}}{ \longrightarrow} \theta_{0}$, as $n \rightarrow \infty$, so by stationarity and ergodicity
\begin{equation*}
- \frac{1}{n} \frac{ \partial}{ \partial \theta^{ \top}}s_{n}( \bar{ \theta}) \overset{ \mathbb{P}}{ \longrightarrow}H( \theta_{0}),
\end{equation*}
as $n \rightarrow \infty$, where
\begin{equation*}
H( \theta_{0}) = - \mathbb{E} \left( \sum_{k^{j} \in Q} \frac{ \partial^{2}}{\partial \theta \partial \theta^{ \top}} \log f_{k^{j}}(Y_{1}^{k^{j}}; \theta) \mid_{ \theta = \theta_{0}} \right)
\end{equation*}
is the negative value of the expected Hessian matrix of $cl( \theta; Y)$.

Hence, the main challenge is to derive an expression for the covariance structure of
\begin{equation*}
\frac{1}{ \sqrt{n}}s_{n}( \theta) = \frac{1}{ \sqrt{n}} \left[ \sum_{k^{j} \in Q} \sum_{i=1}^{n-k_{q^{j}}^{j}} \left(- \frac{1}{2} \frac{ \partial}{ \partial \theta_{r}} \log \det \big( \Sigma_{k^{j}}( \theta) \big) - \frac{1}{2} \frac{ \partial}{ \partial \theta_{r}}(Y_{i}^{k^{j}})^{ \top} \Sigma_{k^{j}}^{-1}( \theta)Y_{i}^{k^{j}} \right) \right]_{r=1}^{p}.
\end{equation*}
Now,
\begin{equation*}
(Y_{i}^{k^{j}})^{ \top} \frac{ \partial}{ \partial \theta_{r}} \Sigma_{k^{j}}^{-1}( \theta)Y_{i}^{k^{j}} = \sum_{j_{1}, j_{2} \in k^{j}} Y_{(i+j_{1})\Delta}Y_{(i+j_{2})\Delta} \frac{ \partial}{ \partial \theta_{r}} ( \Sigma_{k^{j}}^{-1}( \theta))_{j_{1},j_{2}},
\end{equation*}
where $\big( \Sigma_{k^{j}}^{-1}( \theta) \big)_{j_{1},j_{2}}$ is shorthand notation for the entry in $\Sigma_{k^{j}}^{-1}( \theta)$ corresponding to the covariance between $Y_{j_{1}\Delta}$ and $Y_{j_{2}\Delta}$.

This means we can write
\begin{align*}
\frac{1}{ \sqrt{n}} s_n( \theta) = \frac{1}{ \sqrt{n}} \sum_{k^{j} \in Q} \sum_{i=1}^{n-k_{q^{j}}^{j}} \Bigg(&- \frac{1}{2} \left[ \frac{ \partial}{ \partial \theta_{r}} \log \det \big( \Sigma_{k^{j}}( \theta) \big) \right]_{r=1}^{p} \\ &- \frac{1}{2} \sum_{j_{1}, j_{2} \in k^{j}} Y_{(i+j_{1})\Delta} Y_{(i+j_{2})\Delta} \left[ \frac{ \partial}{ \partial \theta_{r}} \big( \Sigma_{k^{j}}^{-1}( \theta) \big)_{j_{1},j_{2}} \right]_{r=1}^{p} \Bigg),
\end{align*}
so we need to calculate
\begin{equation*}
\text{var} \left( \frac{1}{ \sqrt{n}} \sum_{k^{j} \in Q}\sum_{i=1}^{n-k_{q^{j}}} \sum_{j_{1}, j_{2} \in k^{j}}Y_{(i+j_{1})\Delta}Y_{(i+j_{2})\Delta} \right),
\end{equation*}
as the rest are additive or multiplicative constants.

Expanding the variance operator, and ignoring the $\sqrt{n}$ for the time being, yields the following calculation:
\begin{align} \label{equation:covariance}
\begin{split}
\text{var} & \left( \sum_{k^{j} \in Q}\sum_{i=1}^{n-k_{q^{j}}} \sum_{j_{1}, j_{2} \in k^{j}}Y_{(i+j_{1})\Delta}Y_{(i+j_{2})\Delta} \right) \\ &= \sum_{k^{j_{1}}, k^{j_{2}} \in Q} \sum_{i_{1} = 1}^{n-k^{j_{1}}_{q^{j_{1}}}} \sum_{i_{2}=1}^{n-k^{j_{2}}_{q^{j_{2}}}} \sum_{ \iota_{1}, \iota_{2} \in k^{j_{1}}} \sum_{ \iota_{3}, \iota_{4} \in k^{j_{2}}} \mathbb{E} [Y_{(i_{1}+ \iota_{1})\Delta} Y_{(i_{1}+ \iota_{2})\Delta} Y_{(i_{2}+ \iota_{3})\Delta} Y_{(i_{2}+ \iota_{4})\Delta}] \\&- \mathbb{E} [Y_{(i_{1}+ \iota_{1})\Delta} Y_{(i_{1}+ \iota_{2})\Delta}] \mathbb{E}[ Y_{(i_{2}+ \iota_{3})\Delta} Y_{(i_{2}+ \iota_{4})\Delta}] \\
&= \sum_{k^{j_{1}}, k^{j_{2}} \in Q} \sum_{i_{1}=1}^{n-k^{j_{1}}_{q^{j_{1}}}} \sum_{i_{2}=1}^{n-k^{j_{2}}_{q^{j_{2}}}} \sum_{ \iota_{1}, \iota_{2} \in k^{j_{1}}} \sum_{ \iota_{3}, \iota_{4} \in k^{j_{2}}} \big( \mathbb{E}[Y_{(i_{1}+ \iota_{1})\Delta} Y_{(i_{1}+ \iota_{2})\Delta}] \mathbb{E}[Y_{(i_{2}+ \iota_{3})\Delta} Y_{(i_{2} + \iota_{4})\Delta}] \\
&+ \mathbb{E}[Y_{(i_{1}+ \iota_{1})\Delta} Y_{(i_{2}+ \iota_{3})\Delta}] \mathbb{E}[Y_{(i_{1}+ \iota_{2})\Delta} Y_{(i_{2}+ \iota_{4})\Delta}] + \mathbb{E}[Y_{(i_{1}+ \iota_{1})\Delta} Y_{(i_{2}+ \iota_{4})\Delta}] \mathbb{E}[Y_{(i_{1}+ \iota_{2})\Delta} Y_{(i_{2}+ \iota_{3})\Delta}] \big) \\&- \mathbb{E} [Y_{(i_{1}+ \iota_{1})\Delta} Y_{(i_{1}+ \iota_{2})\Delta}] \mathbb{E}[Y_{(i_{2}+ \iota_{3})\Delta} Y_{(i_{2}+ \iota_{4})\Delta}] \\
&= \sum_{k^{j_{1}}, k^{j_{2}} \in Q} \sum_{i_{1}=1}^{n-k^{j_{1}}_{q^{j_{1}}}} \sum_{i_{2}=1}^{n-k^{j_{2}}_{q^{j_{2}}}} \sum_{ \iota_{1}, \iota_{2} \in k^{j_{1}}} \sum_{ \iota_{3}, \iota_{4} \in k^{j_{2}}} \gamma_{(i_{1}+ \iota_{1} - (i_{2}+ \iota_{3}))\Delta} \gamma_{(i_{1}+ \iota_{2} - (i_{2}+ \iota_{4}))\Delta} \\&+ \gamma_{(i_{1}+ \iota_{1} - (i_{2}+ \iota_{4}))\Delta} \gamma_{(i_{1}+ \iota_{2} - (i_{2}+ \iota_{3}))\Delta} \\
&\sim \sum_{k^{j_{1}}, k^{j_{2}} \in Q} \sum_{l=-n}^{n} \sum_{ \iota_{1}, \iota_{2} \in k^{j_{1}}} \sum_{ \iota_{3}, \iota_{4} \in k^{j_{2}}} (n-l) \left( \gamma_{(l+ \iota_{1}- \iota_{3})\Delta} \gamma_{(l+ \iota_{2}- \iota_{4})\Delta} + \gamma_{(l+ \iota_{1} - \iota_{4})\Delta} \gamma_{)l+ \iota_{2} - \iota_{3})\Delta} \right),
\end{split}
\end{align}
where $\gamma_{ \ell}$ is the autocovariance of $Y$ at lag $\ell$, and Isserlis' theorem helps to express higher-order moments of the multivariate normal distribution in terms of its covariance matrix. Thus, convergence of the variance amounts to studying the limiting behavior of the sum of products of autocovariances in \eqref{equation:covariance}.

In the first part of Theorem \ref{theorem:clt}, i.e. the short memory setting and long memory setting with $\beta > 1/2$, the sum converges. Here, we can derive the asymptotic distribution using a Breuer-Major theorem, which was introduced in the one-dimensional case in \citet{breuer-major:83a} and extended to the multivariate setting in \citet{arcones:94a}. From \eqref{equation:score}, we see that the score is a quadratic form of $Y$ and, hence, consists of functionals that are either sums of squares or sums of products that are $k^{j}$ periods apart. Since $\rho(h) = O(h^{- \beta})$ for some $\beta>1/2$, the autocovariance function is square-integrable, so we are done if we can prove that the functions from $\mathbb{R}^{2} \rightarrow \mathbb{R}$ given by
\begin{equation*}
f:(x,y) \mapsto x^{2}+y^{2} \qquad \text{and} \qquad g:(x,y) \mapsto xy
\end{equation*}
are of Hermite rank 2.

We recall that a function $h$ is of Hermite rank $q$ with respect to a Gaussian process $X$ if (a) $\mathbb{E} \big([h(X)-\mathbb{E}(h(X))]p_{m}(X) \big) = 0$ for every polynomial $p_{m}$ of degree $m \leq q-1$ and (b) there exists a polynomial $p_{q}$ of degree $q$ such that $\mathbb{E} \big([h(X)-\mathbb{E}(h(X))]p_{q}(X) \big) \neq 0$. One can always transform the problem to the multivariate standard normal distribution, because the Hermite rank is invariant under linear mappings. To prove (a), note that
\begin{equation} \label{equation:hermite}
\mathbb{E} \bigg[h(Z) \prod_{i=0}^{1}H_{ \alpha_{i}}(Z_{i}) \bigg] = 0,
\end{equation}
where $h \in \{f,g\}$, $\alpha = ( \alpha_{0}, \alpha_{1}) \in \{(1,0), (0,1) \}$, $H_{i}$ is the $i$th Hermite polynomial, and $Z = (Z_{1}, Z_{2})^{ \top} \sim N(0, I_{2})$. \eqref{equation:hermite} follows from independence and the zero skewness of the normal distribution, since $H_{0}(x) = 1$ and $H_{1}(x) = x$. The claim in (b) follows, because the expectation is nonzero with a second-degree Hermite polynomial, $H_{2}(x) = x^{2}-1$. So the score converges to a Gaussian distribution, and by Slutsky's theorem so does $\hat{ \theta}_{ \text{MCLE}}$:
\begin{equation*}
\frac{1}{ \sqrt{n}}s_{n}( \theta_{0}) \overset{d}{ \longrightarrow}N(0, V_{0}),
\end{equation*}
where
\begin{align*}
V_{0} &= \frac{1}{4} \sum_{k^{j_{1}}, k^{j_{2}} \in Q} \sum_{l=- \infty}^{ \infty} \sum_{ \iota_{1}, \iota_{2} \in k^{j_{1}}} \sum_{ \iota_{3}, \iota_{4} \in k^{j_{2}}} \left[ \frac{ \partial}{ \partial \theta_{r}} \big( \Sigma_{k^{j_{1}}}^{-1}( \theta) \big)_{ \iota_{1}, \iota_{2}} \vert_{ \theta = \theta_{0}} \right]_{r=1}^{p} \\
&\times \left( \left[ \frac{ \partial}{ \partial \theta_{r}} \big( \Sigma_{k^{j_{2}}}^{-1}( \theta) \big)_{ \iota_{3}, \iota_{4}} \vert_{ \theta = \theta_{0}} \right]_{r=1}^{p} \right)^{ \top} \left( \gamma_{(l+ \iota_{1}- \iota_{3})\Delta} \gamma_{(l+ \iota_{2}- \iota_{4})\Delta}+ \gamma_{(l+ \iota_{1}- \iota_{4})\Delta} \gamma_{(l+ \iota_{2}- \iota_{3})\Delta} \right),
\end{align*}
see Lemma \ref{lemma:convergence} below.

Next, take $\beta=1/2$. Here, we know from the derivations in \eqref{equation:covariance} that the variance of the score, for $n$ large, essentially behaves as
\begin{align*}
\text{var} & \left( \sum_{k^{j} \in Q}\sum_{i=1}^{n-k_{q^{j}}} \sum_{j_{1}, j_{2} \in k^{j}}Y_{(i+j_{1})\Delta}Y_{(i+j_{2})\Delta)\Delta} \right) \\
&\sim \sum_{k^{j_{1}}, k^{j_{2}} \in Q} \sum_{l=-n}^{n} \sum_{ \iota_{1}, \iota_{2} \in k^{j_{1}}} \sum_{ \iota_{3}, \iota_{4} \in k^{j_{2}}} (n-l) \left( \gamma_{(l+ \iota_{1}- \iota_{3})\Delta} \gamma_{(l+ \iota_{2}- \iota_{4})\Delta} + \gamma_{(l+ \iota_{1} - \iota_{4})\Delta} \gamma_{(l+ \iota_{2} - \iota_{3})\Delta} \right) \\
&\sim n \text{var}(Y) + \sum_{k^{j_{1}}, k^{j_{2}} \in Q} \sum_{l=-n}^{n} \sum_{ \iota_{1}, \iota_{2} \in k^{j_{1}}} \sum_{ \iota_{3}, \iota_{4} \in k^{j_{2}}} \mathbf{1}_{ \mathcal{S}}(l, \iota_{1}, \iota_{2}, \iota_{3}, \iota_{4})(n-l) \\
&\times \left( \frac{L_{ \infty}((l+ \iota_{1}- \iota_{3})\Delta)}{ \sqrt{|l+ \iota_{1}- \iota_{3}|\Delta}} \frac{L_{ \infty}((l+ \iota_{2}- \iota_{4})\Delta)}{ \sqrt{|l+ \iota_{2}- \iota_{4}|\Delta}} + \frac{L_{ \infty}((l+ \iota_{1}- \iota_{4})\Delta)}{ \sqrt{|l+ \iota_{1}- \iota_{4}|\Delta}} \frac{L_{ \infty}((l+ \iota_{2} - \iota_3)\Delta)}{ \sqrt{|l+ \iota_{2}- \iota_{3}|\Delta}} \right) \\
&\sim n \text{var}(Y) + \sum_{l=1}^{n} (n-l) L_{ \infty}^{2}(l\Delta)(l\Delta)^{-1},
\end{align*}
where $\mathcal{S} = \{l+ \iota_{1}- \iota_{3} \neq 0, l+ \iota_{2}- \iota_{4} \neq 0, l+ \iota_{1}- \iota_{4} \neq 0, l+ \iota_{2}- \iota_{3} \neq0 \}$.

We conclude from this and Lemma \ref{lemma:convergence} that $\text{var}(n^{-1/2}s_{n}( \theta)) = O(L_{ \gamma}(n))$, where $L_{ \gamma}(n) = \sum_{l=1}^{n} L_{ \infty}^{2}(l\Delta)(l\Delta)^{-1}$. By formula 1.5.8 and Proposition 1.5.9a in \citet{bingham-goldie-teugels:89a}, $L_{ \gamma}$ is slowly varying and growing faster than $L_{ \infty}^{2}$, so the correct scaling of $s_{n}( \theta)$ is $n^{-1/2}L_{ \gamma}^{-1/2}(n)$.

To derive the limiting distribution, we rely on Theorem 4 in \cite{hosking:96a}, which shows the rather striking result that for $\beta \leq 1/2$, $\hat{ \gamma}_{ \ell} - \hat{ \gamma}_{ \ell'} \overset{ \mathbb{P}}{ \longrightarrow} 0$ for $\ell \neq \ell'$. Hence, the sample autocovariances have a common asymptotic distribution at any lag, which is Gaussian for $\beta = 1/2$, and their multivariate normal distribution function is degenerate. This implies that the entries in the score also exhibit this property, as they are a weighted sum of the sample autocovariances. Consequently,
\begin{equation*}
n^{-1/2}L_{ \gamma}^{-1/2}(n)( \hat{ \theta}_{ \text{MCLE}}- \theta_{0}) \overset{d}{ \longrightarrow} \sum_{k^{j} \in Q} \sum_{j_{1}, j_{2} \in k^{j}} \left[ \frac{ \partial}{ \partial \theta_{r}} \big( \Sigma_{k}^{-1}( \theta) \big)_{j_{1},j_{2}} \vert_{ \theta = \theta_{0}} \right]_{r=1}^{p} N(0,1).
\end{equation*}
At last, we handle the very long memory setting with $\beta \in (0,1/2)$. First, we write
\begin{equation*}
n^{ \beta} L_{2}^{-1/2}(n)( \hat{ \theta}_{ \text{MCLE}}- \theta_{0}) = - \left( \frac{1}{n} \frac{ \partial}{ \partial \theta^{ \top}}s_{n}( \bar{ \theta}) \right)^{-1}n^{ \beta-1}L_{2}^{-1/2}(n)s_{n}( \theta_{0}).
\end{equation*}
By a weak law of large numbers, as $n \rightarrow \infty$,
\begin{equation*}
- \frac{1}{n} \frac{ \partial}{ \partial \theta^{ \top}}s_{n}( \bar{ \theta}) \overset{ \mathbb{P}}{ \longrightarrow}H( \theta_{0}).
\end{equation*}
Looking at the last factor in the above expression, we note that the rate normalization is exactly as required for sample covariances of totally dependent Hermite-Rosenblatt processes to converge in law, see, e.g., Section 4.4.1.3 of \citet{beran-feng-ghosh-kulik:13a} (in their notation, $\beta = 1-2d$):
\begin{equation*}
n^{ \beta-1}L_{2}^{-1/2}(n)s_{n}( \theta_{0}) \overset{d}{ \longrightarrow} \frac{1}{2} \sum_{k^{j} \in Q} \sum_{j_{1}, j_{2} \in k^{j}} \left[ \frac{ \partial}{ \partial \theta_{r}} \big( \Sigma_{k}^{-1}( \theta) \big)_{j_{1},j_{2}} \vert_{ \theta = \theta_{0}} \right]_{r=1}^{p}Z_{2,H}(1).
\end{equation*}
\end{proof}

\begin{lemma} \label{lemma:convergence}
Let $\{x_{k} \}_{k=1}^{ \infty}$ be a sequence of positive numbers. Then, it holds that
\begin{equation*}
X_{n} = \sum_{k=1}^{n} x_{k} \text{ converges } \Leftrightarrow \tilde{X}_{n} = \sum_{k=1}^{n} \frac{n-k}{n}x_{k} \text{ converges.}
\end{equation*}
Furthermore, if they converge, the limit is identical.
\end{lemma}

\begin{proof} As the terms in the sequence are positive, both $X_{n}$ and $\tilde{X}_{n}$ are increasing. Hence, convergence is equivalent to boundedness by monotone convergence theorem. The implication $\Rightarrow$ follows immediately, since $\tilde{X}_{n} \leq X_{n}$.

To show the $\Leftarrow$ part, assume that $\tilde{X}_{n}$ converges. Then it is bounded by some number $C$. Suppose that $X_{n}$ does not converge, which means that it is unbounded, so there is an $N_{1}$ such that $X_{N_{1}} > 2C$. Moreover, for any $N_{2} > 2N_{1}$, the first $N_{1}$ terms in $\tilde{X}_{N_{2}}$ are at least half as large as the corresponding ones in $X_{N_{2}}$, since $1-N_{1}/N_{2} > 1/2$. This implies that $\tilde{X}_{N_{1}} > \frac{1}{2} X_{N_{1}} > C$, which is a contradiction.

To show the limit is identical, suppose that $X_{n} \rightarrow x_{0}$ as $n \rightarrow \infty$. Then, by definition, for any $\epsilon > 0$ there exists $N_{1} \in \mathbb{N}$ such that for $n \geq N_{1}$:
\begin{equation*}
x_{0} - X_{n} < \epsilon,
\end{equation*}
Pick some $N_{2} > N_{1}$ such that $\frac{N_{2}-N_{1}}{N_{2}} > 1- \epsilon$. Then, we observe that
\begin{align*}
0 \leq x_{0} - \tilde{X}_{N_{2}} &= x_{0} - \sum_{k=1}^{N_{2}} \frac{N_{2}-k}{N_{2}}x_{k} \\ &\leq x_{0} - \sum_{k=1}^{N_{1}} \frac{N_{2}-k}{N_{2}}x_{k} \\
&\leq x_{0} - \frac{N_{2}-N_{1}}{N_{2}} \sum_{k=1}^{N_{1}}x_{k} \\
&= x_{0} - \frac{N_{2}-N_{1}}{N_{2}} X_{N_{1}} \\
&\leq x_{0} - \frac{N_{2}-N_{1}}{N_{2}}(x_{0}- \epsilon) \\
&\leq x_{0} - (1- \epsilon)(x_{0}- \epsilon) = \epsilon (1+x_{0}- \epsilon).
\end{align*}
As $\epsilon (1+x_{0}- \epsilon)$ can be made arbitrarily small,  $\tilde{X}_{n} \rightarrow x_{0}$.
\end{proof}

\clearpage

\section{Method of moments-based estimator} \label{appendix:mm}

In this appendix, we review the method of moments-based estimator of the parameters of the fOU process and the Cauchy class, which serve as a benchmark for our MCLE approach. Our description is based on \citet{wang-xiao-yu:23a} and \citet{bennedsen-lunde-pakkanen:22a}. It is not too detailed, but the reader can find further information and references in their papers.

The estimator of the mean, $\mu$, is the sample average,
\begin{equation*}\hat{ \mu} = \frac{1}{n} \sum_{i=1}^{n} Y_{i \Delta}.
\end{equation*}
The roughness index, $\alpha$, is estimated by a change-of-frequency (COF) approach \citep[see, e.g.][]{lang-roueff:01a,barndorff-nielsen-corcuera-podolskij:13a}. To write this down succinctly, we need a notion of the $k$th-order difference, for any $k \in \mathbb{N}$, of a time series observed with time gap $\Delta$ and sampled at frequency $\eta$, where $\eta \in \mathbb{N}$. At stage $i \geq \eta k$, this can be defined as
\begin{equation*}
(1-L^{ \eta})^{k} Y_{i \Delta} \equiv \sum_{j=0}^{k} (-1)^{j} \binom{k}{j} Y_{(i- \eta j) \Delta},
\end{equation*}
where $L$ is the lag operator.

We now introduce the $p$th-order realized power variation:
\begin{equation*}
V(Y,p,k, \eta; \Delta)_{t} = \sum_{i = \eta k}^{[t/ \Delta]} \vert (1-L^{ \eta})^{k} Y_{i \Delta} \vert^{p},
\end{equation*}
for $p>0$.

The COF estimator of $\alpha$ is then
\begin{equation*}
\hat{ \alpha} = \frac{ \log_2( \text{COF}(Y,p, \Delta)_{t})}{p} - \frac{1}{2}, \quad \text{where }
\text{COF}(Y,p, \Delta)_{t} = \frac{V(Y,p,2,2; \Delta)_{t}}{V(Y,p,2,1; \Delta)_{t}}, \quad \text{for } t > 0.
\end{equation*}
$\hat{ \alpha}$ is a consistent estimator of $\alpha$ in the infill limit as $\Delta \rightarrow 0$. We follow \citet{wang-xiao-yu:23a} and implement $\hat{ \alpha}$ with $p = 2$. We remark that this estimator is a special case of the one from \citet{bennedsen-lunde-pakkanen:22a}.

In the fOU, the standard deviation, $\nu$, and mean reversion, $\kappa$, are estimated as:
\begin{equation*}
\hat{ \nu} = \sqrt{ \frac{ \sum_{i=3}^{n}(Y_{i \Delta} - 2Y_{(i-1) \Delta} + Y_{(i-2) \Delta})^{2}}{n \big(4-2^{2 \hat{H}} \big) \Delta^{2 \hat{H}}}}
\quad \text{and} \quad \hat{ \kappa} = \left( \frac{n \sum_{i=1}^{n} Y_{i \Delta}^{2}- \left( \sum_{i=1}^{n}Y_{i \Delta} \right)^{2}}{n^{2} \hat{ \nu}^{2} \hat{H} \Gamma(2 \hat{H})} \right)^{-1/(2 \hat{H})},
\end{equation*}
where $\hat{H} = \hat{ \alpha}+1/2$.

In the Cauchy class, $\nu$ is recovered directly as the sample standard deviation, $\hat{ \nu} = n^{-1} \sum_{i=1}^{n} (Y_{i \Delta} - \hat{ \mu})^{2}$. $\beta$ is estimated by matching the theoretical ACF from \eqref{equation:cauchy-acf} to the empirical one with the first-stage estimator of $\alpha$ plugged in, leaving it as a function of $\beta$.

\clearpage

\section{Additional Monte Carlo analysis} \label{appendix:simulation}

In this appendix, we provide some supplemental simulation evidence to highlight the performance of our MCLE approach under alternative configurations:

\begin{itemize}
\item Table \ref{table:sim-ou-q=2-N=12-mu=0} -- Table \ref{table:sim-cc-q=2-N=12-mu=1} report the pairwise setting ($q = 2$ and $N = 12$).
\item Table \ref{table:sim-ou-q=3-N=1-mu=0} -- Table \ref{table:sim-cc-q=3-N=1-mu=1} consider the setting with daily data ($q = 3$ and $N = 1$).

\end{itemize}

\noindent The tables are presented without comment.

\begin{sidewaystable}[ht!]
\setlength{ \tabcolsep}{0.15cm}
\begin{center}
\caption{Parameter estimation of the fOU process [$q = 2$ and $\mu$ known].}
\label{table:sim-ou-q=2-N=12-mu=0}
\begin{footnotesize}
\begin{tabular}{lrrrrrrrrrrrrrr}
\hline \hline
Parameter & Value & \multicolumn{3}{c}{$T = 1{,}095$} && \multicolumn{3}{c}{$T = 1{,}825$} && \multicolumn{3}{c}{$T = 2{,}555$} \\
\cline{3-5} \cline{7-9} \cline{11-13}
&& \multicolumn{1}{c}{MCLE} & \multicolumn{1}{c}{MME} & \multicolumn{1}{c}{RMSE$_{r}$} && \multicolumn{1}{c}{MCLE} & \multicolumn{1}{c}{MME} & \multicolumn{1}{c}{RMSE$_{r}$} && \multicolumn{1}{c}{MCLE} & \multicolumn{1}{c}{MME} & \multicolumn{1}{c}{RMSE$_{r}$} \\
\hline 
Panel A: \\
$\mu$ & 0.000 & & & & & & & & & & & \\
$\kappa$ & 0.005 & 0.000 (0.005) & 0.006 (0.020) & 0.240 & & -0.000 (0.004) & 0.004 (0.013) & 0.273 & & -0.000 (0.003) & 0.003 (0.010) & 0.297 \\
$\nu$ & 1.250 & 0.000 (0.009) & 0.000 (0.041) & 0.214 & & 0.000 (0.007) & 0.000 (0.032) & 0.212 & & 0.000 (0.006) & -0.000 (0.027) & 0.211 \\
$\alpha$ & -0.450 & 0.000 (0.003) & -0.000 (0.013) & 0.200 & & 0.000 (0.002) & -0.000 (0.010) & 0.199 & & 0.000 (0.002) & -0.000 (0.009) & 0.197 \\
Panel B: \\
$\mu$ & 0.000 & & & & & & & & & & & \\
$\kappa$ & 0.010 & 0.001 (0.006) & 0.004 (0.014) & 0.442 & & 0.001 (0.005) & 0.002 (0.010) & 0.467 & & 0.001 (0.004) & 0.002 (0.008) & 0.486 \\
$\nu$ & 0.750 & 0.000 (0.006) & 0.000 (0.024) & 0.260 & & 0.000 (0.005) & 0.000 (0.019) & 0.257 & & -0.000 (0.004) & 0.000 (0.016) & 0.256 \\
$\alpha$ & -0.400 & 0.000 (0.004) & -0.000 (0.013) & 0.304 & & -0.000 (0.003) & -0.000 (0.010) & 0.303 & & 0.000 (0.003) & -0.000 (0.009) & 0.300 \\
Panel C: \\
$\mu$ & 0.000 & & & & & & & & & & & \\
$\kappa$ & 0.015 & 0.002 (0.006) & 0.002 (0.008) & 0.806 & & 0.001 (0.005) & 0.001 (0.006) & 0.804 & & 0.001 (0.004) & 0.001 (0.005) & 0.811 \\
$\nu$ & 0.500 & 0.000 (0.007) & 0.000 (0.016) & 0.436 & & -0.000 (0.006) & -0.000 (0.013) & 0.430 & & -0.000 (0.005) & -0.000 (0.011) & 0.427 \\
$\alpha$ & -0.200 & 0.000 (0.007) & -0.000 (0.012) & 0.589 & & -0.000 (0.006) & -0.000 (0.010) & 0.583 & & -0.000 (0.005) & -0.000 (0.008) & 0.578 \\
Panel D: \\
$\mu$ & 0.000 & & & & & & & & & & & \\
$\kappa$ & 0.035 & 0.002 (0.009) & 0.002 (0.010) & 0.923 & & 0.001 (0.007) & 0.001 (0.008) & 0.913 & & 0.001 (0.006) & 0.001 (0.007) & 0.915 \\
$\nu$ & 0.300 & 0.000 (0.006) & 0.000 (0.010) & 0.602 & & 0.000 (0.005) & -0.000 (0.008) & 0.599 & & 0.000 (0.004) & -0.000 (0.007) & 0.593 \\
$\alpha$ & 0.000 & 0.000 (0.010) & -0.000 (0.012) & 0.838 & & 0.000 (0.008) & -0.000 (0.009) & 0.836 & & 0.000 (0.006) & -0.000 (0.008) & 0.827 \\
Panel E: \\
$\mu$ & 0.000 & & & & & & & & & & & \\
$\kappa$ & 0.070 & 0.005 (0.019) & 0.003 (0.016) & 1.202 & & 0.003 (0.014) & 0.002 (0.012) & 1.183 & & 0.002 (0.012) & 0.001 (0.010) & 1.182 \\
$\nu$ & 0.200 & 0.001 (0.008) & 0.000 (0.008) & 0.998 & & 0.001 (0.006) & -0.000 (0.006) & 0.986 & & 0.000 (0.005) & -0.000 (0.005) & 0.976 \\
$\alpha$ & 0.200 & 0.001 (0.015) & -0.000 (0.011) & 1.362 & & 0.001 (0.011) & -0.000 (0.008) & 1.359 & & 0.001 (0.010) & -0.000 (0.007) & 1.348 \\
\hline \hline
\end{tabular}
\end{footnotesize}
\smallskip
\begin{scriptsize}
\parbox{0.98\textwidth}{\emph{Note.} 
We simulate the process in the caption of the table 10,000 times on the interval $[0,T]$, where $T$ is interpreted as the number of days. 
There are $N = 12$ observations per unit interval corresponding to twelve every day.
The true value of the parameter vector appears to the left in Panel A -- E. 
We estimate $\theta$ with the maximum composite likelihod estimation (MCLE) procedure developed in the main text, and benchmark it against a method-of-moments estimator (MME).
The table shows the Monte Carlo average of each parameter estimate across simulations (standard deviation in parenthesis).
The last column reports the RMSE ratio defined as $\text{RMSE}_{r} = \text{RMSE}( \hat{ \theta}_{ \text{MCLE}}) / \text{RMSE}( \hat{ \theta}_{ \text{MME}})$. 
}
\end{scriptsize}
\end{center}
\end{sidewaystable}

\begin{sidewaystable}[ht!]
\setlength{ \tabcolsep}{0.15cm}
\begin{center}
\caption{Parameter estimation of the fOU process [$q = 2$ and $\mu$ estimated].}
\label{table:sim-ou-q=2-N=12-mu=1}
\begin{footnotesize}
\begin{tabular}{lrrrrrrrrrrrrrr}
\hline \hline
Parameter & Value & \multicolumn{3}{c}{$T = 1{,}095$} && \multicolumn{3}{c}{$T = 1{,}825$} && \multicolumn{3}{c}{$T = 2{,}555$} \\
\cline{3-5} \cline{7-9} \cline{11-13}
&& \multicolumn{1}{c}{MCLE} & \multicolumn{1}{c}{MME} & \multicolumn{1}{c}{RMSE$_{r}$} && \multicolumn{1}{c}{MCLE} & \multicolumn{1}{c}{MME} & \multicolumn{1}{c}{RMSE$_{r}$} && \multicolumn{1}{c}{MCLE} & \multicolumn{1}{c}{MME} & \multicolumn{1}{c}{RMSE$_{r}$} \\
\hline 
Panel A: \\
$\mu$ & 0.000 & -0.000 (0.257) & -0.000 (0.258) & 0.999 & & 0.000 (0.164) & 0.000 (0.164) & 0.999 & & 0.000 (0.120) & 0.000 (0.120) & 0.998 \\
$\kappa$ & 0.005 & 0.001 (0.006) & 0.008 (0.023) & 0.242 & & 0.001 (0.004) & 0.004 (0.014) & 0.273 & & 0.000 (0.003) & 0.003 (0.011) & 0.296 \\
$\nu$ & 1.250 & 0.000 (0.009) & 0.000 (0.041) & 0.216 & & 0.000 (0.007) & 0.000 (0.032) & 0.212 & & 0.000 (0.006) & 0.000 (0.027) & 0.210 \\
$\alpha$ & -0.450 & 0.000 (0.003) & -0.000 (0.013) & 0.202 & & 0.000 (0.002) & -0.000 (0.010) & 0.200 & & 0.000 (0.002) & -0.000 (0.009) & 0.197 \\
Panel B: \\
$\mu$ & 0.000 & -0.000 (0.094) & -0.000 (0.094) & 1.000 & & 0.000 (0.062) & 0.000 (0.062) & 1.000 & & 0.000 (0.047) & 0.000 (0.047) & 0.999 \\
$\kappa$ & 0.010 & 0.002 (0.007) & 0.005 (0.015) & 0.448 & & 0.001 (0.005) & 0.003 (0.010) & 0.470 & & 0.001 (0.004) & 0.002 (0.008) & 0.490 \\
$\nu$ & 0.750 & 0.000 (0.006) & 0.000 (0.024) & 0.262 & & 0.000 (0.005) & 0.000 (0.019) & 0.258 & & 0.000 (0.004) & -0.000 (0.016) & 0.256 \\
$\alpha$ & -0.400 & 0.000 (0.004) & -0.000 (0.013) & 0.305 & & -0.000 (0.003) & -0.000 (0.010) & 0.303 & & 0.000 (0.003) & -0.000 (0.009) & 0.300 \\
Panel C: \\
$\mu$ & 0.000 & 0.000 (0.098) & 0.000 (0.098) & 1.001 & & 0.000 (0.070) & 0.000 (0.070) & 1.001 & & 0.000 (0.056) & 0.000 (0.056) & 1.000 \\
$\kappa$ & 0.015 & 0.003 (0.007) & 0.003 (0.009) & 0.812 & & 0.002 (0.005) & 0.002 (0.006) & 0.813 & & 0.001 (0.004) & 0.001 (0.005) & 0.819 \\
$\nu$ & 0.500 & 0.000 (0.007) & 0.000 (0.016) & 0.437 & & 0.000 (0.006) & 0.000 (0.013) & 0.433 & & 0.000 (0.005) & -0.000 (0.011) & 0.430 \\
$\alpha$ & -0.200 & 0.000 (0.007) & -0.000 (0.012) & 0.590 & & 0.000 (0.006) & -0.000 (0.010) & 0.586 & & 0.000 (0.005) & -0.000 (0.008) & 0.582 \\
Panel D: \\
$\mu$ & 0.000 & 0.000 (0.068) & 0.000 (0.068) & 1.001 & & 0.000 (0.053) & 0.000 (0.053) & 1.000 & & 0.000 (0.045) & 0.000 (0.045) & 1.000 \\
$\kappa$ & 0.035 & 0.004 (0.010) & 0.004 (0.011) & 0.929 & & 0.003 (0.007) & 0.002 (0.008) & 0.923 & & 0.002 (0.006) & 0.002 (0.007) & 0.923 \\
$\nu$ & 0.300 & 0.001 (0.006) & 0.000 (0.010) & 0.607 & & 0.000 (0.005) & -0.000 (0.008) & 0.604 & & 0.000 (0.004) & -0.000 (0.007) & 0.596 \\
$\alpha$ & 0.000 & 0.001 (0.010) & -0.000 (0.012) & 0.845 & & 0.000 (0.008) & -0.000 (0.009) & 0.841 & & 0.000 (0.006) & -0.000 (0.008) & 0.830 \\
Panel E: \\
$\mu$ & 0.000 & 0.000 (0.069) & 0.000 (0.069) & 1.000 & & 0.000 (0.059) & 0.000 (0.059) & 1.000 & & 0.000 (0.054) & 0.000 (0.054) & 1.000 \\
$\kappa$ & 0.070 & 0.013 (0.019) & 0.009 (0.016) & 1.238 & & 0.009 (0.014) & 0.006 (0.012) & 1.213 & & 0.007 (0.012) & 0.005 (0.010) & 1.210 \\
$\nu$ & 0.200 & 0.003 (0.008) & 0.000 (0.008) & 1.045 & & 0.002 (0.006) & -0.000 (0.006) & 1.019 & & 0.002 (0.005) & -0.000 (0.005) & 1.002 \\
$\alpha$ & 0.200 & 0.005 (0.015) & -0.000 (0.011) & 1.388 & & 0.003 (0.011) & -0.000 (0.008) & 1.376 & & 0.003 (0.009) & -0.000 (0.007) & 1.361 \\
\hline \hline
\end{tabular}
\end{footnotesize}
\smallskip
\begin{scriptsize}
\parbox{0.98\textwidth}{\emph{Note.} 
We simulate the process in the caption of the table 10,000 times on the interval $[0,T]$, where $T$ is interpreted as the number of days. 
There are $N = 12$ observations per unit interval corresponding to twelve every day.
The true value of the parameter vector appears to the left in Panel A -- E. 
We estimate $\theta$ with the maximum composite likelihod estimation (MCLE) procedure developed in the main text, and benchmark it against a method-of-moments estimator (MME).
The table shows the Monte Carlo average of each parameter estimate across simulations (standard deviation in parenthesis).
The last column reports the RMSE ratio defined as $\text{RMSE}_{r} = \text{RMSE}( \hat{ \theta}_{ \text{MCLE}}) / \text{RMSE}( \hat{ \theta}_{ \text{MME}})$. 
}
\end{scriptsize}
\end{center}
\end{sidewaystable}

\begin{sidewaystable}[ht!]
\setlength{ \tabcolsep}{0.15cm}
\begin{center}
\caption{Parameter estimation of the Cauchy class [$q = 2$ and $\mu$ known].}
\label{table:sim-cc-q=2-N=12-mu=0}
\begin{footnotesize}
\begin{tabular}{lrrrrrrrrrrrrrr}
\hline \hline
Parameter & Value & \multicolumn{3}{c}{$T = 1{,}095$} && \multicolumn{3}{c}{$T = 1{,}825$} && \multicolumn{3}{c}{$T = 2{,}555$} \\
\cline{3-5} \cline{7-9} \cline{11-13}
&& \multicolumn{1}{c}{MCLE} & \multicolumn{1}{c}{MME} & \multicolumn{1}{c}{RMSE$_{r}$} && \multicolumn{1}{c}{MCLE} & \multicolumn{1}{c}{MME} & \multicolumn{1}{c}{RMSE$_{r}$} && \multicolumn{1}{c}{MCLE} & \multicolumn{1}{c}{MME} & \multicolumn{1}{c}{RMSE$_{r}$} \\
\hline 
Panel A: \\
$\mu$ & 0.000 & & & & & & & & & & & \\
$\beta$ & 0.250 & 0.069 (0.138) & -0.078 (0.098) & 1.299 & & 0.058 (0.120) & -0.096 (0.082) & 1.196 & & 0.052 (0.112) & -0.105 (0.072) & 1.137 \\
$\nu$ & 1.250 & -0.001 (0.072) & -0.001 (0.072) & 1.000 & & -0.001 (0.067) & -0.001 (0.067) & 1.000 & & -0.001 (0.063) & -0.001 (0.063) & 1.000 \\
$\alpha$ & -0.450 & 0.006 (0.014) & -0.031 (0.010) & 0.607 & & 0.005 (0.012) & -0.032 (0.008) & 0.528 & & 0.004 (0.011) & -0.032 (0.007) & 0.489 \\
Panel B: \\
$\mu$ & 0.000 & & & & & & & & & & & \\
$\beta$ & 0.500 & 0.032 (0.127) & -0.147 (0.121) & 0.810 & & 0.023 (0.106) & -0.159 (0.102) & 0.704 & & 0.019 (0.094) & -0.167 (0.091) & 0.639 \\
$\nu$ & 0.750 & -0.000 (0.021) & -0.000 (0.021) & 1.000 & & -0.000 (0.018) & -0.000 (0.018) & 1.000 & & -0.000 (0.016) & -0.000 (0.016) & 1.000 \\
$\alpha$ & -0.400 & 0.002 (0.013) & -0.059 (0.013) & 0.290 & & 0.002 (0.010) & -0.059 (0.010) & 0.244 & & 0.001 (0.009) & -0.059 (0.009) & 0.218 \\
Panel C: \\
$\mu$ & 0.000 & & & & & & & & & & & \\
$\beta$ & 0.750 & 0.014 (0.102) & 0.023 (0.138) & 0.740 & & 0.008 (0.081) & 0.004 (0.123) & 0.658 & & 0.006 (0.069) & -0.005 (0.111) & 0.622 \\
$\nu$ & 0.500 & -0.000 (0.016) & -0.000 (0.016) & 1.002 & & -0.000 (0.013) & -0.000 (0.013) & 1.002 & & -0.000 (0.011) & -0.000 (0.011) & 1.002 \\
$\alpha$ & -0.200 & 0.001 (0.011) & -0.060 (0.013) & 0.248 & & 0.000 (0.009) & -0.060 (0.010) & 0.199 & & 0.000 (0.007) & -0.060 (0.008) & 0.170 \\
Panel D: \\
$\mu$ & 0.000 & & & & & & & & & & & \\
$\beta$ & 1.000 & 0.011 (0.103) & 0.025 (0.158) & 0.650 & & 0.006 (0.081) & 0.012 (0.132) & 0.609 & & 0.004 (0.068) & 0.006 (0.116) & 0.591 \\
$\nu$ & 0.300 & -0.000 (0.009) & -0.000 (0.009) & 1.002 & & -0.000 (0.007) & -0.000 (0.007) & 1.002 & & -0.000 (0.006) & -0.000 (0.006) & 1.002 \\
$\alpha$ & 0.000 & 0.001 (0.010) & -0.019 (0.012) & 0.587 & & 0.000 (0.008) & -0.018 (0.009) & 0.507 & & 0.000 (0.007) & -0.018 (0.008) & 0.450 \\
Panel E: \\
$\mu$ & 0.000 & & & & & & & & & & & \\
$\beta$ & 1.250 & 0.010 (0.113) & 0.040 (0.184) & 0.606 & & 0.006 (0.087) & 0.030 (0.146) & 0.592 & & 0.004 (0.074) & 0.026 (0.125) & 0.586 \\
$\nu$ & 0.200 & -0.000 (0.006) & -0.000 (0.006) & 1.002 & & -0.000 (0.005) & -0.000 (0.005) & 1.002 & & -0.000 (0.004) & -0.000 (0.004) & 1.002 \\
$\alpha$ & 0.200 & 0.000 (0.010) & 0.041 (0.011) & 0.326 & & 0.000 (0.008) & 0.041 (0.008) & 0.259 & & 0.000 (0.007) & 0.041 (0.007) & 0.222 \\
\hline \hline
\end{tabular}
\end{footnotesize}
\smallskip
\begin{scriptsize}
\parbox{0.98\textwidth}{\emph{Note.} 
We simulate the process in the caption of the table 10,000 times on the interval $[0,T]$, where $T$ is interpreted as the number of days. 
There are $N = 12$ observations per unit interval corresponding to twelve every day.
The true value of the parameter vector appears to the left in Panel A -- E. 
We estimate $\theta$ with the maximum composite likelihod estimation (MCLE) procedure developed in the main text, and benchmark it against a method-of-moments estimator (MME).
The table shows the Monte Carlo average of each parameter estimate across simulations (standard deviation in parenthesis).
The last column reports the RMSE ratio defined as $\text{RMSE}_{r} = \text{RMSE}( \hat{ \theta}_{ \text{MCLE}}) / \text{RMSE}( \hat{ \theta}_{ \text{MME}})$. 
}
\end{scriptsize}
\end{center}
\end{sidewaystable}

\begin{sidewaystable}[ht!]
\setlength{ \tabcolsep}{0.15cm}
\begin{center}
\caption{Parameter estimation of the Cauchy class [$q = 2$ and $\mu$ estimated].}
\label{table:sim-cc-q=2-N=12-mu=1}
\begin{footnotesize}
\begin{tabular}{lrrrrrrrrrrrrrr}
\hline \hline
Parameter & Value & \multicolumn{3}{c}{$T = 1{,}095$} && \multicolumn{3}{c}{$T = 1{,}825$} && \multicolumn{3}{c}{$T = 2{,}555$} \\
\cline{3-5} \cline{7-9} \cline{11-13}
&& \multicolumn{1}{c}{MCLE} & \multicolumn{1}{c}{MME} & \multicolumn{1}{c}{RMSE$_{r}$} && \multicolumn{1}{c}{MCLE} & \multicolumn{1}{c}{MME} & \multicolumn{1}{c}{RMSE$_{r}$} && \multicolumn{1}{c}{MCLE} & \multicolumn{1}{c}{MME} & \multicolumn{1}{c}{RMSE$_{r}$} \\
\hline 
Panel A: \\
$\mu$ & 0.000 & 0.001 (0.360) & 0.001 (0.360) & 1.000 & & 0.000 (0.345) & 0.000 (0.345) & 1.000 & & 0.000 (0.336) & 0.000 (0.336) & 1.000 \\
$\beta$ & 0.250 & 0.218 (0.107) & 0.038 (0.109) & 1.668 & & 0.187 (0.086) & -0.005 (0.093) & 1.688 & & 0.170 (0.075) & -0.031 (0.081) & 1.684 \\
$\nu$ & 1.250 & -0.052 (0.016) & -0.052 (0.016) & 1.000 & & -0.048 (0.015) & -0.048 (0.015) & 1.000 & & -0.046 (0.014) & -0.046 (0.014) & 1.000 \\
$\alpha$ & -0.450 & 0.019 (0.011) & -0.031 (0.010) & 0.725 & & 0.017 (0.009) & -0.032 (0.008) & 0.614 & & 0.015 (0.007) & -0.032 (0.007) & 0.552 \\
Panel B: \\
$\mu$ & 0.000 & 0.001 (0.139) & 0.001 (0.139) & 1.000 & & 0.001 (0.127) & 0.001 (0.127) & 1.000 & & 0.001 (0.119) & 0.001 (0.119) & 1.000 \\
$\beta$ & 0.500 & 0.123 (0.104) & -0.033 (0.107) & 1.237 & & 0.097 (0.084) & -0.066 (0.090) & 1.066 & & 0.083 (0.074) & -0.089 (0.081) & 0.920 \\
$\nu$ & 0.750 & -0.013 (0.011) & -0.013 (0.011) & 1.000 & & -0.011 (0.010) & -0.011 (0.010) & 1.000 & & -0.009 (0.009) & -0.009 (0.009) & 1.000 \\
$\alpha$ & -0.400 & 0.011 (0.010) & -0.059 (0.013) & 0.289 & & 0.008 (0.008) & -0.059 (0.010) & 0.236 & & 0.007 (0.007) & -0.059 (0.009) & 0.207 \\
Panel C: \\
$\mu$ & 0.000 & 0.001 (0.077) & 0.001 (0.077) & 1.001 & & 0.001 (0.065) & 0.001 (0.065) & 1.000 & & 0.000 (0.058) & 0.000 (0.058) & 1.000 \\
$\beta$ & 0.750 & 0.054 (0.097) & 0.139 (0.085) & 0.802 & & 0.037 (0.077) & 0.096 (0.085) & 0.746 & & 0.028 (0.066) & 0.070 (0.083) & 0.710 \\
$\nu$ & 0.500 & -0.006 (0.014) & -0.006 (0.014) & 1.001 & & -0.004 (0.011) & -0.004 (0.011) & 1.001 & & -0.003 (0.010) & -0.003 (0.010) & 1.001 \\
$\alpha$ & -0.200 & 0.005 (0.010) & -0.060 (0.013) & 0.244 & & 0.003 (0.008) & -0.060 (0.010) & 0.195 & & 0.002 (0.007) & -0.060 (0.008) & 0.167 \\
Panel D: \\
$\mu$ & 0.000 & 0.000 (0.032) & 0.000 (0.032) & 1.001 & & 0.000 (0.026) & 0.000 (0.026) & 1.000 & & 0.000 (0.022) & 0.000 (0.022) & 1.000 \\
$\beta$ & 1.000 & 0.033 (0.102) & 0.116 (0.120) & 0.723 & & 0.020 (0.080) & 0.077 (0.106) & 0.681 & & 0.014 (0.068) & 0.056 (0.097) & 0.654 \\
$\nu$ & 0.300 & -0.002 (0.009) & -0.002 (0.009) & 1.002 & & -0.001 (0.007) & -0.001 (0.007) & 1.002 & & -0.001 (0.006) & -0.001 (0.006) & 1.002 \\
$\alpha$ & 0.000 & 0.002 (0.010) & -0.019 (0.012) & 0.581 & & 0.001 (0.008) & -0.018 (0.009) & 0.500 & & 0.001 (0.007) & -0.018 (0.008) & 0.446 \\
Panel E: \\
$\mu$ & 0.000 & 0.000 (0.016) & 0.000 (0.016) & 1.001 & & 0.000 (0.013) & 0.000 (0.013) & 1.000 & & 0.000 (0.011) & 0.000 (0.011) & 1.000 \\
$\beta$ & 1.250 & 0.025 (0.113) & 0.108 (0.157) & 0.657 & & 0.015 (0.088) & 0.074 (0.130) & 0.629 & & 0.010 (0.074) & 0.059 (0.115) & 0.613 \\
$\nu$ & 0.200 & -0.001 (0.006) & -0.001 (0.006) & 1.002 & & -0.001 (0.005) & -0.000 (0.005) & 1.002 & & -0.000 (0.004) & -0.000 (0.004) & 1.002 \\
$\alpha$ & 0.200 & 0.001 (0.010) & 0.041 (0.011) & 0.325 & & 0.001 (0.008) & 0.041 (0.008) & 0.259 & & 0.000 (0.007) & 0.041 (0.007) & 0.222 \\
\hline \hline
\end{tabular}
\end{footnotesize}
\smallskip
\begin{scriptsize}
\parbox{0.98\textwidth}{\emph{Note.} 
We simulate the process in the caption of the table 10,000 times on the interval $[0,T]$, where $T$ is interpreted as the number of days. 
There are $N = 12$ observations per unit interval corresponding to twelve every day.
The true value of the parameter vector appears to the left in Panel A -- E. 
We estimate $\theta$ with the maximum composite likelihod estimation (MCLE) procedure developed in the main text, and benchmark it against a method-of-moments estimator (MME).
The table shows the Monte Carlo average of each parameter estimate across simulations (standard deviation in parenthesis).
The last column reports the RMSE ratio defined as $\text{RMSE}_{r} = \text{RMSE}( \hat{ \theta}_{ \text{MCLE}}) / \text{RMSE}( \hat{ \theta}_{ \text{MME}})$. 
}
\end{scriptsize}
\end{center}
\end{sidewaystable}

\begin{sidewaystable}[ht!]
\setlength{ \tabcolsep}{0.15cm}
\begin{center}
\caption{Parameter estimation of the fOU process [$q = 3$ and $\mu$ known].}
\label{table:sim-ou-q=3-N=1-mu=0}
\begin{footnotesize}
\begin{tabular}{lrrrrrrrrrrrrrr}
\hline \hline
Parameter & Value & \multicolumn{3}{c}{$T = 1{,}095$} && \multicolumn{3}{c}{$T = 1{,}825$} && \multicolumn{3}{c}{$T = 2{,}555$} \\
\cline{3-5} \cline{7-9} \cline{11-13}
&& \multicolumn{1}{c}{MCLE} & \multicolumn{1}{c}{MME} & \multicolumn{1}{c}{RMSE$_{r}$} && \multicolumn{1}{c}{MCLE} & \multicolumn{1}{c}{MME} & \multicolumn{1}{c}{RMSE$_{r}$} && \multicolumn{1}{c}{MCLE} & \multicolumn{1}{c}{MME} & \multicolumn{1}{c}{RMSE$_{r}$} \\
\hline 
Panel A: \\
$\mu$ & 0.000 & & & & & & & & & & & \\
$\kappa$ & 0.005 & 0.000 (0.006) & 0.014 (0.033) & 0.158 & & 0.000 (0.004) & 0.010 (0.024) & 0.152 & & 0.000 (0.003) & 0.008 (0.020) & 0.148 \\
$\nu$ & 1.250 & -0.005 (0.033) & -0.000 (0.032) & 1.036 & & -0.003 (0.026) & -0.000 (0.025) & 1.039 & & -0.003 (0.022) & -0.001 (0.021) & 1.041 \\
$\alpha$ & -0.450 & 0.002 (0.010) & 0.004 (0.038) & 0.270 & & 0.001 (0.008) & 0.002 (0.032) & 0.250 & & 0.001 (0.007) & 0.001 (0.028) & 0.237 \\
Panel B: \\
$\mu$ & 0.000 & & & & & & & & & & & \\
$\kappa$ & 0.010 & 0.002 (0.007) & 0.009 (0.025) & 0.278 & & 0.001 (0.005) & 0.006 (0.018) & 0.268 & & 0.001 (0.004) & 0.004 (0.015) & 0.265 \\
$\nu$ & 0.750 & -0.002 (0.021) & -0.001 (0.019) & 1.141 & & -0.001 (0.017) & -0.001 (0.015) & 1.141 & & -0.001 (0.014) & -0.001 (0.012) & 1.141 \\
$\alpha$ & -0.400 & 0.002 (0.016) & -0.001 (0.045) & 0.361 & & 0.001 (0.012) & -0.000 (0.035) & 0.347 & & 0.001 (0.010) & -0.000 (0.030) & 0.345 \\
Panel C: \\
$\mu$ & 0.000 & & & & & & & & & & & \\
$\kappa$ & 0.015 & 0.002 (0.008) & 0.003 (0.013) & 0.645 & & 0.001 (0.006) & 0.002 (0.009) & 0.632 & & 0.001 (0.005) & 0.001 (0.008) & 0.633 \\
$\nu$ & 0.500 & -0.000 (0.016) & -0.001 (0.012) & 1.308 & & -0.000 (0.012) & -0.000 (0.009) & 1.304 & & -0.000 (0.010) & -0.000 (0.008) & 1.300 \\
$\alpha$ & -0.200 & 0.002 (0.031) & -0.002 (0.044) & 0.705 & & 0.001 (0.024) & -0.001 (0.034) & 0.698 & & 0.001 (0.020) & -0.000 (0.029) & 0.703 \\
Panel D: \\
$\mu$ & 0.000 & & & & & & & & & & & \\
$\kappa$ & 0.035 & 0.004 (0.016) & 0.002 (0.015) & 1.093 & & 0.002 (0.011) & 0.001 (0.011) & 0.948 & & 0.001 (0.009) & 0.001 (0.010) & 0.941 \\
$\nu$ & 0.300 & 0.002 (0.009) & -0.000 (0.009) & 1.035 & & 0.001 (0.006) & -0.000 (0.007) & 0.920 & & 0.001 (0.005) & -0.000 (0.006) & 0.902 \\
$\alpha$ & 0.000 & -0.000 (0.041) & -0.002 (0.041) & 1.002 & & -0.001 (0.031) & -0.002 (0.032) & 0.977 & & -0.000 (0.026) & -0.001 (0.027) & 0.975 \\
Panel E: \\
$\mu$ & 0.000 & & & & & & & & & & & \\
$\kappa$ & 0.070 & 0.011 (0.039) & 0.002 (0.022) & 1.841 & & 0.008 (0.031) & 0.000 (0.017) & 1.907 & & 0.006 (0.026) & -0.001 (0.014) & 1.894 \\
$\nu$ & 0.200 & 0.010 (0.022) & -0.001 (0.010) & 2.414 & & 0.007 (0.017) & -0.001 (0.007) & 2.385 & & 0.005 (0.014) & -0.001 (0.006) & 2.282 \\
$\alpha$ & 0.200 & -0.003 (0.073) & -0.006 (0.038) & 1.900 & & 0.001 (0.062) & -0.005 (0.030) & 2.068 & & 0.002 (0.054) & -0.005 (0.025) & 2.140 \\
\hline \hline
\end{tabular}
\end{footnotesize}
\smallskip
\begin{scriptsize}
\parbox{0.98\textwidth}{\emph{Note.} 
We simulate the process in the caption of the table 10,000 times on the interval $[0,T]$, where $T$ is interpreted as the number of days. 
There is $N = 1$ observation per unit interval, corresponding to one every day.
The true value of the parameter vector appears to the left in Panel A -- E. 
We estimate $\theta$ with the maximum composite likelihod estimation (MCLE) procedure developed in the main text, and benchmark it against a method-of-moments estimator (MME).
The table shows the Monte Carlo average of each parameter estimate across simulations (standard deviation in parenthesis).
The last column reports the RMSE ratio defined as $\text{RMSE}_{r} = \text{RMSE}( \hat{ \theta}_{ \text{MCLE}}) / \text{RMSE}( \hat{ \theta}_{ \text{MME}})$. 
}
\end{scriptsize}
\end{center}
\end{sidewaystable}

\begin{sidewaystable}[ht!]
\setlength{ \tabcolsep}{0.15cm}
\begin{center}
\caption{Parameter estimation of the fOU process [$q = 3$ and $\mu$ estimated].}
\label{table:sim-ou-q=3-N=1-mu=1}
\begin{footnotesize}
\begin{tabular}{lrrrrrrrrrrrrrr}
\hline \hline
Parameter & Value & \multicolumn{3}{c}{$T = 1{,}095$} && \multicolumn{3}{c}{$T = 1{,}825$} && \multicolumn{3}{c}{$T = 2{,}555$} \\
\cline{3-5} \cline{7-9} \cline{11-13}
&& \multicolumn{1}{c}{MCLE} & \multicolumn{1}{c}{MME} & \multicolumn{1}{c}{RMSE$_{r}$} && \multicolumn{1}{c}{MCLE} & \multicolumn{1}{c}{MME} & \multicolumn{1}{c}{RMSE$_{r}$} && \multicolumn{1}{c}{MCLE} & \multicolumn{1}{c}{MME} & \multicolumn{1}{c}{RMSE$_{r}$} \\
\hline 
Panel A: \\
$\mu$ & 0.000 & 0.001 (0.252) & 0.000 (0.259) & 0.975 & & -0.001 (0.158) & -0.001 (0.164) & 0.966 & & -0.001 (0.116) & -0.001 (0.120) & 0.964 \\
$\kappa$ & 0.005 & 0.002 (0.007) & 0.017 (0.037) & 0.179 & & 0.001 (0.004) & 0.011 (0.026) & 0.158 & & 0.001 (0.003) & 0.009 (0.021) & 0.153 \\
$\nu$ & 1.250 & -0.006 (0.033) & -0.000 (0.032) & 1.032 & & -0.004 (0.026) & -0.000 (0.025) & 1.037 & & -0.003 (0.022) & -0.001 (0.021) & 1.036 \\
$\alpha$ & -0.450 & 0.003 (0.010) & 0.004 (0.038) & 0.275 & & 0.002 (0.008) & 0.002 (0.032) & 0.251 & & 0.001 (0.007) & 0.001 (0.028) & 0.237 \\
Panel B: \\
$\mu$ & 0.000 & 0.001 (0.093) & 0.001 (0.094) & 0.992 & & -0.001 (0.061) & -0.001 (0.062) & 0.985 & & -0.000 (0.046) & -0.000 (0.047) & 0.983 \\
$\kappa$ & 0.010 & 0.003 (0.008) & 0.010 (0.025) & 0.300 & & 0.002 (0.005) & 0.006 (0.018) & 0.276 & & 0.001 (0.004) & 0.005 (0.015) & 0.271 \\
$\nu$ & 0.750 & -0.003 (0.021) & -0.001 (0.019) & 1.137 & & -0.002 (0.017) & -0.001 (0.015) & 1.140 & & -0.001 (0.014) & -0.001 (0.012) & 1.142 \\
$\alpha$ & -0.400 & 0.003 (0.016) & -0.001 (0.045) & 0.363 & & 0.002 (0.012) & -0.001 (0.035) & 0.348 & & 0.001 (0.010) & -0.000 (0.030) & 0.346 \\
Panel C: \\
$\mu$ & 0.000 & 0.003 (0.098) & 0.002 (0.097) & 1.015 & & 0.000 (0.070) & 0.000 (0.069) & 1.006 & & 0.000 (0.055) & 0.000 (0.055) & 1.004 \\
$\kappa$ & 0.015 & 0.004 (0.009) & 0.004 (0.013) & 0.692 & & 0.002 (0.006) & 0.002 (0.009) & 0.658 & & 0.002 (0.005) & 0.002 (0.008) & 0.650 \\
$\nu$ & 0.500 & -0.001 (0.015) & -0.001 (0.012) & 1.261 & & -0.001 (0.012) & -0.000 (0.009) & 1.286 & & -0.000 (0.010) & -0.000 (0.008) & 1.295 \\
$\alpha$ & -0.200 & 0.006 (0.030) & -0.002 (0.044) & 0.702 & & 0.003 (0.024) & -0.001 (0.034) & 0.698 & & 0.002 (0.020) & -0.000 (0.029) & 0.704 \\
Panel D: \\
$\mu$ & 0.000 & 0.001 (0.068) & 0.001 (0.067) & 1.014 & & 0.000 (0.053) & 0.000 (0.052) & 1.008 & & 0.000 (0.044) & 0.000 (0.044) & 1.006 \\
$\kappa$ & 0.035 & 0.007 (0.015) & 0.004 (0.016) & 1.006 & & 0.004 (0.011) & 0.002 (0.012) & 0.950 & & 0.003 (0.009) & 0.002 (0.010) & 0.945 \\
$\nu$ & 0.300 & 0.003 (0.008) & -0.000 (0.009) & 1.002 & & 0.002 (0.006) & -0.000 (0.007) & 0.933 & & 0.001 (0.005) & -0.000 (0.006) & 0.910 \\
$\alpha$ & 0.000 & 0.004 (0.038) & -0.002 (0.041) & 0.919 & & 0.002 (0.030) & -0.002 (0.032) & 0.933 & & 0.002 (0.025) & -0.001 (0.027) & 0.947 \\
Panel E: \\
$\mu$ & 0.000 & 0.001 (0.069) & 0.001 (0.068) & 1.010 & & 0.000 (0.059) & 0.000 (0.059) & 1.006 & & 0.000 (0.053) & 0.000 (0.053) & 1.005 \\
$\kappa$ & 0.070 & 0.006 (0.026) & 0.007 (0.022) & 1.150 & & 0.005 (0.021) & 0.004 (0.017) & 1.267 & & 0.004 (0.018) & 0.003 (0.014) & 1.315 \\
$\nu$ & 0.200 & 0.009 (0.013) & -0.001 (0.010) & 1.522 & & 0.006 (0.011) & -0.001 (0.007) & 1.565 & & 0.005 (0.009) & -0.001 (0.006) & 1.574 \\
$\alpha$ & 0.200 & -0.025 (0.049) & -0.006 (0.038) & 1.350 & & -0.017 (0.043) & -0.005 (0.030) & 1.478 & & -0.013 (0.039) & -0.005 (0.025) & 1.579 \\
\hline \hline
\end{tabular}
\end{footnotesize}
\smallskip
\begin{scriptsize}
\parbox{0.98\textwidth}{\emph{Note.} 
We simulate the process in the caption of the table 10,000 times on the interval $[0,T]$, where $T$ is interpreted as the number of days. 
There is $N = 1$ observation per unit interval, corresponding to one every day.
The true value of the parameter vector appears to the left in Panel A -- E. 
We estimate $\theta$ with the maximum composite likelihod estimation (MCLE) procedure developed in the main text, and benchmark it against a method-of-moments estimator (MME).
The table shows the Monte Carlo average of each parameter estimate across simulations (standard deviation in parenthesis).
The last column reports the RMSE ratio defined as $\text{RMSE}_{r} = \text{RMSE}( \hat{ \theta}_{ \text{MCLE}}) / \text{RMSE}( \hat{ \theta}_{ \text{MME}})$. 
}
\end{scriptsize}
\end{center}
\end{sidewaystable}

\begin{sidewaystable}[ht!]
\setlength{ \tabcolsep}{0.15cm}
\begin{center}
\caption{Parameter estimation of the Cauchy class [$q = 3$ and $\mu$ known].}
\label{table:sim-cc-q=3-N=1-mu=0}
\begin{footnotesize}
\begin{tabular}{lrrrrrrrrrrrrrr}
\hline \hline
Parameter & Value & \multicolumn{3}{c}{$T = 1{,}095$} && \multicolumn{3}{c}{$T = 1{,}825$} && \multicolumn{3}{c}{$T = 2{,}555$} \\
\cline{3-5} \cline{7-9} \cline{11-13}
&& \multicolumn{1}{c}{MCLE} & \multicolumn{1}{c}{MME} & \multicolumn{1}{c}{RMSE$_{r}$} && \multicolumn{1}{c}{MCLE} & \multicolumn{1}{c}{MME} & \multicolumn{1}{c}{RMSE$_{r}$} && \multicolumn{1}{c}{MCLE} & \multicolumn{1}{c}{MME} & \multicolumn{1}{c}{RMSE$_{r}$} \\
\hline 
Panel A: \\
$\mu$ & 0.000 & & & & & & & & & & & \\
$\beta$ & 0.250 & 0.160 (0.374) & 0.051 (0.182) & 2.104 & & 0.120 (0.247) & 0.016 (0.154) & 1.690 & & 0.101 (0.203) & -0.005 (0.140) & 1.534 \\
$\nu$ & 1.250 & -0.002 (0.075) & -0.001 (0.074) & 1.003 & & -0.001 (0.068) & -0.001 (0.068) & 1.002 & & -0.001 (0.064) & -0.001 (0.064) & 1.002 \\
$\alpha$ & -0.450 & 0.022 (0.055) & -0.020 (0.028) & 1.832 & & 0.016 (0.036) & -0.024 (0.022) & 1.359 & & 0.013 (0.029) & -0.026 (0.019) & 1.164 \\
Panel B: \\
$\mu$ & 0.000 & & & & & & & & & & & \\
$\beta$ & 0.500 & 0.112 (0.341) & -0.124 (0.204) & 1.576 & & 0.074 (0.232) & -0.149 (0.178) & 1.147 & & 0.056 (0.193) & -0.163 (0.165) & 0.980 \\
$\nu$ & 0.750 & -0.000 (0.026) & -0.000 (0.026) & 1.006 & & -0.000 (0.021) & -0.000 (0.021) & 1.004 & & -0.000 (0.018) & -0.000 (0.018) & 1.003 \\
$\alpha$ & -0.400 & 0.018 (0.059) & -0.063 (0.032) & 1.098 & & 0.012 (0.040) & -0.067 (0.026) & 0.764 & & 0.009 (0.033) & -0.068 (0.023) & 0.635 \\
Panel C: \\
$\mu$ & 0.000 & & & & & & & & & & & \\
$\beta$ & 0.750 & 0.045 (0.196) & -0.057 (0.122) & 1.541 & & 0.028 (0.151) & -0.062 (0.097) & 1.431 & & 0.021 (0.128) & -0.066 (0.088) & 1.295 \\
$\nu$ & 0.500 & -0.000 (0.018) & -0.000 (0.018) & 1.007 & & -0.000 (0.014) & -0.000 (0.014) & 1.005 & & -0.000 (0.012) & -0.000 (0.012) & 1.004 \\
$\alpha$ & -0.200 & 0.013 (0.065) & -0.170 (0.045) & 0.513 & & 0.008 (0.050) & -0.169 (0.035) & 0.406 & & 0.006 (0.043) & -0.168 (0.030) & 0.350 \\
Panel D: \\
$\mu$ & 0.000 & & & & & & & & & & & \\
$\beta$ & 1.000 & 0.030 (0.197) & -0.118 (0.106) & 1.476 & & 0.019 (0.150) & -0.126 (0.088) & 1.206 & & 0.013 (0.127) & -0.131 (0.079) & 1.048 \\
$\nu$ & 0.300 & -0.000 (0.010) & -0.000 (0.010) & 1.007 & & -0.000 (0.008) & -0.000 (0.008) & 1.005 & & -0.000 (0.006) & -0.000 (0.006) & 1.003 \\
$\alpha$ & 0.000 & 0.010 (0.082) & -0.259 (0.044) & 0.439 & & 0.006 (0.063) & -0.258 (0.034) & 0.341 & & 0.004 (0.054) & -0.258 (0.029) & 0.292 \\
Panel E: \\
$\mu$ & 0.000 & & & & & & & & & & & \\
$\beta$ & 1.250 & 0.026 (0.215) & -0.196 (0.115) & 1.198 & & 0.016 (0.165) & -0.202 (0.092) & 0.972 & & 0.012 (0.140) & -0.204 (0.081) & 0.846 \\
$\nu$ & 0.200 & -0.000 (0.006) & -0.000 (0.006) & 1.006 & & -0.000 (0.005) & -0.000 (0.005) & 1.004 & & -0.000 (0.004) & -0.000 (0.004) & 1.003 \\
$\alpha$ & 0.200 & 0.008 (0.097) & -0.349 (0.042) & 0.386 & & 0.005 (0.074) & -0.348 (0.033) & 0.300 & & 0.004 (0.063) & -0.348 (0.028) & 0.256 \\
\hline \hline
\end{tabular}
\end{footnotesize}
\smallskip
\begin{scriptsize}
\parbox{0.98\textwidth}{\emph{Note.} 
We simulate the process in the caption of the table 10,000 times on the interval $[0,T]$, where $T$ is interpreted as the number of days. 
There is $N = 1$ observation per unit interval, corresponding to one every day.
The true value of the parameter vector appears to the left in Panel A -- E. 
We estimate $\theta$ with the maximum composite likelihod estimation (MCLE) procedure developed in the main text, and benchmark it against a method-of-moments estimator (MME).
The table shows the Monte Carlo average of each parameter estimate across simulations (standard deviation in parenthesis).
The last column reports the RMSE ratio defined as $\text{RMSE}_{r} = \text{RMSE}( \hat{ \theta}_{ \text{MCLE}}) / \text{RMSE}( \hat{ \theta}_{ \text{MME}})$. 
}
\end{scriptsize}
\end{center}
\end{sidewaystable}

\begin{sidewaystable}[ht!]
\setlength{ \tabcolsep}{0.15cm}
\begin{center}
\caption{Parameter estimation of the Cauchy class [$q = 3$ and $\mu$ estimated].}
\label{table:sim-cc-q=3-N=1-mu=1}
\begin{footnotesize}
\begin{tabular}{lrrrrrrrrrrrrrr}
\hline \hline
Parameter & Value & \multicolumn{3}{c}{$T = 1{,}095$} && \multicolumn{3}{c}{$T = 1{,}825$} && \multicolumn{3}{c}{$T = 2{,}555$} \\
\cline{3-5} \cline{7-9} \cline{11-13}
&& \multicolumn{1}{c}{MCLE} & \multicolumn{1}{c}{MME} & \multicolumn{1}{c}{RMSE$_{r}$} && \multicolumn{1}{c}{MCLE} & \multicolumn{1}{c}{MME} & \multicolumn{1}{c}{RMSE$_{r}$} && \multicolumn{1}{c}{MCLE} & \multicolumn{1}{c}{MME} & \multicolumn{1}{c}{RMSE$_{r}$} \\
\hline 
Panel A: \\
$\mu$ & 0.000 & 0.005 (0.361) & 0.005 (0.360) & 1.002 & & 0.005 (0.346) & 0.005 (0.345) & 1.001 & & 0.005 (0.336) & 0.005 (0.336) & 1.001 \\
$\beta$ & 0.250 & 0.484 (0.588) & 0.205 (0.197) & 2.785 & & 0.368 (0.301) & 0.153 (0.164) & 2.023 & & 0.316 (0.218) & 0.126 (0.150) & 1.790 \\
$\nu$ & 1.250 & -0.053 (0.029) & -0.053 (0.029) & 1.004 & & -0.049 (0.024) & -0.049 (0.024) & 1.003 & & -0.046 (0.021) & -0.046 (0.021) & 1.002 \\
$\alpha$ & -0.450 & 0.061 (0.087) & -0.020 (0.028) & 3.101 & & 0.047 (0.045) & -0.024 (0.022) & 2.015 & & 0.040 (0.033) & -0.026 (0.019) & 1.641 \\
Panel B: \\
$\mu$ & 0.000 & 0.002 (0.141) & 0.002 (0.140) & 1.005 & & 0.002 (0.127) & 0.002 (0.127) & 1.003 & & 0.002 (0.120) & 0.002 (0.119) & 1.002 \\
$\beta$ & 0.500 & 0.324 (0.410) & -0.011 (0.205) & 2.292 & & 0.230 (0.244) & -0.046 (0.176) & 1.635 & & 0.186 (0.193) & -0.066 (0.163) & 1.377 \\
$\nu$ & 0.750 & -0.013 (0.019) & -0.013 (0.018) & 1.005 & & -0.011 (0.015) & -0.011 (0.015) & 1.004 & & -0.010 (0.013) & -0.010 (0.013) & 1.003 \\
$\alpha$ & -0.400 & 0.049 (0.071) & -0.063 (0.032) & 1.434 & & 0.035 (0.043) & -0.067 (0.026) & 0.926 & & 0.028 (0.034) & -0.068 (0.023) & 0.741 \\
Panel C: \\
$\mu$ & 0.000 & 0.001 (0.078) & 0.001 (0.077) & 1.009 & & 0.001 (0.065) & 0.001 (0.065) & 1.005 & & 0.001 (0.058) & 0.001 (0.058) & 1.004 \\
$\beta$ & 0.750 & 0.138 (0.195) & 0.026 (0.108) & 1.998 & & 0.092 (0.148) & 0.011 (0.077) & 2.097 & & 0.071 (0.125) & 0.002 (0.064) & 2.112 \\
$\nu$ & 0.500 & -0.006 (0.016) & -0.006 (0.016) & 1.005 & & -0.004 (0.013) & -0.004 (0.013) & 1.005 & & -0.004 (0.011) & -0.003 (0.011) & 1.003 \\
$\alpha$ & -0.200 & 0.039 (0.064) & -0.169 (0.045) & 0.544 & & 0.026 (0.049) & -0.169 (0.035) & 0.420 & & 0.020 (0.041) & -0.168 (0.030) & 0.357 \\
Panel D: \\
$\mu$ & 0.000 & -0.000 (0.032) & -0.000 (0.032) & 1.008 & & 0.000 (0.026) & -0.000 (0.026) & 1.006 & & 0.000 (0.022) & 0.000 (0.022) & 1.004 \\
$\beta$ & 1.000 & 0.090 (0.199) & -0.051 (0.084) & 2.273 & & 0.057 (0.151) & -0.073 (0.068) & 1.835 & & 0.042 (0.127) & -0.084 (0.060) & 1.544 \\
$\nu$ & 0.300 & -0.002 (0.010) & -0.002 (0.009) & 1.006 & & -0.001 (0.007) & -0.001 (0.007) & 1.005 & & -0.001 (0.006) & -0.001 (0.006) & 1.003 \\
$\alpha$ & 0.000 & 0.031 (0.081) & -0.259 (0.044) & 0.445 & & 0.019 (0.062) & -0.258 (0.034) & 0.344 & & 0.014 (0.053) & -0.258 (0.029) & 0.294 \\
Panel E: \\
$\mu$ & 0.000 & 0.000 (0.016) & 0.000 (0.016) & 1.008 & & 0.000 (0.013) & 0.000 (0.013) & 1.005 & & 0.000 (0.011) & 0.000 (0.011) & 1.004 \\
$\beta$ & 1.250 & 0.069 (0.215) & -0.142 (0.097) & 1.576 & & 0.042 (0.164) & -0.163 (0.078) & 1.201 & & 0.031 (0.139) & -0.173 (0.068) & 1.007 \\
$\nu$ & 0.200 & -0.001 (0.006) & -0.001 (0.006) & 1.005 & & -0.000 (0.005) & -0.000 (0.005) & 1.005 & & -0.000 (0.004) & -0.000 (0.004) & 1.003 \\
$\alpha$ & 0.200 & 0.024 (0.095) & -0.349 (0.042) & 0.385 & & 0.014 (0.074) & -0.348 (0.033) & 0.299 & & 0.010 (0.063) & -0.348 (0.028) & 0.256 \\
\hline \hline
\end{tabular}
\end{footnotesize}
\smallskip
\begin{scriptsize}
\parbox{0.98\textwidth}{\emph{Note.} 
We simulate the process in the caption of the table 10,000 times on the interval $[0,T]$, where $T$ is interpreted as the number of days. 
There is $N = 1$ observation per unit interval, corresponding to one every day.
The true value of the parameter vector appears to the left in Panel A -- E. 
We estimate $\theta$ with the maximum composite likelihod estimation (MCLE) procedure developed in the main text, and benchmark it against a method-of-moments estimator (MME).
The table shows the Monte Carlo average of each parameter estimate across simulations (standard deviation in parenthesis).
The last column reports the RMSE ratio defined as $\text{RMSE}_{r} = \text{RMSE}( \hat{ \theta}_{ \text{MCLE}}) / \text{RMSE}( \hat{ \theta}_{ \text{MME}})$. 
}
\end{scriptsize}
\end{center}
\end{sidewaystable}

\clearpage

\section{The impact of market microstructure noise} \label{appendix:noise}

\begin{figure}[ht!]
\begin{center}
\caption{Volatility signature plot.}
\label{figure:volatility-signature}
\begin{tabular}{cc}
\small{Panel A: ETH.} & \small{Panel B: SOL.} \\
\includegraphics[height=8.00cm,width=0.48\textwidth]{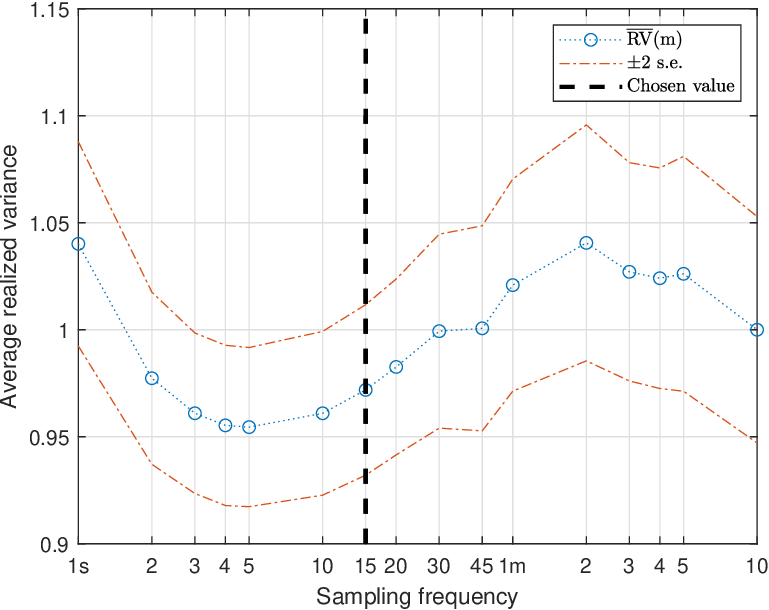} &
\includegraphics[height=8.00cm,width=0.48\textwidth]{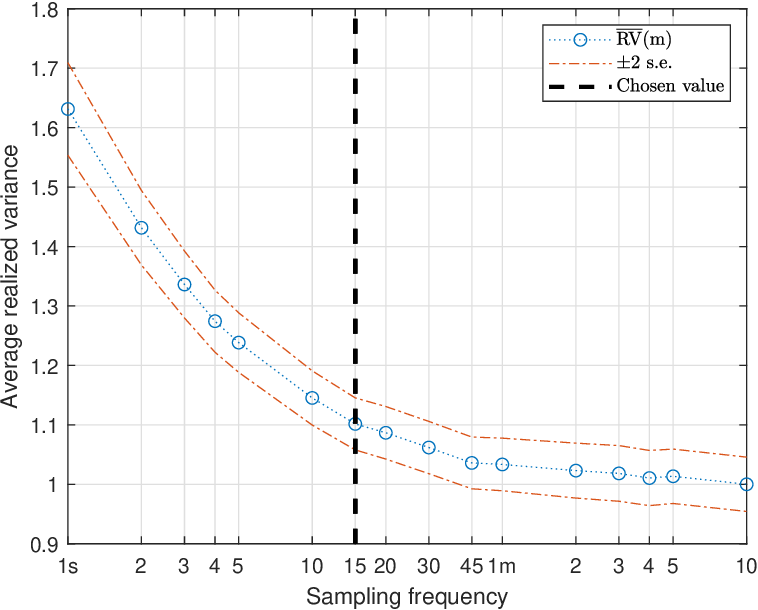} \\
\small{Panel C: XRP.} & \small{Panel D: BNB.} \\
\includegraphics[height=8.00cm,width=0.48\textwidth]{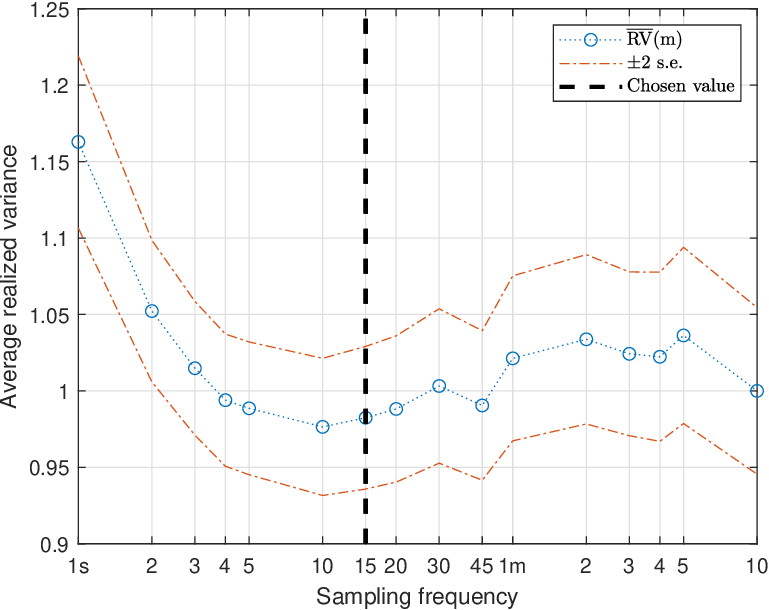} &
\includegraphics[height=8.00cm,width=0.48\textwidth]{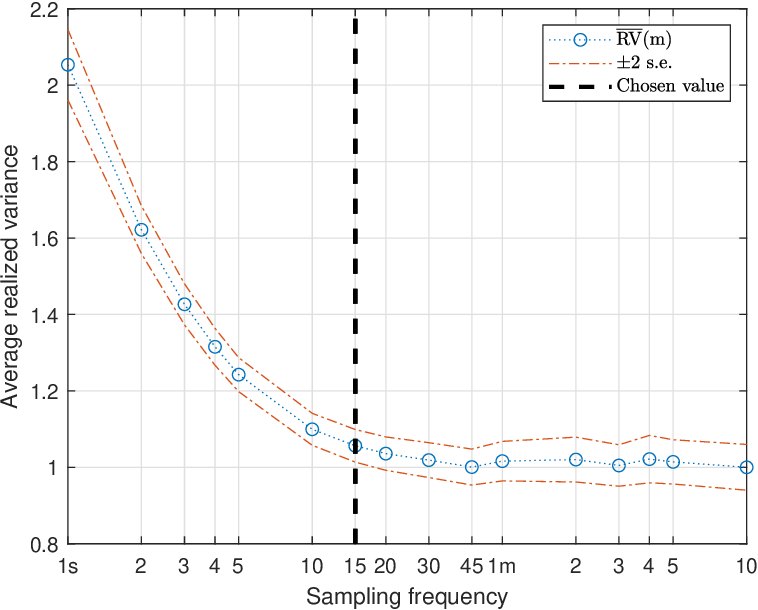} \\
\end{tabular}
\begin{scriptsize}
\parbox{\textwidth}{\emph{Note.} We construct a volatility signature plot. As a function of the sampling frequency $m$, we plot the sample average daily RV, $\overline{RV}(m) = T^{-1} \sum_{t=1}^{T} RV_{t}^{ \Delta}$, where $RV_{t}^{ \Delta}$ is the RV on day $t$ sampled with a time gap $\Delta = 1/m$. The latter ranges from a 10-minute (slowest) to a one-second (fastest) interval. Each $\overline{RV}(m)$ observation is scaled by the 10-minute estimate as a normalization. We also show a two standard error band centered around each point estimate.}
\end{scriptsize}
\end{center}
\end{figure}

\clearpage

\small
\bibliographystyle{rfs}
\bibliography{userref}

\end{document}